\documentclass[10pt, aps, english, prx, superscriptaddress, twocolumn,
notitlepage, nofootinbib, longbibliography, 
floatfix]{revtex4-2}

\usepackage{graphicx}
\usepackage{dcolumn}
\usepackage{bm}
\usepackage{bbm}
\usepackage{amsthm} 
\usepackage{diagbox}
\usepackage{braket}
\usepackage{amsmath}
\usepackage{enumerate}
\usepackage{subfigure}
\usepackage{epstopdf}

\usepackage[colorlinks, linkcolor=NavyBlue, anchorcolor=NavyBlue, citecolor=blue]{hyperref}

\usepackage{amssymb}
\usepackage{framed}
\usepackage{tabularx}
\usepackage[normalem]{ulem}
\usepackage{booktabs}
\usepackage[dvipsnames]{xcolor}
\usepackage{array}
\usepackage{soul}
\usepackage{tikz}
\usepackage{pgfplots}
\pgfplotsset{compat=1.15}
\hypersetup{colorlinks=true}
\usetikzlibrary{positioning}
\usetikzlibrary{decorations,arrows}
\usetikzlibrary{decorations.pathreplacing}
\usetikzlibrary{calc, shapes.geometric}
\usetikzlibrary{chains, scopes, positioning, backgrounds, shapes, fit, shadows, calc, arrows.meta}
\tikzset{>=stealth}
\usepackage{color}
\usepackage{wasysym}
\usepackage{mathtools}
\usepackage{comment}
\usepackage{bbold}

\usetikzlibrary{through,hobby}
\usetikzlibrary{decorations.markings}
\usetikzlibrary{fadings,shadings}
\usetikzlibrary{plotmarks}
\definecolor{darkblue}{rgb}{0.,0.,0.4}
\definecolor{darkred}{rgb}{0.5,0.,0.}
\definecolor{BlueViolet}{RGB}{138,43,226}
\definecolor{SkyBlue}{RGB}{30,144,255}
\definecolor{DarkGreen}{RGB}{0,150,0}
\definecolor{DarkYellow}{RGB}{0,1,1}
\usetikzlibrary{intersections}
\definecolor{iro100}{cmyk}{1,0,0,0}
\definecolor{iro90}{cmyk}{.9,0,0,0}
\definecolor{iro80}{cmyk}{.8,0,0,0}
\definecolor{iro60}{cmyk}{0,.6,0,0}
\definecolor{iro10}{cmyk}{0,.1,0,0}

\newtheorem{theorem}{Theorem}[section]

\newtheorem{lemma}[theorem]{Lemma}
\newtheorem{corollary}[theorem]{Corollary}

\theoremstyle{definition}
\newtheorem{definition}[theorem]{Definition}

\newcommand{\E}{\mathcal{E}}
\newcommand{\1}{\text{\uppercase\expandafter{\romannumeral1}}}
\newcommand{\2}{\text{\uppercase\expandafter{\romannumeral2}}}
\newcommand{\3}{\text{\uppercase\expandafter{\romannumeral3}}}
\newcommand{\4}{\text{\uppercase\expandafter{\romannumeral4}}}
\newcommand{\5}{\text{\uppercase\expandafter{\romannumeral5}}}
\newcommand{\6}{\text{\uppercase\expandafter{\romannumeral6}}}

\newcommand{\norm}[1]{\left\lVert#1\right\rVert}
\newcommand{\defeq}{\vcentcolon=}
\newcommand{\eqdef}{=\vcentcolon}

\newcommand{\op}[1]{\ket{#1}\!\bra{#1}}
\newcommand{\ketbra}[2]{\mathinner{\ket{#1}\!\bra{#2}}}
\newcommand{\kett}[1]{\mathinner{\ket{#1}\!\rangle}}

\newcommand{\one}{\mathbb{1}}
\newcommand{\R}{\mathbb{R}}
\newcommand{\Z}{\mathbb{Z}}
\newcommand{\C}{\mathbb{C}}

\newcommand{\OO}{\mathcal{O}}
\newcommand{\Hilb}{\mathcal{H}}

\DeclareMathOperator{\Tr}{Tr}

\DeclareMathOperator{\supp}{supp}

\DeclareMathOperator{\dist}{dist}

\begin{document}
\title{Local Strong-to-Weak Spontaneous Symmetry Breaking}

\author{Francisco Divi}\email{fdivi@pitp.ca}
\affiliation{Perimeter Institute for Theoretical Physics, Waterloo, Ontario N2L 2Y5, Canada}
\affiliation{Department of Physics and Astronomy, University of Waterloo, Waterloo, Ontario N2L 3G1, Canada}

\author{Leonardo A. Lessa}\email{llessa.physics@pm.me}
\affiliation{Perimeter Institute for Theoretical Physics, Waterloo, Ontario N2L 2Y5, Canada}
\affiliation{Department of Physics and Astronomy, University of Waterloo, Waterloo, Ontario N2L 3G1, Canada}

\author{Chong Wang}\email{cwang4@pitp.ca}
\affiliation{Perimeter Institute for Theoretical Physics, Waterloo, Ontario N2L 2Y5, Canada}

\begin{abstract}

We propose a local notion of strong-to-weak spontaneous symmetry breaking (SW-SSB), through a local one-point fidelity correlator. Compared with the previous definition in terms of a two-point fidelity correlator, our local formulation offers two key advantages: (1) it is easier to detect in large systems: for a system of size $N$ and with ${\rm poly}(N)$ amount of resources, one can detect the local fidelity order up to volume scale $O(\log(N))$; and (2) the local SW-SSB order remains well defined in the thermodynamic limit, where the density matrix itself is not well defined. We show that key features of SW-SSB, including stability under finite-depth symmetric channels and long-range conditional mutual information, persist within this local framework. Our definition is conceptually analogous to local thermalization, as exemplified by pure states obeying the eigenstate thermalization hypothesis (ETH). For critical states, the local one-point fidelity correlator defines an interesting class of defect problems. We demonstrate the applicability of the local formulation through several concrete examples, and derive the universal scaling behavior of the local fidelity correlator in a range of critical systems, including ground states of conformal field theories as well as ballistic and diffusive free-fermion metals.
\end{abstract}

\maketitle

\tableofcontents

\section{Introduction}
Open quantum systems display novel phenomena with no direct analogue in closed systems, and have therefore emerged as a rich yet conceptually challenging setting. Part of this richness stems from a fundamental distinction in how symmetries act on mixed states. While the notion of a symmetric state is unambiguous for pure states, ensembles of states can be symmetric in two distinct ways~\cite{BucaProsen2012,AlbertJiang2014,deGroot2022, ma_average_2023}. If $U$ represents a symmetry, a mixed-state $\rho$ is \emph{strongly} symmetric if $U \rho \propto  \rho$, or it is \emph{weakly} symmetric if it only satisfies $U \rho U^\dagger = \rho$. Physically, this distinction reflects that an ensemble can be symmetric even if its constituents carry different charges. For pure states, these two notions coincide, but for mixed states, they differ fundamentally, providing a basic organizing principle for open quantum phases of matter. 

A direct consequence is the emergence of a novel pattern of spontaneous symmetry breaking (SSB), where a strong symmetry is reduced to a weak symmetry without fully breaking the symmetry. The notion of \emph{strong-to-weak spontaneous symmetry breaking} (SW-SSB) has emerged as a unifying framework for mixed-state quantum phases, linking diverse directions ranging from information-theoretic characterizations of phases of matter~\cite{LeeJianXu,Ma2024SPTdoubled,ma_topological_2025,lessa2025strong, sala_spontaneous_2024,ChenGrover,sang2024stability} to topological quantum memory~\cite{zhang_strongweak_2024, RamanjitAbhinav, wang_intrinsic_2025, lessa_higherform_2025, sang_mixedstate_2025} and emergent hydrodynamics~\cite{OgunnaikeFeldmeierLee2023,Moudgalya_2024,GuWangWang2025,huang_hydrodynamics_2025, hauser_strongtoweak_2026}. 

Diagnosing SW-SSB requires probing the internal ensemble structure, and traditional observables linear in the density matrix are sensitive only to the average behavior. Instead, one must consider quantum-information theoretic order parameters, which are nonlinear in the density matrix. A canonical example is the two-point fidelity correlator~\cite{lessa2025strong}
\begin{equation}\label{eq:conventional_fidelity}
    F(\rho; O_x O_y^\dagger) \defeq F(\rho, O_x O_y^{\dagger}\rho O_y O^{\dagger}_x), 
\end{equation}where $O_x$ is a charged operator, and the fidelity between two density matrices $\rho$, $\sigma$ is defined as 
\begin{equation}
    F(\rho, \sigma) \defeq \Tr \sqrt{ \sqrt{\rho} \sigma \sqrt{\rho}}. 
\end{equation}
Physically, the fidelity correlator measures the similarity between the state $\rho$ and the state $\sigma\propto O_x O_y^{\dagger}\rho O_y O^{\dagger}_x$, obtained from $\rho$ by moving a charge from $x$ to $y$. In the pure state limit, this simply becomes the ordinary two-point correlator $|\langle\psi|O_xO_y^\dagger|\psi\rangle|$. In the absence of conventional SSB (diagnosed via the ordinary two-point correlator), a long-range ordered two-point fidelity, $ F(\rho; O_x O_y^\dagger) > 0$ as $|x-y|\to\infty$, defines SW-SSB and characterizes a robust mixed-state phase. Intuitively, a state with SW-SSB is insensitive to moving a charge from one place to another, so the information about the total charge becomes global, i.e. cannot be locally retrieved. For this reason, SW-SSB can also be thought of as a type of ``charge thermalization''. Indeed, the simplest example of a state exhibiting SW-SSB is an Ising spin system, with strong $\Z_2$ symmetry generated by $\prod_iX_i$, in an ``infinite-temperature'' state within a fixed global charge sector
\begin{equation}
\label{eq:infT}
    \rho_\infty=\frac{\one+\prod_i X_i}{2^N},
\end{equation}
where $N$ is the total number of qubits. The fidelity corrrelator Eq.~\eqref{eq:conventional_fidelity} for $O_x=Z_x$ is $1$ since $\rho_\infty=Z_xZ_y\rho_\infty Z_xZ_y$.

Recall that to evaluate an ordinary correlation function $\langle O_x O_y^{\dagger} \rangle = \Tr[\rho O_x O_y^{\dagger}] = \Tr[\rho_{x,y} O_x O_y^{\dagger}]$, only the reduced density matrix $\rho_{x,y}$ in the local regions around $x$ and $y$ is needed. Physically, this means that we only need local measurements near $x$ and $y$ to determine the correlation function. However, for quantities nonlinear in the density matrix, such as the fidelity correlator, it appears that the full density matrix $\rho$, not just the reduced density matrix in some local region, is required. This poses two challenges:
\begin{enumerate}
    \item \textit{Scalability}: For a system of size $V = L^d$, determining the full unknown density matrix is a tomographic task that requires $e^{\Omega(V)}$ resources, so it is clearly not scalable for large system sizes. Only in certain situations can one use special properties of the state to circumvent the exponential complexity; for example, if the state is close to a free-fermion Gaussian state, which was exploited in recent experimental detections of SW-SSB in cold Fermi gases~\cite{SWSSBExp}. In general, however, the exponential cost is unavoidable: in the absence of prior knowledge about the state, Ref.~\cite{Feng2025} shows that SW-SSB cannot be faithfully diagnosed using only ${\rm poly}(L)$ resources. 

    \item \textit{Thermodynamic limit}: It is unclear whether SW-SSB admits a sharp definition in an infinite system. In an infinite system, the very notion of a Hilbert space, and consequently the density matrix $\rho$, becomes ill-defined. Existing definitions therefore proceed by considering a sequence of finite systems ${L_i}$ with states $\rho_{L_i}$ and taking the limit $L_i \to \infty$ at the end of calculations. However, it is conceptually desirable to formulate SW-SSB intrinsically in the thermodynamic limit. How can this be achieved without reference to a global density matrix?  
\end{enumerate}

Both challenges can be resolved if SW-SSB can be formulated and detected locally; namely, based not on the global state, but only on the reduced density matrices on finite local regions. In this way, SW-SSB can be computed and measured more efficiently and remains well-defined even when the whole system is infinitely large. 

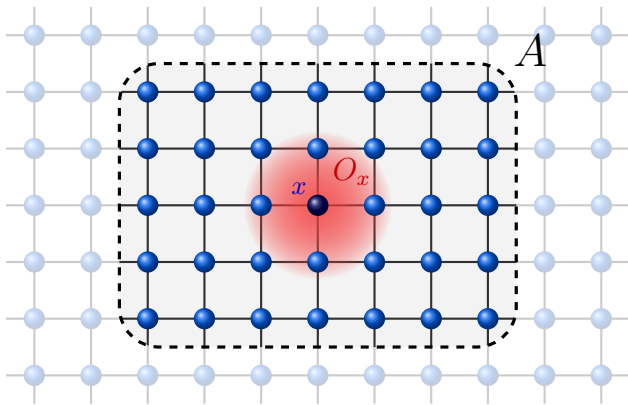
\begin{figure}[t]
        \begin{tikzpicture}[scale = 0.75]
    
    
    \def\Nx{11} 
    \def\Ny{7}  
    
    \def\Axmin{3} 
    \def\Axmax{9} 
    \def\Aymin{2} 
    \def\Aymax{6} 
    
    \def\ix{6} 
    \def\iy{4} 
    
    \tikzset{
        site/.style={
            circle,
            minimum size=8pt,
            inner sep=0pt,
            shading=ball,
            ball color=blue!60!cyan 
        },
        highlighted site/.style={
            circle,
            minimum size=8pt,
            inner sep=0pt,
            shading=ball,
            ball color=blue!40!black 
        },
        grid line/.style={
            thick,
            black!80
        },
        region A border/.style={
            dashed,
            very thick,
            black,
            rounded corners=15pt
        },
        overlay/.style={
            fill=white, 
            opacity=0.75, 
            even odd rule
        }
    }
    
    
    
    \foreach \y in {1,...,\Ny} {
        \draw[grid line] (0.5, \y) -- (\Nx+0.5, \y);
    }
    \foreach \x in {1,...,\Nx} {
        \draw[grid line] (\x, 0.5) -- (\x, \Ny+0.5);
    }
    
    \foreach \r in {1, 2, ..., 40} {
        \fill[red!70, opacity=0.05] (\ix, \iy) circle ({1.3 * (41-\r)/40});
    }
    
    \foreach \x in {1,...,\Nx} {
        \foreach \y in {1,...,\Ny} {
            \ifnum\x=\ix
                \ifnum\y=\iy
                    \node[highlighted site] at (\x, \y) {};
                \else
                    \node[site] at (\x, \y) {};
                \fi
            \else
                \node[site] at (\x, \y) {};
            \fi
        }
    }
    
    \node[blue!80!black, font=\normalsize, anchor=south east, inner sep=4pt] at (\ix, \iy) {$x$};
    \node[red!80!black, font=\large] at (\ix + 0.6, \iy + 0.6) {$O_x$};
    
    \fill[overlay] 
        (0.2, 0.2) rectangle (\Nx+0.8, \Ny+0.8) 
        {[rounded corners=15pt] (\Axmin-0.5, \Aymin-0.5) rectangle (\Axmax+0.5, \Aymax+0.5)};
    
    \draw[region A border, fill = black, fill opacity = 0.05] (\Axmin-0.5, \Aymin-0.5) rectangle (\Axmax+0.5, \Aymax+0.5);
    
    \node[font=\LARGE, anchor=south west] at (\Axmax+0.3, \Aymax+0.3) {$A$};
    
    \end{tikzpicture}
    \vspace{-0.5em}
    \caption{Schematic illustration of the local fidelity correlator $F(\rho_A;O_x) \defeq F(\rho_A, O_x \rho_A O_x^\dagger)$, Eq. \eqref{eq:LFC}. }
    \label{fig:mainfig}
\end{figure}

In this work, we propose such a local formulation. We consider the one-point fidelity correlator of a reduced density matrix on a finite local region $A$ containing the point $x$, or simply one-point \textit{local fidelity correlator} (LFC), defined as (See Fig. \ref{fig:mainfig})
\begin{equation}
\label{eq:LFC}
    F(\rho_A;O_x) \defeq F(\rho_A,O_x\rho_AO^{\dagger}_x).
\end{equation}
This quantity is manifestly local, with computational and measurement complexity only depending on the size of $A$ rather than the global length scale $L$. We are interested in the behavior as the distance from $x$ to the boundary of $A$, $\ell=\dist(x , \partial A)$, grows. The dependence only on the local state $\rho_A$ instead of the global state $\rho$ also ensures that the LFC is well-defined even for infinite systems. We then say that a state exhibits \textit{local SW-SSB} if
\begin{equation}
\label{eq:FOP}
    \lim_{|A| \to \infty} F(\rho_A, O_x \rho_A O_x^\dagger)>0,
\end{equation}
where the limit is taken over an increasing sequence of regions that covers the entire infinite system.   

Practically, on a large but finite system, local SW-SSB means that the LFC saturates to a nonzero constant beyond certain length scale $\ell>\ell_0$, where $\ell=\dist(x,\partial A)$. The absence of local SW-SSB means that the LFC decays to zero as some function of $\ell$, typically as $e^{-\ell/\xi}$ for some decay length $\xi$, or as some power-law $\ell^{-\alpha}$ for critical states. 

If the reduced density matrix $\rho_A$ has strong symmetry --- for example, in a symmetric pure product state $|++\cdots\rangle$ --- then the one-point LFC is trivially zero. In the opposite limit, if we have the infinite-temperature symmetric state of Eq.~\eqref{eq:infT}, the local reduced density matrix is simply the maximally mixed state $\rho_A \propto \one_A$, and the one-point LFC is always nonzero. So the LFC can be viewed as a measure of global charge spreading. Intuitively, we are measuring the ``fidelity correlation'' between the local operator $O_x$ and the boundary of the region $\partial A$, which can be viewed as a defect that can break the strong symmetry.   

In this work, we show that local SW-SSB, as the name suggests, can indeed be viewed as a local version of SW-SSB, sharing analogous properties with the latter. Specifically, we show that \
\begin{enumerate}
    \item Standard SW-SSB implies local SW-SSB (Theorem~\ref{thm:inequality_of_measures}).
    \item If the global state is strongly symmetric, then local SW-SSB implies the presence of long-range conditional mutual information (Theorem~\ref{thm:LRCMI}). This is a hallmark feature, though not a sufficient condition, of standard SW-SSB~\cite{lessa2025strong}. 
    
    \item There are states with local SW-SSB, but not SW-SSB in the standard sense. Examples include pure states satisfying the eigenstate thermalization hypothesis (ETH, see Sec.~\ref{sec:eth}), as well as the pseudo-SWSSB states discussed in Ref.~\cite{Feng2025} (Sec.~\ref{sec:pseudo-swssb}). In this sense, the standard SW-SSB is a stronger statement, and we will call it \textit{global SW-SSB}.
    
    \item Similar to global SW-SSB, local SW-SSB is robust against finite-depth, strongly symmetric quantum channels (Theorem~\ref{thm:stability}). This \textit{stability theorem} makes local SW-SSB a universal property of symmetric mixed-state phases.
    
    \item In an infinite system, the limit that defines the fidelity order parameter Eq.~\eqref{eq:FOP} exists and is independent of the choice of sequence of $A$ (Theorem~\ref{thm:existence_and_uniqueness}). This makes the notion of local SW-SSB well defined in infinite systems. Interestingly, in an infinite system, the strong symmetry itself does not appear to be well defined. Nevertheless, the local version of strong-symmetry breaking remains a well-defined and universal property. Local SW-SSB on infinite systems also implies a version of long-range CMI after averaging over symmetry sectors (Theorem.~\ref{thm:long-range_symmetry-averaged_CMI}).
    
    \item We also discuss a local version of the two-point fidelity correlator
    \begin{equation}
        \lim_{|x-y|\to\infty}\lim_{|A|\to\infty}F(\rho_A;O_x O_y^\dagger),
    \end{equation}
    which is closer in appearance to the original fidelity correlator Eq.~\eqref{eq:conventional_fidelity}. We show that, as long as $A$ is much smaller than the entire system --- more formally, the limit $L\to\infty$ is taken before $|A|\to\infty$ --- the local two-point fidelity correlator is equivalent to the local one-point fidelity correlator in terms of characterizing SW-SSB (Theorem~\ref{thm:two-point_one-point_equiv}). 
\end{enumerate}

Our local notion of SW-SSB significantly improves the efficiency of detection: even without using any structure of $\rho_A$, one can evaluate or measure the LFC with $e^{O(\ell^d)}$ amount of resources ($d$ being the space dimension). With ${\rm poly}(L)$ amount of resources, one can then measure LFC up to a length scale $\ell_{\rm max}\sim \log^{1/d}(L)$. This can be used to detect the absence of local SW-SSB if the decay length $\xi<O(\log^{1/d}(L))$, or the presence of local SW-SSB if the saturation length $\ell_0<O(\log^{1/d}(L))$. Moreover, we show in Sec. \ref{sec:markov_length_symmetry-averaged} that the LFC approximately saturates to its asymptotic value at length scales $\xi$ large compared to the \textit{local} Markov length $\xi_M$, assuming the CMI decays exponentially as $I(A:C|B) \sim \exp(-\dist(A,C)/\xi_M)$ over finite but large regions $ABC$.

The trade-off of making SW-SSB locally detectable is that local SW-SSB does not, in general, imply its global counterpart. Nevertheless, in most physically relevant settings studied recently, such as finite-time or finite-decoherence transitions, the local and global notions yield identical phase diagrams. This agreement arises because symmetric phases in these settings typically have decaying conditional mutual information.

We study several examples of long-range and critical local SW-SSB, including thermal states, the decohered Ising model, pure ground states of conformal field theories (CFT) and free-fermion metals. For critical states, the local one-point fidelity or R\'enyi correlator defines an interesting defect problem. For CFT ground states, the R\'enyi-$1$ correlator scales the same way as the ordinary two-point correlator
\begin{equation}
    R^{(1)}(\rho_A;\OO_x)\sim \frac{1}{\ell^{2\Delta_{\OO}}}.
\end{equation}

The scaling behavior is quite different for metals. For a free-fermion metal in $d$ space dimensions with a Fermi surface, the R\'enyi-$1$ correlator for the fermion operator scales as
\begin{equation}
    R^{(1)}(\rho_A;c_x)\sim \frac{1}{\ell},
\end{equation}
independent of the dimension $d$. If we introduce quenched disorder and consider a diffusive metal, the scaling behavior becomes
\begin{equation}
    R^{(1)}(\rho_A;c_x)\sim \frac{1}{\ell^2},
\end{equation}
again independent of $d$. In both examples the LFC scales very differently from the two-point fermion correlator in general dimension, which is very different from the behavior in CFTs.

The rest of the paper is organized as follows. In Sec.~\ref{sec:properties} we discuss properties of local SW-SSB order on large but finite systems. In Sec.~\ref{sec:infinite} we discuss local SW-SSB defined directly on infinite systems. We then discuss alternative formulations of SW-SSB using R\'enyi correlators (Sec.~\ref{sec:alternative_definitions}) as well as two-point generalizations of the fidelity and R\'enyi correlators (Sec.~\ref{sec:equivalence_1P_2P}). In Sec.~\ref{sec:examples} we discuss several physically relevant examples, including the decohered Ising paramagnet (Sec.~\ref{sec:ZZ_decoherence}), ground states of CFTs (Sec.~\ref{sec:CFT}), free-fermion metals (Sec.~\ref{sec:metals} and random free-fermion Gaussian states that realize a weak form of ETH (Sec.~\ref{sec:Gaussian}). In Sec.~\ref{sec:non_abelian_definition} we generalize some of our discussions to non-abelian symmetry groups. We end with some discussions on future directions in Sec.~\ref{sec:discussion}. Several Appendices contain peripheral details.

\section{Universal properties}
\label{sec:properties}
We now explore universal properties of local SW-SSB defined through Eq.~\eqref{eq:FOP}. In this Section we shall focus on large but finite systems, and we shall discuss infinite systems in the next Section. On large but finite systems, the symbol $|A|\to\infty$ should be understood as having $|A|$ much larger than any microscopic characteristic length scale, but much smaller than the system size $L$~\footnote{One can be mathematically more precise, by defining a sequence of systems $\{\rho_n\}$ and regions $A_n$, with system size $L_n\to\infty$ and $|A_n|\to\infty$, but keeping $\ell_{A_n}/L_n$ small throughout.}. We shall also work on lattice systems where the local Hilbert space is finite-dimensional, so that the operator norm of any local operator $\norm{O_x}_{\infty}$ is finite.

\subsection{Global SW-SSB implies local SW-SSB}

Our first result is that global SW-SSB, defined as $F(\rho;O_x O^{\dagger}_y)>0$ for large $|x-y|$, implies local SW-SSB. This is an immediate consequence of data processing inequality.

\begin{theorem}\label{thm:inequality_of_measures}
For any mixed state $\rho$, region $A$, and local operators $O_x$ supported in $A$ and $O_y$ supported outside $A$,  we have
\begin{equation}
\norm{O_y}_\infty F(\rho_A;O_x) \ge F(\rho;O_x O_y^\dagger).
\end{equation}
\end{theorem}

\begin{proof}
The inequality follows from the monotonicity of fidelity under completely positive trace-preserving maps (the data processing inequality). Applying the partial trace $\Tr_{\overline{A}}(\cdot)$ yields
\begin{equation}
\begin{aligned}
 F(\rho, O_x O_y^\dagger \rho O_y  O_x^\dagger) & \leq F(\rho_{A}, O_x \Tr_{\overline{A}}[O_{y}^\dagger \rho O_{y}] O_x^\dagger ) \\
 & \leq \norm{O_y}_\infty F(\rho_A, O_x \rho_{A} O_x^\dagger ),
\end{aligned}
\end{equation}
where in the last line we used the fact that $O_y O_y^\dagger \leq \norm{O_y}_\infty^2 \one$, and that $\sigma \leq \sigma' \Rightarrow F(\rho, \sigma) \leq F(\rho, \sigma')$.

\end{proof}

One can similarly use the data processing inequality to prove that $F(\rho_A;O_x)\geq F(\rho_{A^+};O_x)$ if $A\subseteq A^+$. Namely, the one-point LFC on a finite region upper bounds the same LFC on larger regions.

\subsection{Long-range conditional mutual information}
\label{sec:long-range_CMI}

One of the most important physical consequences of SW-SSB is that it obstructs recovering the state from local data. Intuitively, a state with SW-SSB has a well-defined global charge, yet this information is not accessible in any finite region. As a result, the full mixed state cannot be reconstructed from the reduced density matrix alone.

A natural way to quantify this obstruction is through the conditional mutual information (CMI), $I(A\!:\!C|B)$. Given a tripartition $A|B|C$ of the system (Fig.~\ref{fig:cmi_geometry}), the CMI measures whether the full state $\rho$ can be reconstructed from $\rho_{AB}$ by a recovery channel acting only on $B$.  
Vanishing CMI means that the correlations between $A$ and $C$ are completely encoded in $AB$, and optimal recovery is possible (i.e., the state is quantum Markov). In contrast, finite CMI means that $C$ retains genuinely nonlocal information that cannot be reconstructed from local data near $A$, as quantified by the  approximate recovery relation \cite{fawzi_quantum_2015, junge_universal_2018}
\begin{equation}\label{eq:fawzi_relation}
    I({A:C}|B)_\rho \geq - 2 \log F(\rho, \mathcal{R}_{B \to BC}[\rho_{AB}]),
\end{equation}
where $\mathcal{R}_{B \to BC}$ is related to the Petz recovery map.

\begin{figure}[t]
\centering
\begin{tikzpicture}[ scale=0.7, boundary/.style={draw=black!70, line width=0.8pt}, labelstyle/.style={font=\large}]
\def\rA{1.05}
\def\rAB{2*\rA}
\def\wC{(\rAB+0.4)}
\fill[green!12] ({-\wC}, {-\wC}) rectangle ({\wC}, {\wC});
\fill[black!10] (0,0) circle (\rAB);
\fill[blue!12] (0,0) circle (\rA);
\draw[boundary] (0,0) circle (\rAB);
\draw[boundary] (0,0) circle (\rA);
\node[labelstyle] at (0,0) {$A$};
\node[labelstyle] at (-1.15,1.15) {$B$};
\node[labelstyle] at (-1.9,1.9) {$C$};
\end{tikzpicture}
\caption{Tripartition $A|B|C$ of the system used to compute the conditional mutual information $I(A\!:\!C|B)$.}
\label{fig:cmi_geometry}
\end{figure}
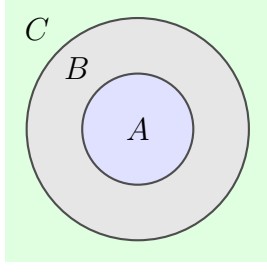

For strongly symmetric states with SW-SSB, this missing information is precisely the global charge sector. A local charged operator changes the global state to an orthogonal one, but by SW-SSB, this changes the reduced density matrix only slightly in local fidelity. Therefore, $\rho_{AB}$ does not contain enough information to determine the correct global continuation into $ABC$. The following theorem quantifies how a finite local SW-SSB measure obstructs recoverability through a finite CMI.

\begin{theorem}
\label{thm:LRCMI}
Let $A|B|C$ be a tripartition of a finite system, and let $\rho$ be a strongly symmetric state. Then, for any charged operator $O_A$ supported on $A$
\begin{equation}
    I({A:C}|B)_\rho \geq \left( \frac{F(\rho_{AB}; O_A)}{c \norm{O_A}_\infty } \right)^2,
\end{equation}
where $c = O(1)$ is a numerical constant.
\end{theorem}
\begin{proof}
    This follows from the more general and upcoming Theorem~\ref{thm:local_computability} by noting that $F(\rho; O_A) = 0$ due to the strong symmetry of $\rho$.
\end{proof}

As a consequence, a state with local SW-SSB has long-range CMI and cannot be an approximate quantum Markov chain. For that, it is crucial to assume the strong symmetry for the global state $\rho$. Otherwise, the statement can be trivially violated by the maximally mixed state $\rho\propto \one$, for which any reduced density matrix is also maximally mixed and hence has nonzero LFC.

Long-range CMI is an important feature of global SW-SSB but there are other ways to have long-range CMI without having SW-SSB. Therefore, Thm.~\ref{thm:LRCMI} does not imply that local SW-SSB requires global SW-SSB. Indeed, there are states that exhibit SW-SSB locally but not globally.

\subsection{Examples of local, but not global, SW-SSB}
\label{sec:localbutnotglobal}

We now discuss some conceptually simple examples with local, but not global, SW-SSB. 

\subsubsection{Eigenstate thermalization}\label{sec:eth}
Our first example is simply a pure state satisfying the eigenstate thermalization hypothesis (ETH)~\cite{ETHReview}: $\rho=|\Psi\rangle\langle\Psi|$. Here, $|\Psi\rangle$ is a highly excited eigenstate of a thermalizing Hamiltonian, with nonzero excitation energy density. We also assume that the energy density (effective temperature) is high enough that the symmetry is not broken in the ordinary sense.

A key property we anticipate from ETH eigenstates is that local regions are thermal~\cite{SubETH1,SubETH2}: for a region $A$ small compared to the system size $\ell_A\ll L$, the reduced density matrix $\rho_A\approx e^{-\beta H}/Z$ for some effective inverse temperature $\beta$, with error vanishing as $L\to \infty$. Note that the local thermal state is in canonical ensemble, with only the weak symmetry.

As argued in Refs.~\cite{lessa2025strong,liu2025diagnosing}, the fidelity correlator (one-point or two-point) is non-vanishing for a thermal Gibbs state. Therefore the one-point local fidelity correlator Eq.~\eqref{eq:LFC} is non-vanishing and the ETH eigenstate has local SW-SSB order. 

One can also see that the ETH eigenstate cannot have global SW-SSB order: for a pure state, the (global) two-point fidelity correlator is simply the ordinary two-point correlation function $F(|\Psi\rangle\langle\Psi|;O_xO_y^\dagger)=|\langle\Psi|O_xO_y^\dagger|\Psi\rangle|$. For an ETH eigenstate at high energy density, the effective temperature is high and the two-point function will decay exponentially.

\begin{figure}[t]
\centering
\begin{tikzpicture}
\begin{axis}[width=1\linewidth, height=0.6\linewidth, xmin=0, xmax=13, ymin=0, ymax=1, axis lines=left, xlabel={$\ell$}, xtick=\empty, ytick=\empty, clip=false, thick]

\draw[<->, black] (axis cs:0,0.52) -- (axis cs:2.9,0.52) node[midway, below, font=\footnotesize] {$\xi_M$};

\addplot[thick, darkred, domain=0:12, samples=100] {0.9*exp(-x)};

\addplot[thick, BlueViolet, domain=0:12, samples=200] {(0.5 + 0.18*exp(-x/1.6))*(1/(1+exp((x-10.4)/0.28)))};

\draw[<->, black] (axis cs:0,0.16) -- (axis cs:1.7,0.16) node[midway, below, font=\footnotesize] {$\xi$};

\draw[densely dashed, black!70] (axis cs:10.4,0) -- (axis cs:10.4,0.72);

\node[font=\footnotesize, anchor=north] at (axis cs:10.4,0) {$\sim L$};

\node[darkred, font=\footnotesize, anchor=west] at (axis cs:1.3,0.28)
{$F(\rho;O_{\ell} O_0^\dagger)$};

\node[BlueViolet!80!black, font=\footnotesize, anchor=west] at (axis cs:4.8,0.57)
{$F(\rho_A;O_x)$};

\end{axis}
\end{tikzpicture}
\caption{Schematic behavior for a pure state satisfying ETH. The one-point local fidelity $F(\rho_A;O_x)$ saturates to a finite value at intermediate length scales set by the Markov length $\xi_M$ (see Sec.~\ref{sec:markov_length_symmetry-averaged}), and decreases to zero when $\ell$, the linear size of the region $A$, becomes comparable to $L$, the linear size of the full system. In contrast, the (global) two-point fidelity $F(\rho;O_{\ell}O_0^\dagger)$ equals the ordinary two-point function and decays exponentially with a finite correlation length, $\xi$.}
\label{fig:ETH_schematic}
\end{figure}
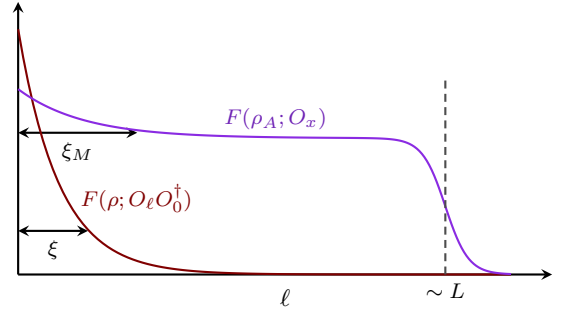

\subsubsection{Global correlation in pure states}

We can gain more insights on the ETH example through the following theorem.

\begin{theorem}\label{thm:pure_state_non_local_charge}
Consider a bipartite system $\Hilb_A \otimes \Hilb_B$, with $\dim(\Hilb_B) \geq \dim(\Hilb_A)$, and let $\rho_{AB} = |\psi_{AB}\rangle \langle \psi_{AB}|$ be a pure state strongly symmetric under a bipartite Abelian symmetry $U(g) = U_A(g) \otimes U_B(g)$. Then, for any charged operator $O_A$ supported on $A$, there exists a charged operator $O_B$ supported on $B$ of unit norm, $\norm{O_B}_\infty  = 1$, such that
\begin{equation}\label{eq:pure_state_non_local_charge}
    F(\rho_A; O_A ) =|\langle\psi_{AB}|O_AO_B^\dagger |\psi_{AB}\rangle|.
\end{equation}
\end{theorem}

Note that $O_B$ need not be a local operator. The theorem implies that any pure state exhibiting local SW-SSB --- such as ETH eigenstates --- must display long-range correlations. At the same time, the absence of global SW-SSB (which, for pure states, is equivalent to the absence of ordinary SSB) requires that the corresponding charged operator $O_B$ acting on the larger region be intrinsically nonlocal.

\begin{proof}
Since $\dim(\Hilb_B)\geq \dim(\Hilb_A)$, by Uhlmann's theorem there exist purifications $\ket{\phi_{AB}}$ and $\ket{\phi'_{AB}}$ of $\rho_A$ such that
\begin{equation}
    F(\rho_A;O_A)=|\langle\phi_{AB}| O_A |\phi'_{AB}\rangle| .
\end{equation}
Furthermore, there exist unitaries $V,V'\in \mathcal U(\Hilb_B)$ such that
\begin{equation}
    |\phi_{AB}\rangle=(\one_A\otimes V) |\psi_{AB}\rangle,\quad
    |\phi'_{AB} \rangle=(\one_A\otimes V') |\psi_{AB}\rangle .
\end{equation}
Thus, introducing $W \defeq V^\dagger V'$, we have
\begin{equation}\label{eq:W_def}
    F(\rho_A;O_A)=
    \bigl|\langle \psi_{AB}| O_A\otimes W |\psi_{AB}\rangle\bigr| .
\end{equation}

Let $\lambda$ be the charge of $O_A$. Since $\ket{\psi_{AB}}$ is symmetric, only the inverse charge sector of $O_A$, $\bar \lambda$, of $W$ can contribute to Eq.~\eqref{eq:W_def}. Indeed, if $X$ has charge $\mu$ on $B$, then $ \langle \psi_{AB}|O_A\otimes X|\psi_{AB}\rangle = \lambda(g)\mu(g)  \langle \psi_{AB}|O_A\otimes X|\psi_{AB}\rangle$ for all $g\in G$, and therefore the matrix element can be nonzero only when $\mu=\bar\lambda$.

For an Abelian group, the projector onto a charge sector is
\begin{equation}
    \Pi_\mu(X)=  \int_G \!dg~\overline{\mu(g)} ~U_B(g)^\dagger X U_B(g),
\end{equation}
where $dg$ is the Haar measure on $G$ (for finite $G$, replace
$\int_G dg$ by $\frac{1}{|G|}\sum_{g\in G}$). Therefore, Eq.~\eqref{eq:pure_state_non_local_charge} holds with
\begin{equation}
    O_B^\dagger \defeq \Pi_{\bar\lambda}(W) .
\end{equation}

It remains to show that $\norm{O_B}_\infty =1$. First,
\begin{equation}
\begin{aligned}
    \norm{\Pi_{\bar\lambda}(W)}_\infty \leq
    \int_G \! dg~\norm{U_B(g)^\dagger W U_B(g)}_\infty = \|W\|_\infty = 1 ,
\end{aligned}
\end{equation}
since $W$ is unitary. 

For the converse bound, we apply Theorem~\ref{thm:inequality_of_measures}
\begin{equation}
\begin{aligned}
    F(\rho_A;O_A) &\geq \norm{O_B}_\infty ^{-1}
    F(\rho_{AB};O_A O_B^\dagger)\\ 
    &= \norm{O_B}_\infty^{-1}F(\rho_A;O_A).
\end{aligned}
\end{equation}
Thus, in the nontrivial case $F(\rho_A;O_A)\neq 0$ we have $\norm{O_B}_\infty \geq 1$. Combining both bounds yields $\norm{O_B}_\infty =1$.
\end{proof}

\subsubsection{Pseudo-SWSSB}\label{sec:pseudo-swssb}

Another example with local, but not global, SW-SSB order is the ``pseudo-SWSSB'' ensemble of states constructed in Ref.~\cite{Feng2025} through the application of pseudo-random unitary circuits. For our purposes, the pseudo-SWSSB ensemble $\{p_i, \rho_i\}$ has two important properties:
\begin{enumerate}
    \item Each $\rho_i$ is strongly symmetric and does not have global SW-SSB order; namely, the standard two-point fidelity correlator Eq.~\eqref{eq:conventional_fidelity} vanishes at long distances\footnote{More precisely, \cite{Feng2025} requires the vanishing of the Rényi-1 correlator on average: $\lim_{|A| \to \infty} \frac{1}{|A|^2}\sum_{x,y \in A} R^{(1)}(\rho; O_x O_y^\dagger) = 0$. From the discussion in Sec.~\ref{sec:equivalence_1P_2P}, it suffices to consider the long-distance behavior.}.
    \item There is no efficient measurement protocol that distinguishes $\rho_{i}$ from $\rho_\infty=\frac{1}{2^N}(\one+\prod_i X_i)$ with high probability. Here being ``efficient'' means only requiring ${\rm poly}(N)$ amount of resources, including the number of copies of the state and the number of unitary gauges and measurements.
\end{enumerate}

These two properties ensure that, even though the ensemble has no global SW-SSB, confirming this in a state-agnostic way is exponentially hard, as any ${\rm poly}(N)$ protocol will not be able to distinguish the ensemble states from $\rho_\infty\sim \one+\prod_i X_i$, which does have SW-SSB order.

From our local point of view, property $2$ above implies that, on any subsystem with volume $|A|\lesssim O(\log(L))$, where any reduced density matrix $\rho_{i,A}$ has ${\rm poly}(L)$ dimensions and can therefore be efficiently determined (say, by tomography), the state $\rho_{i,A}$ will not be distinguishable from $\Tr_{A^c}[\rho_\infty] \propto \one_A$. Indeed, they will be $o(1/\text{poly}(N))$ close in trace norm, which immediately implies local SW-SSB order for the pseudo-SWSSB ensemble.

\subsection{Stability theorem}
\label{sec:stability}

A key property establishing SW-SSB as a robust phase of matter is its stability under deformations by strongly symmetric finite-depth quantum channels.  Here, a channel is strongly symmetric if it maps strongly symmetric states to strongly symmetric states or, equivalently, if it admits a Kraus representation 
\begin{equation}
    \E:\rho\to\sum_a K_a\rho K_a^{\dagger}
\end{equation}
in which every Kraus operator $K_a$ commutes with the symmetry. The channel is ``finite-depth'' if it can be purified to a finite-depth local unitary circuit:
\begin{equation}\label{eq:proof_stability_theorem_definition_rho}
    \E_{\rm SFD}[\rho] = \Tr_{a}\left(U^\dagger \rho \otimes \ket{0}\bra{0}_a  U\right),
\end{equation}
where $a$ denotes the ancilla Hilbert space, which is taken to be a copy of the original Hilbert space, and $U$ is a finite-depth local unitary acting on both the ancilla and physical Hilbert spaces. For strongly symmetric channels, $U$ commutes with the strong symmetry $g \otimes \mathbbm{1}_a$.

The stability theorem for global SW-SSB~\cite{lessa2025strong} states that a state exhibiting global long-range order retains the order after applying any finite-depth strongly symmetric channel. This robustness is essential: quantities that are not stable under such deformations, such as the two-point R\'enyi-$2$, cannot serve as intrinsic order parameters of a phase of matter.

Analogously, the following theorem establishes local SW-SSB as a robust mixed-state phase.
\begin{theorem}\label{thm:stability} 
If a mixed state $\rho$ has SW-SSB in the one-point fidelity sense and $\E_{\rm SFD}$ is a strongly symmetric finite-depth local quantum channel, then $\E_{\rm SFD}[\rho]$ has SW-SSB in the one-point fidelity sense. 
\end{theorem}

\begin{proof}

Consider the purified (Stinespring dilation) form of the symmetric finite depth channel Eq.~\eqref{eq:proof_stability_theorem_definition_rho}. Let $r$ be the light cone size of $U$, which is proportional to the circuit depth $D$, and define the enlarged region
\begin{equation}
    A^{+r} = \{x \in \Lambda |~{\rm dist}(x, A) \leq r\}.
\end{equation}

\begin{figure}[b]
\centering
\begin{tikzpicture}[
    scale=0.6,
    site/.style={circle, fill=black, inner sep=1.pt},
    gate/.style={draw, rounded corners=2pt, fill=cyan!17, minimum width=0.9cm, minimum height=0.1cm},
    cone/.style={draw=blue!6, fill=blue!6, opacity=0.55},
]

\path[cone] (4.6,0) -- (8.4,0) -- (11.4,2.6) -- (1.6,2.6) -- cycle;

\draw[gray!50, -] (4.6,0) -- (1.6,2.6);
\draw[gray!50, -] (8.4,0) -- (11.4,2.6);

\foreach \i in {0,...,13}{
    \draw[gray!35] (\i,0) -- (\i,2.6);
}

\foreach \i in {0,...,13}{
    \node[site] at (\i,0) {};
    \node[site] at (\i,2.6) {};
}

\foreach \y in {0.6,2}{
    \foreach \i in {0,2,...,12}{
        \node[gate] at ({\i+0.5},\y) {};
    }
}
\foreach \y in {1.3}{
    \foreach \i in {1,3,...,11}{
        \node[gate] at ({\i+0.5},\y) {};
    }
}

\draw[decorate, decoration={brace, mirror, amplitude=5pt}, blue!30!black] 
    (4.7,-0.38) -- (8.3,-0.38) node[midway, below=6pt] {$A$};

\draw[decorate, decoration={brace, amplitude=4pt}, blue!70!black] 
    (1.7,3) -- (11.2,3) node[midway, above=5pt] {$A^{+r}$};

\draw[<->, black] (13.5,0.3) -- (13.5,2.4) node[midway, right] {$r$};
\end{tikzpicture}
\caption{After applying a finite-depth local channel $\mathcal{E}$, the reduced density matrix on region $A$ can only be influenced by the state in the enlarged region $A^{+r}$.}
\label{fig:stability_theorem_circuit}
\end{figure}
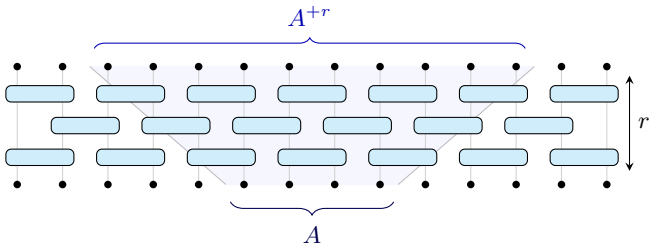

This region contains all sites that can influence $\rho_A$ after a single application of $\E_{\rm SFD}$, as depicted in Fig.~\ref{fig:stability_theorem_circuit}. Consequently, we may define a new channel supported on $A^{+r}$ whose action on $A$ coincides with that of $\E_{\rm SFD}$. To do so, we introduce a reference pure state $\rho_{\overline{A^{+r}}}^{(0)} = \ket{0} \bra{0}_{\overline{A^{+r}}}$ (the overline denotes the complement of a region), and define the channel   
\begin{equation}
    \E_{A^{+r}}[\rho_{A^{+r}}] = \Tr_a \Tr_{\overline{A^{+r}}}\left(U \left(\rho_{A^{+r}} \otimes \rho_{\overline{A^{+r}}}^{(0)} \otimes \ket{0}\bra{0}_a \right) U^{\dagger} \right).
\end{equation}
By construction, this channel satisfies
\begin{equation}\label{eq:proof_stability_theorem_property_channel}
    \Tr_{A^{+r}/ A}\E_{A^{+r}}[\rho_{A^{+r}}] = \Tr_{\bar A}\E_{\rm SFD}[\rho] 
\end{equation}
Now, since $\rho$ has local SW-SSB, we have 
\begin{equation}
    F(\rho_{A^{+r}}, O_x\rho_{A^{+r}}O_x^\dagger ) = O(1)
\end{equation}
This implies that there exist purifications $\ket{\phi_{\rho_{A^{+r}}}}$ and $\ket{\phi_{\rho_{A^{+r}}}'}$ of $\rho_{A^{+r}}$, such that 
\begin{equation}
    \big|\bra{\phi_{\rho_{A^{+r}}}} O_x \ket{\phi_{\rho_{A^{+r}}}'}\big| = O(1)
\end{equation}
Then, we find
\begin{equation}
\begin{aligned}
    O(1) &= \big|\bra{\phi_{\rho_{A^{+r}}}} \bra{0}_{\bar A^{+r}} \bra{0}_a U^\dagger U O_x  U^\dagger U \ket{0}_a \ket{0}_{\bar A^{+r}} \ket{\phi_{\rho_{A^{+r}}}'}\big| \\
    &= \big|\bra{\phi_{\E_{A^{+r}}[\rho_{A^{+r}}]}} \tilde O_x \ket{\phi_{\E_{A^{+r}}[\rho_{A^{+r}}]}'}\big|,
\end{aligned}
\end{equation}
where $\ket{\phi_{\E_{A^{+r}}[\rho_{A^{+r}}]}}$, $\ket{\phi_{\E_{A^{+r}}[\rho_{A^{+r}}]}'}$ denote purifications of $\E_{A^{+r}}[\rho_{A^{+r}}]$, and $\tilde O_x =  U O_x U^\dagger$. Note that, since $U$ is strongly symmetric and has finite depth, we may expand $\tilde O_x = \sum_{x'}  O_{x'}\otimes O_{x'}^a$, where the sum is over a (finite) set of charged operators supported near $x$. It follows that there must exist at least one term in this sum such that
\begin{equation}
      \big|\bra{\phi_{\E_{A^{+r}}[\rho_{A^{+r}}]}} O_{x'}\otimes  O_{x'}^a \ket{\phi_{\E_{A^{+r}}[\rho_{A^{+r}}]}'}\big| = O(1).
\end{equation}
This implies that 
\begin{equation}
    F(\E_{A^{+r}}[\rho_{A^{+r}}], O_{x'}\E_{A^{+r}}[\rho_{A^{+r}}]O_{x'}^\dagger) = O(1)
\end{equation}
Lastly, we apply the data processing inequality to trace over $A^{+r}/A$ and use Eq.~\eqref{eq:proof_stability_theorem_property_channel}, to conclude 
\begin{equation}
    F(\Tr_{\bar A}\E_{\rm SFD}[\rho], O_{x'}\Tr_{\bar A}\E_{\rm SFD}[\rho]O_{x'}^\dagger) = O(1).
\end{equation}
\end{proof}

In Sec.~\ref{sec:localbutnotglobal}, we examined states that exhibit local, but not global, SW-SSB, including ETH eigenstates and pseudo-SWSSB states. In light of the stability theorem, these states nonetheless belong to nontrivial symmetric mixed-state phases.

\section{Infinite systems}
\label{sec:infinite}

We now discuss the notion of SW-SSB directly on an infinite lattice $\Lambda$, without taking limits on sequences of finite-size systems. The main subtlety for an infinite system is that the total Hilbert space is not well-defined, so the global density matrix $\rho$ is also not well-defined. Instead, in the operator-algebraic approach a state is defined through a mapping $\omega$ that takes quasi-local operators (i.e. limits of finitely supported operators) to their expectation values~\cite{bratteli_operator_2012}:
\begin{equation}
    \omega: O_{ql}\to \langle O_{ql}\rangle.
\end{equation}
As a consequence, the standard definition of strong symmetry $U\rho=e^{i\alpha}\rho$ does not generalize easily --- physically, it is not obvious how to define the ``total charge'' of an infinite system, since it is not measured by a quasi-local operator. Informally, the total charge can be changed by acting with a charged operator ``at infinity''. Weak symmetry, on the other hand, remains well defined by requiring $\langle U^{\dagger}O_{ql}U\rangle=\langle O_{ql}\rangle$ for all quasi-local operators $O_{ql}$. Curiously, as we will see in this Section, the notion of SW-SSB, at least in the local sense, remains well defined. 

Notation-wise, we will use $\omega$ to refer to the quantum state of an infinite-volume system, defined above as a positive functional over the quasi-local algebra of operators. When expressing the reduced density matrix of this state over a finite region $A$, we will use the standard notation $\rho_A$, while leaving $\rho$ to occasions where the total system is finite. Namely, $\rho_A$ is defined as the unique matrix satisfying $\omega(O_A) = \Tr[\rho_A O_A]$ for all operators $O_A$ supported on $A$.

Our first goal is to sharply define the order parameter $\lim_{|A|\to\infty}F(\rho_A;O_x)$ through a limiting procedure. We formalize this procedure through the notion of a covering, which is illustrated in Fig.~\ref{fig:covering}.

\begin{definition}[Covering centered at $x$] \label{def:Convering_centered_around_i} Given a site $x$, we define a \textbf{covering centered at $x$} as a sequence of finite regions $\{A^{(n)}_x\}_{n=1}^\infty$, such that
\begin{enumerate}
    \item Each region contains the site $x$:  $x \in A_x^{{(n)}} $ for all $n$.
    \item Subsequent regions contain the earlier ones:  $A^{(n)}_x \subseteq A^{(m)}_x$ for $n < m$.  
    \item The distance from $x$ to the complement of $A_x^{(n)}$ diverges:  $\dist(x, \overline{A_x^{(n)}}) \rightarrow \infty$ for $n\rightarrow \infty$. \label{prop:diverging_distance}
\end{enumerate}
\end{definition}
Condition \ref{prop:diverging_distance} ensures that  $\bigcup_n A^{(n)}_x = \lim_{n \to \infty }A^{(n)}_x = \Lambda$, i.e., the sequence exhausts the entire system.

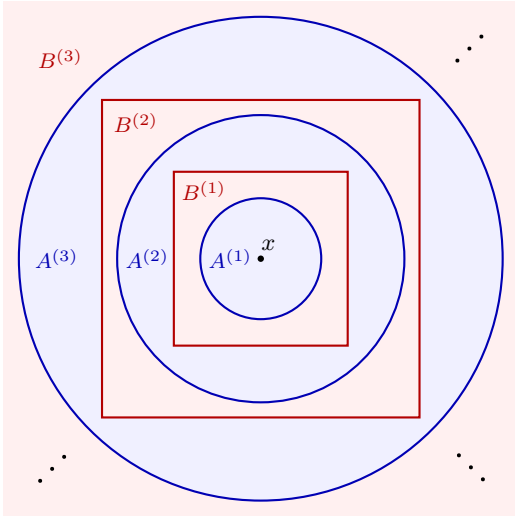
\begin{figure}[h!]
\label{fig:covering}
\centering
\begin{tikzpicture}[scale=1]
    \fill[red!6!]  (-3.4,-3.4) rectangle (3.4,3.4);
  \draw[thick, blue!70!black, fill = blue!6!] (0,0) circle (3.2);
    \draw[thick, red!70!black,  fill = red!6!]  (-2.1,-2.1) rectangle (2.1,2.1);
  \draw[thick, blue!70!black, fill = blue!6!] (0,0) circle (1.9);
   \draw[thick, red!70!black,  fill = red!6!] (-1.15,-1.15) rectangle (1.15,1.15);
  \draw[thick, blue!70!black, fill = blue!6!] (0,0) circle (0.8);

  \fill (0,0) circle (1.2pt);
  \node[font=\small, anchor=south ] at (0.1,0) {$x$};

  \node[blue!70!black, inner sep=1pt, font=\footnotesize] at (-0.4, 0) {$A^{(1)}$};
  \node[red!70!black,  inner sep=1pt, font=\footnotesize] at (-0.75, 0.9)   {$B^{(1)}$};
  \node[blue!70!black, inner sep=1pt, font=\footnotesize] at (-1.5, 0) {$A^{(2)}$};
  \node[red!70!black,  inner sep=1pt, font=\footnotesize] at (-1.65, 1.8)   {$B^{(2)}$};
  \node[blue!70!black, inner sep=1pt, font=\footnotesize] at (-2.7, 0) {$A^{(3)}$};
  \node[red!70!black,  inner sep=1pt, font=\footnotesize] at (-2.65, 2.65)   {$B^{(3)}$};
  \node[font=\Large, rotate = 45] at (2.8, 2.8) {$\cdots$};
  \node[font=\Large, rotate = -45] at (2.8, -2.8) {$\cdots$};
  \node[font=\Large, rotate = -135] at (-2.8, -2.8) {$\cdots$};
\end{tikzpicture}
\caption{Two coverings $\{A^{(n)}\}$ and $\{B^{(n)}\}$ centered at $x$.}
\end{figure}

Next, we introduce the local order parameter to diagnose SW-SSB.
\begin{definition}[Local SW-SSB measure]\label{def:local_SW-SSB_measure}
    Given a local charged operator $O_x$ and a covering $\{A^{(n)}_x\}_{n=1}^\infty$ centered around $x$, the (one-point fidelity) \textbf{local measure of SW-SSB} of a state $\rho$ is defined as
    \begin{align}
        F(\omega; O_x) = \lim_{n\to \infty } F(\rho_{A_x^{(n)}};  O_x) \equiv  \lim_{|A| \to \infty } F(\rho_A;  O_x).
    \end{align}
\end{definition}

Our central result in this Section is the existence of the above limit:
\begin{theorem}\label{thm:existence_and_uniqueness}
Let $\rho$ be a mixed state and $O_x$ a local operator. Then the limit
\begin{equation}
    \lim_{n \to \infty} F(\rho_{A_x^{(n)}}, O_x \rho_{A_x^{(n)}} O_x^\dagger) = F(\rho; O_x)
\end{equation}
exists for any covering $\{A_x^{(n)}\}_n$ centered at $x$, and is independent of the choice of covering.
\end{theorem}
\begin{proof}
We first show existence. By the data processing inequality applied to the partial trace $\Tr_{A_x^{(n+1)} \setminus A_x^{(n)}}$, the sequence
\begin{equation}\label{eq:monotonicity_1PF}
    F(\rho_{A_x^{(n)}}, O_x \rho_{A_x^{(n)}} O_x^\dagger) \geq F(\rho_{A_x^{(n+1)}}, O_x \rho_{A_x^{(n+1)}} O_x^\dagger)
\end{equation}
is monotonically decreasing. Since it is also bounded below by zero, $F(\rho, \sigma) \geq 0$, the limit exists.

We now prove the independence of the covering. Let $\{A_x^{(n)}\}_n$ and $\{B_x^{(m)}\}_m$ be two coverings centered at $x$. For each $n$, there exists $m_n$ such that $A_x^{(n)} \subseteq B_x^{(m_n)}$ (see Fig.~\ref{fig:covering}), by property~\ref{def:Convering_centered_around_i}. Moreover, since both coverings exhaust the system, we have $m_n \to \infty$ as $n \to \infty$.

Applying the data processing inequality, we get
\begin{equation}
    F(\rho_{A_x^{(n)}}, O_x \rho_{A_x^{(n)}} O_x^\dagger)
    \geq
    F(\rho_{B_x^{(m_n)}}, O_x \rho_{B_x^{(m_n)}} O_x^\dagger).
\end{equation}
Taking the limit $n \to \infty$ yields
\begin{equation}
    \lim_{n \to \infty} F(\rho_{A_x^{(n)}}; O_x) \geq
    \lim_{m \to \infty} F(\rho_{B_x^{(m)}}; O_x).
\end{equation}
By switching the coverings in the argument above, we get the equality. Therefore, the limit exists, and its value does not depend on the covering. 
\end{proof}

Note that the key property underlying the existence and uniqueness of the covering limit, $|A| \to \infty$, is the data processing inequality (DPI) for the fidelity. As a result, information-theoretic quantities that do not satisfy the DPI, such as the R\'enyi-2, cannot define a robust local SW-SSB measure (see Sec.~\ref{sec:alternative_definitions} for examples where the limit does not exist). Moreover, the DPI guarantees monotonicity under enlarging regions, so that fidelities computed on finite regions provide upper bounds on the local SW-SSB measure and can be used to approximate the limiting value (see also Sec.~\ref{sec:markov_length_symmetry-averaged}).

We also remark that the fidelity between states can be defined already for infinite-volume states~\cite{UHLMANN1976273}. In particular, the one-point fidelity $F(\omega, \omega_O)$, where $\omega_O(\cdot) = \omega_O(O_x^\dagger \cdot O_x)$, has a well-defined and finite real value\footnote{We thank Ruizhi Liu for pointing this out.}, and it can be shown to coincide with the local measure of SW-SSB $F(\omega; O_x)$ defined above through the local-region limiting procedure~\cite{alberti_note_1983}.

We can now sharply define the local notion of SW-SSB for infinite systems:

\begin{definition}[SW-SSB]\label{def:local_SW-SSB}
A state $\omega$ has \textbf{strong-to-weak symmetry breaking} (in the local sense) if the following two conditions are met:
\begin{enumerate}
    \item The (one-point fidelity) measure of SW-SSB is non-vanishing: 
    $F(\omega; O_x) > 0$ 
    \item The state does not exhibit ordinary SSB: $ \lim_{|x-y| \to \infty} \langle O_x O_y\rangle_\omega = 0$ 
\end{enumerate}
\end{definition}

While the covering independence ensures that $F(\omega; O_x)$ is well-defined, it will, in general, depend on the choice of operator. However, in the presence of weak translation symmetry, the LFC is independent of the position at which the operator is inserted, as we show next.

\begin{lemma}
If $\omega$ has weak translation symmetry, then
\begin{equation}
    F(\omega; O_x) = F(\omega; O_y), \qquad \forall y \in \Lambda .
\end{equation}
\end{lemma}
\begin{proof}
Let $T$ be a translation mapping $x$ to $y$, with unitary representation $U_T$, so that $O_y = U_T O_x U_T^\dagger$. Let $A \subset \Lambda$ be a finite region containing $x$, and denote by $TA$ the corresponding translated region.

The translation induces isomorphism between the local operators on $A$ and local operators on $TA$, defined by
\begin{equation}
    \mathcal{U}_{T,A}(O_A) = U_T O_A U_T^\dagger,
\end{equation}
under which the fidelity is invariant:
\begin{equation}
F(\rho_A, O_x \rho_A O_x^\dagger)
= F(\mathcal{U}_{T,A}(\rho_A), O_y \mathcal{U}_{T,A}(\rho_A) O_y^\dagger).
\end{equation}

Using weak translation symmetry, one verifies that reduced density matrices transform as,
\begin{equation}
    \rho_{TA} = \mathcal{U}_{T,A}(\rho_A),
\end{equation}
which implies
\begin{equation}
    F(\rho_A, O_x \rho_A O_x^\dagger) = F(\rho_{TA}, O_y \rho_{TA} O_y^\dagger).
\end{equation}
Finally, for any covering centered at $x$, $\{A_x^{(n)}\}$, the translated covering, $\{TA_x^{(n)}\}$, defines a covering centered at $y$, and taking the covering limit yields the result.
\end{proof}

Despite the difficulties in defining notions like strong symmetry and global CMI (both require knowing the full density matrix $\rho$ on finite systems) on infinite systems, the local notion of SW-SSB remains well-defined. Moreover, the stability theorem discussed in Sec.~\ref{sec:stability} remains valid on infinite systems: strong symmetry on local channels is perfectly well-defined even on infinite systems, and the proof in Sec.~\ref{sec:stability} only requires the region size $\ell_A$ to be larger than the light-cone size of the channel $\mathcal{E}$. Therefore the local SW-SSB, even on infinite systems, is a universal property robust against strongly symmetric finite-depth channels. 

\subsection{Markov length and symmetry-averaged CMI}\label{sec:markov_length_symmetry-averaged}

In the last section, we showed that the LFC $F(\rho_A; O_x)$ is monotonically non-increasing as the region $A$ gets larger, eventually stabilizing at the limiting value $\lim_{|A| \to \infty} F(\rho_A; O_x) = F(\omega; O_x)$. One might naturally ask how big the region $A$ has to be so that $F(\rho_A; O_x) \approx F(\omega; O_x)$. This is particularly important to the numerical task of estimating local SW-SSB through the local fidelity correlator $F(\rho_A; O_x)$, which can be exponentially hard in the size of $|A|$. Here, we show that the Markov length $\xi_M$ --- the length scale associated with the decay of CMI --- sets an upper bound to the convergence scale of local SW-SSB. The main result is as follows, and we quote from \cite{li_unified_2026}:

\begin{theorem}[Theorem 9' of \cite{li_unified_2026}]\label{thm:local_computability}
    For $\rho \equiv \rho_{ABC}$ a state over a finite tripartioned system $ABC$ and $O_A$ a local operator supported on $A$, we have
    \begin{equation}
        F(\rho_{AB}; O_A) - F(\rho; O_A) \leq c \norm{O_A}_\infty I({A:C}|B)_\rho^{1/2},
    \end{equation}
    where $c$ is a numerical constant.
\end{theorem}

The intuition behind this result is that if the CMI $I({A:C}|B)$ is small, then the information of the state on $C$ can be approximately recovered from $B$~\cite{fawzi_quantum_2015}, which does not alter the local fidelity correlator $F(\rho_{AB}; O_A)$.

In the notation of the theorem above, $A$ can be as small as the support of $O_A$, which in many cases is a single site. Keeping $B$ fixed and taking the limit as $C$ covers the whole system, we have

\begin{corollary}
    For $\omega$ a state in an infinite-volume tripartitioned system $ABC$ with $AB$ finite, and $O_A$ a local operator supported on $A$, we have
    \begin{equation}
        F(\rho_{AB};O_A) - F(\omega;O_A) \leq c \norm{O_A}_\infty I({A:C}|B)_\omega^{1/2}, 
    \end{equation}
    where the \textbf{local CMI} $I({A:C}|B)_\omega$ is defined as
    \begin{equation}
        I({A:C}|B)_\omega \defeq \lim_{n\to \infty} I({A:C}^{(n)}|B)_{\rho_{ABC^{(n)}}},    
    \end{equation}
    with $ABC^{(n)}$ a covering containing $AB$ (See Fig.~\ref{fig:local_computability}).
\end{corollary}

\begin{figure}[t]
    \centering
    \begin{tikzpicture}[ scale=0.7, boundary/.style={draw=black!70, line width=0.8pt}, labelstyle/.style={font=\large}]
        
        \def\rA{0.4}
        \def\rAB{\rA+1}
        \def\rCi{(\rAB+0.3)}
        \def\rCf{(\rCi+1)}

        \foreach \step in {0,...,3} {
            \pgfmathsetmacro{\r}{\rCf - \step * (\rCf-\rCi)/3}
            \fill[green!20, opacity=0.5, rounded corners=3em] (-\r,-\r) rectangle (\r, \r);
        }
        
        \draw[boundary, fill=black!10] (0,0) circle (\rAB);
        \draw[boundary, fill=blue!12] (0,0) circle (\rA);
        
        \node[labelstyle] at (0,0) {$A$};
        \node[labelstyle] at (-{(\rAB + \rA)/2 + 0.3},{(\rAB + \rA)/2 - 0.3}) {$B$};
        \node[labelstyle] at (-{(\rCi + \rCf)/2 + 0.7},{(\rCi + \rCf)/2 - 0.7}) {$C^{(n)}$};
        
        \node[rotate=45, font=\Large] at ({0.75*\rCf}, {0.75*\rCf}) {$\cdots$};
        \node[rotate=-45, font=\Large] at ({0.75*\rCf}, {-0.75*\rCf}) {$\cdots$};
        \node[rotate=-135, font=\Large] at ({-0.75*\rCf}, {-0.75*\rCf}) {$\cdots$};

        \draw[<->, black, thick] (\rA,0) -- (\rAB,0) node[midway, below, font=\normalsize] {$\ell$};

        \begin{scope}[xshift=4cm, yshift=-1.75cm] 
            \begin{axis}[
                width=6cm, 
                height=5cm, 
                xmin=0, xmax=10, 
                ymin=0, ymax=1, 
                axis lines=left, 
                xlabel={\Large $\ell$}, 
                xtick=\empty, 
                ytick=\empty, 
                clip=false, 
                very thick,
            ]
            
            \addplot[thick, BlueViolet, domain=0:10, samples=100] {0.35 + 0.55*exp(-x/2)};

            \draw[<->, black, thick] (axis cs:0,0.52) -- (axis cs:2.2,0.52) node[midway, below, font=\Large] {$\xi_M$};

            \node[BlueViolet!80!black, font=\Large, anchor=south west] at (axis cs:3.2, 0.45)
            {$F(\rho_{AB};O_A)$};

            \end{axis}
        \end{scope}
    \end{tikzpicture}
\caption{The tripartitions used in defining the local CMI $I({A:C}|B)_\omega = \lim_{n \to \infty} I({A:C}^{(n)}|B)$ (left), whose characteristic length scale $\xi_M$, the Markov length, dictates the convergence behavior of the LFC $F(\rho_{AB};O_A)$ in $\ell = \dist(A,C)$ (right).}
    \label{fig:local_computability}
\end{figure}
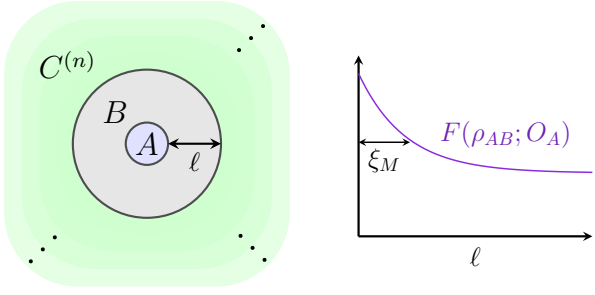

The $|C| \to \infty$ limit in the local CMI exists because $I(A:C|B)$ is monotonically non-decreasing as $C$ grows, and upper bounded by $2 \log \dim(\Hilb_A)$, assumed small. Moreover, the limit is independent of the covering $ABC^{(n)}$ for the same reason as $F(\omega; O_A)$ is (see proof of Theorem~\ref{thm:existence_and_uniqueness}). 

In the particular case where the local CMI decays exponentially as $I({A:C}|B)_\omega = O(\exp({-\ell/ \xi_M}))$, where $\ell \defeq \dist(A,C)$ is the width of $B$ and $\xi_M$ is the Markov length~\cite{sang2024stability}, the LFC $F(\rho_{AB}; O_A)$ will be exponentially close to its limiting value $F(\omega; O_A)$ for $\ell \gg \xi_M$ (See Fig.~\ref{fig:local_computability}). 

Even though the local CMI $I({{A:C}}|B)_\omega$ is a probe of conditional correlations of the infinite-volume state $\omega$, it is defined through the local density matrices by construction. This means that the long-range CMI coming from the strong symmetry constraint of finite-volume SW-SSB states, as discussed in Sec.~\ref{sec:long-range_CMI}, is not shared by $I({A:C}|B)_\omega$. Despite this, we now show a local counterpart to the long-range CMI result that is well-defined in the thermodynamic limit by reintroducing the strong symmetry condition on the local state. To do so, we restrict ourselves to on-site symmetry representations $U = \prod_i u_i$ of Abelian groups $G$ and states $\omega$ that are weakly symmetric under $U$. For those symmetries, the reduced density matrix $\rho_{R}$ on a finite region $R$ is also weakly symmetric --- now under $U_{R} = \prod_{i \in R} u_i$ --- and thus block diagonal in the symmetry basis. In other terms, it admits the isotypic decomposition $\rho_{R} = \sum_\lambda p_{R}^{(\lambda)} \rho_{R}^{(\lambda)}$, where $p_{R}^{(\lambda)} = \Tr[\rho_R \Pi_R^{(\lambda)}]$ and $\rho_{R}^{(\lambda)} = \Pi_R^{(\lambda)} \rho_{R} / p_R^{(\lambda)}$, with $\Pi_R^{(\lambda)}$ being the projector onto the irrep subspace of $\lambda : G \to U(1)$. Each isotypic component $\rho_{R}^{(\lambda)}$ is strongly symmetric under $U_{R}$, as $U_{R}(g) \rho_{R}^{(\lambda)} = \lambda(g) \rho_R^{(\lambda)}$, and is then amenable to the previous analysis. In particular, under an additional assumption of local indistinguishability, we have
\begin{theorem}[Long-range symmetry-averaged CMI]\label{thm:long-range_symmetry-averaged_CMI}
    Let $\omega$ be weakly symmetric under an on-site Abelian symmetry representation $U = \prod_i u_i$ and $AB$ a finite region. If, for all irreps $\lambda$, its isotypical component $\rho_{ABC}^{(\lambda)} $ is locally indistinguishable to $\rho_{AB}$ for sufficiently large $C$, i.e. $\lim_{n \to \infty} \norm{\Tr_{C^{(n)}}[\rho_{ABC^{(n)}}^{(\lambda)}] - \rho_{AB}}_1 = 0$, then
    \begin{equation}
        I({A:C}|B)^{\text{Sym}}_{\omega} \geq \left( \frac{F(\rho_{AB};O_A)}{c \norm{O_A}_\infty} \right)^2,
    \end{equation}
    where the \textbf{symmetry-averaged local CMI} $I({A:C}|B)^{\text{Sym}}_{\omega}$ is defined as
    \begin{equation}
        I({A:C}|B)^{\text{Sym}}_{\omega} \defeq \lim_{n\to \infty}\sum_{\lambda} p_{ABC^{(n)}}^{(\lambda)} I({A:C}^{(n)}|B)_{\rho_{ABC^{(n)}}^{(\lambda)}}.
    \end{equation}
\end{theorem}

\begin{proof}
    First, let us prove that the limit defining $I({A:C}|B)^{\text{Sym}}_\omega$ exists and is independent of the covering. For the existence, note that this quantity equals the CMI $I({A:C}^{(n)}|B,Q_n)_{\sigma_n}$ of the quantum-classical state 
    \begin{equation}
        \sigma_n \defeq \sum_\lambda p_{ABC^{(n)}}^{(\lambda)} \rho_{ABC^{(n)}}^{(\lambda)} \otimes \op{\lambda}_{Q_n}.    
    \end{equation}
    In words, $\sigma_n$ arises from measuring the total charge (i.e., irrep) $\lambda$ in the region $ABC^{(n)}$ and recording the result in the classical register $Q_n$. From the chain rule of the conditional mutual information, we have
    \begin{equation}
        I({A:C}^{(n)}|B Q_n) = I({A:C}^{(n)} Q_n|B) - I(A:Q_n|B).
    \end{equation}
    We evaluate the two terms on the right-hand side. For the first term, consider nested regions $C^{(m)} \supset C^{(n)}$. For Abelian symmetries, the charge $\lambda_n$ of the region $ABC^{(n)}$ is related to the charges $\lambda'$ of the region $\Delta C \defeq C^{(n)} \setminus C^{(m)}$ and $\lambda_m$ as $\lambda_n = \lambda_m / \lambda'$. In that way, the joint system $C^{(n)} Q_n$ can thus be obtained from $C^{(m)} Q_m$ via a local CPTP map: measure $\lambda'$ on $\Delta C$, compute $\lambda_n$ and store the result on a new register $Q_n$, and trace out $Q_m$ and $\Delta C$. Because this map acts strictly on the target system of the mutual information, the data processing inequality yields $I({A:C}^{(m)} Q_m|B) \geq I({A:C}^{(n)} Q_n|B)$. Being monotonically non-decreasing and bounded by $2 \log \dim(\Hilb_A)$, this term converges to a covering-independent limit by the same arguments of Theorem \ref{thm:existence_and_uniqueness}. 

    For the second term, the local indistinguishability condition $\lim_{n \to \infty} \norm{\rho_{AB, n}^{(\lambda)} - \rho_{AB}}_1 = 0$, where $\rho_{AB,n}^{(\lambda)} \defeq \Tr_{C^{(n)}} [\rho_{ABC^{(n)}}^{(\lambda)}]$, dictates that the local state on $AB$ becomes independent of the global charge $Q_n$. Hence, $\lim_{n\to \infty}\Tr_{C^{(n)}} [\sigma_n] \propto \rho_{AB} \otimes \one_{Q_n}$, which implies $I(A:Q_n|B) \to 0$. Thus, the limit of the symmetry-averaged CMI $I({A:C}^{(n)}|B Q_n)$ exists and is uniquely defined.

    With the limit rigorously established, we apply the finite-volume bound. Since each $\rho_{ABC^{(n)}}^{(\lambda)}$ is strongly symmetric, Theorem \ref{thm:LRCMI} yields 
    \begin{equation}
        I({A:C}^{(n)}|B)_{\rho^{(\lambda)}_{ABC^{(n)}}} \geq \left( \frac{F(\rho_{AB,n}^{(\lambda)} ; O_A)}{c \norm{O_A}_\infty} \right)^2.    
    \end{equation}
    Averaging over all irreps $\lambda$, we obtain
    \begin{equation}
        I({A:C}^{(n)}|B Q_n) \geq \sum_{\lambda} p_{ABC^{(n)}}^{(\lambda)} \left( \frac{F(\rho_{AB,n}^{(\lambda)} ; O_A)}{c \norm{O_A}_\infty} \right)^2.
    \end{equation}
    Taking the limit $n \to \infty$, local indistinguishability guarantees $F(\rho_{AB,\lambda}^{(n)}; O_A) \to F(\rho_{AB}; O_A)$ by the continuity of fidelity. Since $\sum_\lambda p_{ABC^{(n)},\lambda} = 1$, the right-hand side converges exactly to the desired lower bound, completing the proof.
\end{proof}

As seen in the proof of the above theorem, the symmetry-averaged local CMI $I({{A:C}}|B)^{\text{Sym}}_\omega$ can be interpreted as the CMI when the symmetry charge of increasingly large regions is known and conditioned upon. Similarly to the finite volume case of Theorem~\ref{thm:LRCMI}, it is also long-range if $\omega$ exhibits local SWSSB, with the crucial difference being that it is well defined in the thermodynamic limit. For that, we needed the extra local indistinguishability assumption, which we expect to hold for the physically motivated scenarios where local observables are oblivious to global strong symmetry constraints. For Gibbs states of local Hamiltonians, for example, this was proven at sufficiently high temperatures~\cite{negari_symmetry_2025}. 

\subsection{Order-disorder inequalities}
\label{sec:order-disorder}

The quantum fidelity is the minimum over all measurements (POVMs) of the classical fidelity between the measured probability distributions~\cite{nielsen_quantum_2010}. This implies
\begin{align}
    F(\rho_A; O_x) & = \min_{\{ E^A_m\}} \sum_m \sqrt{\Tr[E^A_m \rho_A] \Tr[E^A_m O_x \rho_A O_x^\dagger]} \\
    & = \min_{\{ E^A_m\}} \sum_m \sqrt{\langle E_m^A\rangle_\omega \langle O_x^\dagger E^A_m O_x\rangle_\omega},
\end{align}
where the minimum is over all POVMs $E_m^A \geq 0$, $\sum_m E_m^A = \one$, supported on $A$. 

The conceptual advantage of this reexpression is that it describes the one-point fidelity directly in terms of observable quantities: the measurement probabilities 
\begin{align}
    p_{m} & = \Tr[E_m^A \rho_A] = \omega(E_m^A),
    \\ p_{m,O_x} & = \Tr[E_m^A O_x \rho_A O_x^\dagger] = \omega(O_x^\dagger E_m^A O_x),
\end{align}
which in turn makes it an explicit function of the global state $\omega$. This allows us to take the thermodynamic limit $|A| \to \infty$, arriving at
\begin{align}
    F(\omega; O_x) & = \inf_{\{ E_m\}} \sum_m \sqrt{\langle E_m\rangle_\omega \langle O_x^\dagger E_m O_x\rangle_\omega},
\end{align}
where now the infimum is over all (quasi-)local POVMs. 

In the case where $O_x$ is unitary\footnote{As explained in Appendix~\ref{app:unitary_OP}, the assumption of a unitary order parameter is not a strong one.}, $p_{m,O}$ can be viewed as the probabilities of measuring the charge-flipped POVM $E_{m,O_x}^A \defeq O_x^\dagger E_m^A O_x$ on $\omega$. This motivates a natural ansatz for a distinguishing POVM: measuring the local charge of $A$. More precisely, for an on-site symmetry and a finite region $A$, we can take $E_\lambda = \Pi_A^{(\lambda)}$, the projector onto the irrep (i.e., charge) sector $\lambda$ of region $A$. Suppose, for simplicity, an Abelian symmetry group $G$ and let $\mu : G \to U(1)$ be the charge of the \textit{unitary} charged operator $O_x$ (so $U_g O_x U_g^\dagger = \mu_g O_x$). Then, we have
\begin{align}\label{eq:order-disorder_general}
    F(\rho_A; O_x) \leq \sum_\lambda \sqrt{p_A^{(\lambda)} p_{A}^{(\lambda \mu^*)}},
\end{align}
where $p_A^{(\lambda)} = \langle \Pi_A^{(\lambda)}\rangle = \frac{1}{|G|} \sum_{g \in G} \lambda_g^* \langle U_{A,g} \rangle$. Since the right-hand side (RHS) depends only on the expectation value of the disorder operators $U_{A,g} = \prod_{i \in A} u_{i,g}$, we call Eq. \eqref{eq:order-disorder_general} an \textit{order-disorder inequality}. Moreover, the RHS is a faithful measure of strong asymmetry, since it is zero if, and only if, $\rho_A$ is strongly symmetric, and it is invariant under any strongly symmetric CPTP map acting on region $A$, since it only depends on the probability weights $p_A^{(\lambda)}$. These properties qualify the RHS as a proper monotone in the resource theory of strong asymmetry described in Ref.~\cite{kusuki_resourcetheoretic_2026}.

For the simple case of $G = \Z_2$ generated by $U = \prod_i u_i$, this order-disorder inequality simplifies to 
\begin{equation}\label{eq:order-disorder_Z2}
    F(\rho_A; O_x)^2 + |\langle U_{A}\rangle|^2 \leq 1,
\end{equation}
where it is now clear the tradeoff between the local fidelity order parameter and the linear disorder parameter. Due to the monotonicity of the LFC, we also have
\begin{equation}
    F(\omega, O_x)^2 + \sup_{A \supseteq \supp(O)}|\langle U_A \rangle|^2 \leq 1.
\end{equation}
For many states, the disorder operator decreases as $A$ increases, in which case the supremum is saturated at the smallest $A = \supp(O)$. For example, in the case of the $\Z_2$ spin flip symmetry $U = \prod_i X_i$ and $O_x = Z_x$, it reduces to
\begin{equation}
    F(\omega; Z_i) \leq \sqrt{1 - \langle X_{x}\rangle^ 2}.
\end{equation}

These inequalities derived from taking the particular charge-measurement POVM are not necessarily tight. Indeed, at transitions in or out of a SW-SSB phase, the LFC is oftentimes singular at the transition point, even though the disorder operator $\langle U_A \rangle$ --- being a linear observable --- is not. In some simple cases, however, these bounds can be saturated. One such case is the Gibbs state of a quantum paramagnet, $\rho_\beta = \exp(\beta \sum_i X_i)$. The LFC for such state is calculated in Appendix \ref{app:paramagnet} and it satisfies
\begin{equation}
    F(\rho_{\beta, A}; Z_x)^2 + |\langle X_x\rangle|^2 = \frac{1}{\cosh^2(\beta)} + \tanh^2(\beta) = 1.
\end{equation}

It is instructive to compare the order-disorder inequalities above to the analogous constraints for ground states of $1d$ quantum spin chains, in particular the results of Ref.~\cite{levin_constraints_2020a}. On one side, our bounds hold for any mixed state, whereas Ref.~\cite{levin_constraints_2020a} is restricted to gapped ground states. On the other side, this limited scope allows for a much stronger and direct tradeoff between order and disorder parameters: one or the other is present. Such a relation is not implied here.

\section{R\'enyi correlators}
\label{sec:alternative_definitions}

The global SW-SSB can be characterized in several different ways, some of which (such as the R\'enyi-$1$ correlator~\cite{liu2025diagnosing,weinstein2025efficient}) are fully equivalent to the fidelity correlator. This is also true for the local SW-SSB. In this section, we discuss a natural family is provided by one-point R\'enyi correlators:
\begin{definition}
Given a local charged operator $O_x$ and a finite region $A$ containing $x$, the one-point local R\'enyi-$\alpha$ correlator of a state is defined, for $\alpha>0$, by
\begin{equation}
    R^{(\alpha)}(\rho_A;O_x)  =
    \frac{\Tr(\rho_A^{\alpha/2} O_x \rho_A^{\alpha/2} O_x^\dagger)}
         {\Tr(\rho_A^\alpha)} .
\end{equation}
\end{definition}
Again we are interested in its behavior as $|A|$ grows. On an infinite system, similar to the discussions in Sec.~\ref{sec:infinite}, one can define the $|A|\to\infty$ limit by considering a covering centered at $x$:
\begin{equation}
    R^{(\alpha)}(\omega;O_x)=\lim_{n\to\infty} R^{(\alpha)}(\rho_{A_x^{(n)}};O_x)\equiv\lim_{|A|\to\infty}R^{(\alpha)}(\rho_A;O_x).
\end{equation}

The corresponding two-point R\'enyi correlators have been widely used in the literature. When $\alpha$ is an even integer, they are particularly prone to replica calculations since only integer powers of $\rho$ are required, which makes them accessible both analytically and numerically. The most prominent example is the R\'enyi-$2$ correlator~\cite{LeeJianXu,Ma2024SPTdoubled}, which admits a simple interpretation as an ordinary two-point function of the vectorized state $\kett{\rho}$ in the doubled Hilbert-space formalism. Another distinguished case is the two-point R\'enyi-$1$ correlator, which admits a natural representation in terms of ordinary two-point functions in the canonical purification $\kett{\sqrt{\rho}}$, and its long-distance behavior is equivalent to that of the fidelity correlator.

For $\alpha\neq 1$, R\'enyi correlators do not in general provide an intrinsic definition of SW-SSB. The underlying obstruction is that they do not satisfy the data processing inequality. As a consequence, they are not robust under strongly symmetric finite-depth deformations and can therefore appear to distinguish states that belong to the same mixed-state phase. On an infinite system, the same lack of the DPI leads to a more basic problem: the covering limit, $|A| \to \infty$ defining $R^{(\alpha)}(\omega;O_x)$, depends on the choice of covering $\{A^{(n)}\}$ and hence is not well defined.

This can be illustrated by a simple dimer example in a one-dimensional spin-$1/2$ chain. Consider
\begin{equation}
     \rho = \frac{1}{2} |\psi_{\rm e} \rangle \langle \psi_{\rm e}| + \frac{1}{2}|\psi_{\rm o} \rangle \langle \psi_{\rm o}|, 
\end{equation}
where $|\psi_{\rm e}\rangle$ and $|\psi_{\rm o}\rangle$ are dimerized states made of singlets between nearest-neighbors, that is
\newcommand{\dimerchain}[1]{
\begin{tikzpicture}[baseline=-0.5ex]
    \foreach \x in {1,...,8} {
        \fill (\x*0.4,0) circle (1.2pt) node (n\x) {};
    }
    \ifnum#1=0
        \foreach \i in {1,3,5,7} {
            \pgfmathtruncatemacro{\j}{\i+1}
            \draw[rounded corners=2.5pt, line width=0.4pt] 
                (\i*0.4-0.12, -0.12) rectangle (\j*0.4+0.12, 0.12);
        }
    \else
        \draw[rounded corners=2.5pt, line width=0.4pt] 
                (0.4-0.12, 0.12) -- (0.4+0.12, 0.12) -- (0.4+0.12, -0.12) -- (0.4-0.12, -0.12);
        \foreach \i in {2,4,6} {
            \pgfmathtruncatemacro{\j}{\i+1}
            \draw[rounded corners=2.5pt, line width=0.4pt] 
                (\i*0.4-0.12, -0.12) rectangle (\j*0.4+0.12, 0.12);
        }
        \draw[rounded corners=2.5pt, line width=0.4pt] 
                (0.4*8+0.12, 0.12) -- (0.4*8-0.12, 0.12) -- (0.4*8-0.12, -0.12) -- (0.4*8+0.12, -0.12);
    \fi
    \pgfresetboundingbox
    \path (1*0.4-0.15, -0.15) rectangle (8*0.4+0.15, 0.15);
\end{tikzpicture}
}
\begin{equation}
\begin{aligned}
    \ket{\psi_{\rm e}} &= |\cdots \dimerchain0 \cdots \rangle,\\
    \ket{\psi_{\rm o}}  &= |\cdots \dimerchain1 \cdots\rangle,
\end{aligned}
\end{equation}
where each obround denotes the singlet $\frac{1}{\sqrt{2}}\left(\ket{\uparrow\downarrow}-\ket{\downarrow\uparrow}\right)$.
Now, consider the two coverings
\begin{equation}
    A^{(n)}=[-2n,2n+1],
    \qquad
    B^{(n)}=[-2n-1,2n].
\end{equation}
The interval $A^{(n)}$ respects the even dimers but cuts the odd dimers at its endpoints, so that
\begin{equation}
    \rho_{A^{(n)}}
    =
    \frac{1}{2}\ket{\psi_{\rm e}}\bra{\psi_{\rm e}}_{A^{(n)}}
    +
    \frac{1}{2} \cdot 
    \frac{\one}{2}\otimes
    \ket{\psi_{\rm o}}\bra{\psi_{\rm o}}_{B^{(n-1)}}
    \otimes \frac{\one}{2}.
\end{equation}
Similarly, $B_N$ respects $\ket{\psi_{\rm o}}$ but cuts the even dimers at the endpoints.

Then, the one-point local R\'enyi-$\alpha$ for the local operator $X_0X_1$ depends on the covering
\begin{equation}
\begin{aligned}
    \lim_{n\to \infty} R^{(\alpha)}(\rho_{A^{(n)}}; X_0 X_1) &= \frac{4^{\alpha-1}}{4^{\alpha-1} + 1}, \\
     \lim_{n\to \infty} R^{(\alpha)}(\rho_{B^{(n)}}; X_0 X_1) &= \frac{1}{4^{\alpha-1} + 1}
\end{aligned}
\end{equation}
Hence, the covering limit is not well-defined for $\alpha\neq 1$. Depending on the sequence of regions used to approach the thermodynamic limit, the limit may fail to exist or may converge to different values.

This example makes clear why, despite their usefulness as computational probes, higher-R\'enyi correlators cannot be viewed as intrinsic local order parameters for SW-SSB. In contrast, $R^{(1)}$ is distinguished: it is monotone under enlarging regions and therefore admits a well-defined covering limit. The following result establishes this and its equivalence to the one-point fidelity as a local order parameter for SW-SSB.
\begin{lemma}\label{lem:R1_local_general}
Let $\rho$ be a mixed state and $O_x$ a bounded local operator. Then the following hold
\begin{enumerate}
    \item If $A \subset B$ are finite regions containing the support of $O_x$, 
    \begin{equation}\label{eq:R1_is_monotone}
        R^{(1)}(\rho_A;O_x)\ge R^{(1)}(\rho_B;O_x).
    \end{equation}
    \item For any finite region $A$ containing $O_x$,
    \begin{equation}\label{eq:R1_fidelity_sandwich_finite}
        F(\rho_A;O_x)^2 \le R^{(1)}(\rho_A;O_x) \le\norm{O_x}_\infty F(\rho_A;O_x).
    \end{equation}
\end{enumerate}
\end{lemma}
\begin{proof}
To establish monotonicity, let $C = B\setminus A$, and let $D_C=\dim \mathcal H_C$, the dimension of the Hilbert space of $C$. If $O_x$ is unitary, the monotonicity follows from the standard data processing inequality of the Holevo fidelity $\Tr(\sqrt{\rho}\sqrt{\sigma})$. For general $O_x$, we can relate the reduced density matrices via
\begin{equation}
    \int_{\rm Haar}\hspace{-1.3em} dU_C ~(\mathbbm{1}_A\otimes U_C)\rho_B(\mathbbm{1}_A\otimes U_C^\dagger)
    = \rho_A \otimes \frac{\mathbbm{1}_C}{D_C},
\end{equation}where the integral is over the unitary group on $\mathcal H_C$. Then, we have
\begin{equation}
\begin{aligned}
    R^{(1)}(\rho_A;O_x) &=  R^{(1)}\left(\rho_A\otimes \frac{\mathbbm{1}_C}{D_C}; O_x\right) \\
    &\ge \int_{\rm Haar}\hspace{-1.3em}dU_C~ R^{(1)}\left((\mathbbm{1}_A\otimes U_C)\rho_B(\mathbbm{1}_A\otimes U_C^\dagger) ; O_x\right)\\
    &= R^{(1)}(\rho_B;O_x),
\end{aligned}
\end{equation}where in the second line we used the concavity property, 
$R^{(1)}\left(\sum_n p_n\rho_n;O_x\right)\ge \sum_n p_n R^{(1)}(\rho_n;O_x)$, and in the last line we used that $O_x$ acts trivially on $C = B \setminus A$. 

The bounds in Eq.~\eqref{eq:R1_fidelity_sandwich_finite} follow by a direct adaptation of the (generalized) H\"older-inequality argument of Ref.~\cite{liu2025diagnosing} to the one-point correlator.

\end{proof}
\begin{corollary}
Let $O_x$ be a bounded operator. Then, on an infinite system, the limit
\begin{equation}
    R^{(1)}(\omega;O_x) = \lim_{|A|\to\infty} R^{(1)}(\rho_A;O_x)
\end{equation}
exists for any covering centered at $x$, and is independent of the choice of covering. Moreover,
\begin{equation}\label{eq:R1_fidelity_sandwich_infinite}
    F(\omega;O_x)^2 \le R^{(1)}(\omega;O_x) \le \norm{O_x}_\infty F(\omega;O_x).
\end{equation}
In particular,
\begin{equation}
    R^{(1)}(\omega;O_x)>0 \Longleftrightarrow F(\omega;O_x)>0.
\end{equation}
Therefore, the local definition of SW-SSB based on the one-point
R\'enyi-$1$ measure is equivalent to the definition based on the one-point
fidelity.
\end{corollary}
\begin{proof}
Since $R^{(1)}(\rho_A, O_x) \geq 0$, existence and covering-independence follow from the monotonicity exactly as in the one-point fidelity case, Theorem~\ref{thm:existence_and_uniqueness}. Taking the covering limit of Eq.~\eqref{eq:R1_fidelity_sandwich_finite} gives Eq.~\eqref{eq:R1_fidelity_sandwich_infinite}.
\end{proof}

\section{Local two-point fidelity correlator} \label{sec:equivalence_1P_2P}

The one-point local fidelity Eq.~\eqref{eq:LFC} appears to be quite different from the standard fidelity correlator Eq.~\eqref{eq:conventional_fidelity}, which is a two-point object. Physically, in the one-point LFC the boundary of the subregion $\partial{A}$ plays a role somewhat similar to $O_y$ in the standard two-point fidelity correlator. More precisely, $\partial{A}$ can be viewed as a defect that can break the strong symmetry, and the one-point LFC can be viewed as a correlation function between the defect $\partial A$ and a local operator $O_x$ --- this viewpoint will become apparent in our later discussion on several examples in Sec.~\ref{sec:examples}. 

One may wonder if we can define a local fidelity order parameter that more closely resembles the standard two-point fidelity correlator. To this end, we define the two-point local fidelity correlator (two-point LFC):
\begin{equation}
    F(\rho_A;O_x O_y^{\dagger}),
\end{equation}
where $O_{x,y}$ are charged local operators supported near $x,y\in A$. If $A$ becomes the entire system $\Lambda$, we recover the standard two-point fidelity correlator. However, in the spirit of this work, we keep $\ell_A\ll L$. In other words, $L\to\infty$ is taken before the $\ell_A\to\infty$ limit. We can then define SW-SSB, in the local two-point sense, if the state satisfies
\begin{equation}
    \lim_{|x-y|\to\infty}\lim_{|A|\to\infty}F(\rho_A;O_x O_y^{\dagger})\equiv \lim_{|x-y|\to\infty}F(\omega;O_xO_y^\dagger )>0.
\end{equation}
Here, on an infinite system the $|A|\to\infty$ limit is taken with a sequence of regions $A_n$ that eventually covers the whole system, as discussed in Sec.~\ref{sec:infinite}.

\subsection{Equivalence to local one-point fidelity correlator}

We now show that the above notion of SW-SSB in terms of two-point LFC is actually equivalent to the one-point local SW-SSB discussed in earlier Sections. 

First, applying Theorem \ref{thm:inequality_of_measures} (essentially data processing inequality), we have
\begin{equation}\label{eq:first_inequality}
  \norm{O_y}_\infty F(\rho_A;O_x) \geq F(\rho_{AB}; O_x O_y^\dagger),
\end{equation}
where $B$ is another region and $y\in B$. This shows that local SW-SSB in the two-point sense implies local SW-SSB in the one-point sense.

We then show that local SW-SSB in the one-point fidelity sense on average implies SW-SSB in the two-point fidelity sense on average.

\begin{theorem}\label{theorem:average_one-point_two-point}
    Let $A$ be a finite region, and $O \equiv \{ O_x \}_{x \in \Omega}$ a collection of charged operators in $A$. Then,
    \begin{equation}
        \frac{1}{|\Omega|^2} \sum_{x,y \in \Omega} F(\rho_A;O_x O_y^\dagger) \geq \left( \frac{1}{|\Omega|} \sum_{x \in \Omega} F(\rho_A;O_x) \right)^2
    \end{equation}
\end{theorem}
\begin{proof}
    Let $\ket{\Psi} \in \Hilb_A \otimes \Hilb'_A$ be any purification of the reduced density matrix of the state in $A$, $\rho_A$, with $\dim(\Hilb_A') \geq \dim(\Hilb_A)$. By Uhlmann's theorem applied to $F(\rho_A, O_x)$, for each site $x \in \Omega$ there exists a unitary $U_x \in \mathcal{U}(\Hilb_A')$ acting on the ancilla space of the purification satisfying
    \begin{equation}\label{eq:one-point_fid_Psi_Phi}
        F(\rho_A, O_x) = \braket{\Psi | O_x \otimes U_x | \Psi} = \langle \Phi_x | \Psi \rangle,
    \end{equation}
    where we have defined the (unnormalized) vector $\ket{\Phi_x} \defeq O_x^\dagger \otimes U_x^\dagger \ket{\Psi}$. Then, by Uhlmann's theorem applied to $F(\rho_A, O_x O_y^\dagger)$, we have
    \begin{equation}
    \begin{aligned}
        F(\rho_A, O_x O_y^\dagger) & = \max_{U \in \mathcal{U}(\Hilb'_A)} |\langle\Psi | O_x O_y^\dagger \otimes U | \Psi \rangle| \\
        & \geq \langle \Psi | O_x O_y^\dagger \otimes U_x U_y^\dagger | \Psi \rangle \label{eq:max_U} \\
        & = \langle \Phi_x | \Phi_y \rangle,
    \end{aligned}
    \end{equation}
    where, in the second line, we chose $U \propto U_x U_y^\dagger$ so that $\langle \Phi_x | \Phi_y\rangle  \geq 0$. Averaging over $x, y \in \Omega$ in the final equation above, we have
    \begin{align}
        \frac{1}{|\Omega|^2} \sum_{x,y \in \Omega} F(\rho_A, O_x O_y^\dagger) & \geq \frac{1}{|\Omega|^2} \sum_{x,y \in \Omega} \langle \Phi_x | \Phi_y \rangle = \langle \Phi |  \Phi \rangle,
    \end{align}
    where we have defined the vector $\ket{\Phi} \defeq \frac{1}{|\Omega|} \sum_{x \in \Omega} \ket{\Phi_x}$. Its norm squared satisfies
    \begin{equation}
    \begin{aligned}
        \langle \Phi | \Phi \rangle & \geq \langle \Phi | \Psi \rangle \langle \Psi | \Phi \rangle  = \left| \frac{1}{|\Omega|} \sum_{x \in \Omega} \langle \Phi_x | \Psi \rangle \right|^2 \\
        & \stackrel{\eqref{eq:one-point_fid_Psi_Phi}}{=} \left( \frac{1}{|\Omega|} \sum_{x \in \Omega} F(\rho_A; O_x) \right)^2 
    \end{aligned}
    \end{equation}
\end{proof}

Theorem ~\ref{theorem:average_one-point_two-point} can be viewed as an ``averaged'' relation between one-point and two-point SW-SSB. The following result can be viewed as another version, on an infinite system, that does not require spacial averaging:

\begin{theorem}\label{thm:two-point_one-point_equiv}
    If $O$ is unit-bounded ($\norm{O}_\infty\leq1$) and $F(\omega, O_x) \geq F$ for infinitely many $x \in \Lambda$, then
    \begin{equation}
        \limsup_{|x-y| \to \infty} F(\omega, O_x O_y^\dagger) \defeq \lim_{R \to \infty} \sup_{|x-y| > R}F(\omega, O_x O_y^\dagger) \geq F^2.
    \end{equation}

    In particular, for all $x \in \Lambda$ such that $F(\omega, O_x) \geq F$, there exists a sequence $(y_n)_{n=1}^\infty$ satisfying $\lim_{n \to \infty} \dist(x, y_n) = \infty$ and
    \begin{align}
        \lim_{n\to \infty} F(\omega, O_x O_{y_n}^\dagger) \geq F^2.
    \end{align}
\end{theorem}
\begin{proof}
    Let $K = \{ x\in \Lambda \mid F(\rho, O_x) \geq F\}$. Since $K$ is infinite by assumption, then for all $R \in \mathbb{N}$, there exists a sequence $(x_n)_{n=1}^\infty$ of points in $K$ that are more than a distance $R$ apart from each other: $\forall n \neq m, |x_n-x_m| > R$.
    
    We now apply Theorem \ref{theorem:average_one-point_two-point} to the set $\Omega_N \defeq \{x_1, x_2, \ldots, x_N\}$. Since the set $A$ in the statement of Thm.~\ref{theorem:average_one-point_two-point} is arbitrary, as long as it contains $\Omega_N$, we will take the limit $A \to \Lambda$ through an arbitrary finite covering, thus replacing $F(\rho_A, O) \to F(\omega, O)$:
    \begin{align}
        F^2 & \leq \left( \frac{1}{N} \sum_{n=1}^N F(\omega, O_{x_n}) \right)^2 \\
        & \leq \frac{1}{N^2} \sum_{n,m=1}^N F(\omega, O_{x_n} O_{x_m}^\dagger) \\
        & = \frac{1}{N^2} \left( \sum_{n=1}^N F(\omega, O_{x_n} O_{x_n}^\dagger) + \sum_{n \neq m}^N F(\omega, O_{x_n} O_{x_m}^\dagger) \right) \\
        &\leq \frac{1}{N^2} \left( N + (N^2 - N) \sup_{|x-y| > R} F(\omega, O_x O_y^\dagger) \right),
    \end{align}
    where in the last line we have used that $F(\rho, O_x O_y^\dagger) \leq \norm{O_x}_\infty \norm{O_y}_\infty \leq 1$ and that $x_n$ are more than a distance $R$ apart from each other. By taking the limit $N \to \infty$, we have
    \begin{equation}
        F^2 \leq \sup_{|x-y| > R} F(\omega, O_x O_y^\dagger).
    \end{equation}
    Since $R$ is arbitrary, we can also take the limit $R \to \infty$, thus arriving at the desired conclusion.

    For the sequence part of the statement, given an arbitrary $x \in K$, we construct the sequence $(y_n)_{n=1}^\infty$ as follows: First, choose a vanishing sequence of real numbers $\epsilon_n \to 0$ where $\forall n, \epsilon_n < 1/2$. Then, choose $y_1 \neq x$ so that $F(\omega, O_x O_{y_1}^\dagger) \geq F^2 - \epsilon_1$. It always exists because of $\sup_{|x-y| > R} F(\omega, O_x O_y^\dagger) \geq F^2$ applies to $R = 1$. Then, let $R_1 = \dist(x, y_1)$, and choose $y_2$ so that $F(\omega, O_x O_{y_2}^\dagger) \geq F^2 - \epsilon_2$ and $|x-y_2| > R_1 = |x-y_1|$. Continue building the sequence $y_n$ in this way. Since $|x-y_n| \geq n$, then $|x-y_n| \to \infty$, and also
    \begin{align}
        \lim_{n \to \infty} F(\omega, O_x O_{y_n}^\dagger) \geq F^2 - \lim_{n \to \infty} \epsilon _n = F^2.
    \end{align}
\end{proof}

\subsection{Connected correlators}\label{sec:connected_correlator}

In many-body systems, the connected two-point function is a central observable because it isolates the genuine correlation between two widely separated operators. It is a natural probe of clustering: in states with a finite correlation length, the two-point function factorizes at large separation, and the connected part decays to zero, whereas a non-vanishing connected correlator signals long-range correlations.

It is then natural to wonder what physical information is carried by the connected LFC. Thus, we define
\begin{definition}
    Let $O \equiv \{ O_x \}_{x \in \Omega}$ be a collection of operators supported on a finite region $A$. Then, the connected fidelity matrix is defined as
    \begin{equation}
        [F_c]_{x,y} = F(\rho_A; O_x O_y^\dagger) - F(\rho_A; O_x) F(\rho_A; O_y^\dagger).
    \end{equation}
    Similarly, the connected R\'enyi-$\alpha$ matrix is defined as
    \begin{equation}
        [R^{(\alpha)}_c]_{x,y} = R^{(\alpha)}(\rho_A; O_x O_y^\dagger) - R^{(\alpha)}(\rho_A; O_x) R^{(\alpha)}(\rho_A; O_y^\dagger).
    \end{equation}
\end{definition}

One way to reinterpret Theorem \ref{theorem:average_one-point_two-point} is that for any vector of real positive entries, we have $v^T F_c v \geq 0$. Though this suggests that $F_c$ is a positive-semidefinite matrix, this is not true. In Appendix~\ref{app:counter_Example_PSD_Fc} we show a counterexample. However, as the following result establishes,  $R^{(\alpha)}_c$ is positive-semidefinite for all $\alpha$.

\begin{lemma}\label{lemma:Renyi_PSD}
    Let $O \equiv \{ O_x \}_{x \in \Omega}$ be a collection of operators supported on a finite region $A$. Then, the connected R\'enyi-$\alpha$ matrix $R^{(\alpha)}_c$ is positive-semidefinite.
\end{lemma}
\begin{proof}
    Denote by $\langle A, B\rangle_{HS} \defeq \Tr(A^\dagger B)$ the Hilbert-Schmidt inner product. Then, define $C_x$ to be the component of $O_x^\dagger \rho_A^{\alpha/2} O_x$ orthogonal to $\rho_A^{\alpha/2}$:
    \begin{equation}
        C_x \defeq O_x^\dagger \rho_A^{\alpha/2} O_x - \frac{\langle O_x^\dagger \rho_A^{\alpha/2} O_x, \rho_A^{\alpha/2}\rangle_{HS}}{\langle \rho_A^{\alpha/2}, \rho_A^{\alpha/2} \rangle_{HS}} \rho_A^{\alpha/2}.
    \end{equation}
    We have 
    \begin{equation}
        [R_c^{(\alpha)}]_{x,y} = \frac{\langle C_x, C_y \rangle_{HS}}{\Tr(\rho_A^\alpha)}.
    \end{equation}
    Since $[\langle C_x, C_y \rangle]_{x,y}$ is a Gram matrix, $R_c^{(\alpha)}$ is positive-semidefinite.
\end{proof}
\begin{corollary}
    Let $O \equiv \{ O_x \}_{x \in \Omega}$ be a collection of operators supported on a finite region $A$, then
    \begin{equation}
        \frac{1}{|\Omega|^2} \sum_{x,y \in \Omega} R^{(\alpha)}(\rho_A, O_x O_y^\dagger) \geq \left( \frac{1}{|\Omega|} \sum_{x \in \Omega} R^{(\alpha)}(\rho_A, O_x) \right)^2.
    \end{equation}
\end{corollary}
\begin{proof}
    From positive-semidefiniteness of $R_c^{(\alpha)}$, for any vector $w \in \C^{|\Omega|}$ we have $w^T R^{(\alpha)}_c w \geq 0 $. Then, the results follow from considering the constant $w_x \defeq 1/|\Omega|$. 
\end{proof}

Lastly, we briefly comment on the behavior of the connected $R^{(1)}_c$ as $|x-y| \to \infty$. First, if the canonical purification satisfies cluster decomposition (e.g., if it has a finite correlation length), then $R^{(1)}_c \to 0$. By the arguments of \cite{li_unified_2026}, this is also implied if the original mixed state has both short-range correlations and short-range CMI. As a consequence, thermal states (both for gapped and for gapless Hamiltonians) have finite R\'enyi-1 correlators, but $R_c \to 0$.

\section{Examples}
\label{sec:examples}

In this section, we discuss a few concrete examples with either long-range or critical local SW-SSB order. For critical states the calculation of the local fidelity correlator itself becomes an interesting defect problem. We demonstrate this in the decohered Ising model, which is known to be described by the classical random-bond Ising model (RBIM) along the Nishimori line. We also discuss the local R\'enyi-1 correlator in pure states, including ground states of conformal field theories (CFT) and free fermion metals (ballistic and diffusive), and systematically obtain universal scaling behaviors of the local R\'enyi-1 correlator in each case. Our results are summarized in Table~\ref{tab:local_scalings}.

Throughout the section, we denote by
\begin{equation}
    \ell = \dist(x,\partial A)
\end{equation}
the distance from the insertion point to the boundary of $A$.

\begin{table}[h]
\centering
\renewcommand{\arraystretch}{1.2}
\setlength{\tabcolsep}{8pt}
\begin{tabular}{lll}
\toprule
System/state & Diagnostic & Scaling \\
\midrule
Thermal state 
& $R^{(1)}(O_x)$ 
& SW-SSB \\

CFT  
& $R^{(1)}(\mathcal O_x)$ 
& $\ell^{-2\Delta_{\mathcal O}}$ \\

Clean Fermi metal 
& $R^{(1)}(c_x)$ 
& $\ell^{-1}$ \\

Dirac semimetal 
& $R^{(1)}(c_x)$ 
& $\ell^{-d}$ \\

Diffusive metal 
& $\overline{R^{(1)}(c_x)}$ 
& $\ell^{-2}$ \\

Random Gaussian state 
& $\overline{R^{(1)}(c_x)}$ 
& local SW-SSB \\
\bottomrule
\end{tabular}
\caption{
Summary of local one-point R\'enyi-$1$ scalings. Here $\ell=\dist(x,\partial A)$, $O_x$ is an arbitrary local charged operator, and $\mathcal{O}_x$, a primary operator. The thermal state has both local and global SW-SSB, while the random Gaussian state has only local SW-SSB. For Fermi metals not described by CFT, the $R^{(1)}$ correlator scales differently from ordinary two-point correlation functions.
}
\label{tab:local_scalings}
\end{table}

\subsection{Thermal states}
\label{sec:thermal}
Thermal states are expected to exhibit the standard global SW-SSB, based on the arguments in Refs.~\cite{lessa2025strong,liu2025diagnosing}. Therefore by Theorem \ref{thm:inequality_of_measures} a thermal state should also have local SW-SSB order.

More explicitly, for a thermal state $\rho$, we expect on the reduced density matrix a finite region $A$ to be also thermal: $\rho_A =\frac{1}{Z_A} e^{-\beta H_A}$. For local systems, $H_A$ is roughly speaking the restriction of $H$ to $A$ plus boundary contributions, which need not be symmetric. We then expect the one-point fidelity/R\'eny correlator on $\rho_A$ to be finite. This is simplest to see for the R\'enyi-$1$ correlator~\cite{liu2025diagnosing}: for a thermal state $\sqrt{\rho_A}\sim e^{-\beta H_A/2}$ is another  thermal state, so the R\'enyi-$1$ correlator takes the form
\begin{eqnarray}
    R^{(1)}(\rho_A;O_x)&=&\frac{1}{Z}\Tr(e^{-\frac{\beta}{2}H_A}O_xe^{-\frac{\beta}{2}H_A} O_x^{\dagger}) \nonumber \\ 
    &=&\langle O_x(\tau=\beta/2)O_x^\dagger(\tau=0)\rangle,
\end{eqnarray}
where the expectation value is taken over the imaginary-time path integral over $0\leq\tau\leq\beta$. For a generic local operator $O_x$, we expect this to give a finite value as long as $\beta$ is finite. If the Hamiltonian has a symmetric ground state with a finite gap $\Delta$, we expect the correlator to decay as $e^{-\beta\Delta/2}$ at low temperature.

As a concrete illustration, in Appendix~\ref{app:paramagnet} we explicitly calculate the one-point LFC of a finite-temperature Ising paramagnet in the canonical ensemble, and show that local SW-SSB indeed arises.

\subsection{\texorpdfstring{$ZZ$}{ZZ}-decohered Ising paramagnet}\label{sec:ZZ_decoherence}

Here, we consider one of the simplest decoherence models with SW-SSB transition~\cite{LeeJianXu,lessa2025strong}. The symmetry group is $\Z_2$ generated by $\prod_i X_i$ and the initial state is $\rho_0 = \ketbra{+\cdots +}{+ \cdots +}$, a pure state in the trivial paramagnetic phase. The system is subject to nearest-neighbor $ZZ$ decoherence with probability $p$:
\begin{equation}
    \rho_p = \prod_{e=\langle i, j \rangle} \E_{e,p}^{ZZ} [\rho_0], 
\end{equation}
where $\E_{(i,j), p}^{ZZ}[ \cdot] = (1-p) [ \cdot] + p Z_i Z_j [\cdot] Z_i Z_j$. In two dimensions, it is well known that this model exhibits a transition from a symmetric phase to an SW-SSB phase at $p_c\approx0.109$, and the critical properties are described by the random-bond Ising model (RBIM) along the Nishimori line. It is not difficult to see that the phase diagram for local SW-SSB order will be the same: Thm.~\ref{thm:inequality_of_measures} guarantees that states at $p>p_c$ (with global SW-SSB) will also exhibit local SW-SSB; states at $p<p_c$ is known to have exponentially decaying CMI, therefore by Thm.~\ref{thm:LRCMI} will not exhibit local SW-SSB.

Now let us consider the LFC more explicitly. After reducing the state to a region $A$, we have
\begin{equation}\label{eq:reduced_state_ZZ}
    \rho_{p,A} = \Tr_{\bar A}[\rho_p] = \prod_{\substack{\langle i, j\rangle \\ i \in A \\ j \notin A}} \E_{i,p}^{Z} \circ \prod_{\substack{\langle i, j\rangle \\ i,j \in A}} \E_{\langle i,j\rangle,p}^{ZZ} [\rho_{0,A}], 
\end{equation}
where $\E_{i,p}^{Z}[\cdot] = (1-p) [\cdot] + p Z_i [\cdot] Z_i$.
That is, $ZZ$ decoherence at the boundary becomes single $Z$ decoherence, which breaks the strong symmetry explicitly. To treat bond and boundary errors uniformly, we enlarge the edge set of the region $A$ to $\tilde{E}_A$, including edges straddling the boundary of $A$. For such edges $e = \langle i, j \rangle \in \tilde{E}_A$, we define $\E_{e, p}^{ZZ}[\cdot] = (1-p) [\cdot] + \tilde{Z}_i \tilde{Z}_j [\cdot] \tilde{Z}_i \tilde{Z}_j$, where $\tilde{Z}_i = Z_i$ if $i \in A$, otherwise $\tilde{Z}_i = \one_i$. In this notation, $\rho_{p,A}$ is simply
\begin{equation}
    \rho_{p,A} = \prod_{e \in \tilde{E}_A} \E_{e,p}^{ZZ}[\rho_{0,A}].
\end{equation}
Since the $ZZ$ decoherence maps Pauli-$X$ basis states to other basis states, we have
\begin{equation}
    \rho_{p,A} = \sum_{s \in \{0,1\}^A} p_s \ketbra{s}{s}_X,
\end{equation}
where $\ket{s}_X$ is the product basis state in the $X$-basis so that $s=0$ corresponds to $\ket{+}$ and $s=1$, to $\ket{-}$. We introduce a new variable $\tau_e$ to indicate when the $\tilde{Z}\tilde{Z}$ ($ZZ$ or $Z$) decoherence is acted on edge $e \in \tilde{E}_A$, in which case $\tau_e = -1$. In similar spirit, we will call $\sigma_i \defeq (-1)^{s_i}$. With that, we can express the probability distribution $p$ as 
\begin{align}
    p_s & = (2 \cosh(\beta))^{-|\tilde{E}_A|}\sum_\tau e^{\beta\sum_{e \in \tilde{E}_A} \tau_e} \prod_{i \in A} \delta \left(\sigma_i,\prod_{e \ni i} \tau_e\right),
\end{align}
where 
\begin{equation}\label{eq:temperature_p_relation}
e^{-2 \beta} = \frac{p}{1-p}.    
\end{equation}
We can reexpress $p_s$ in two different ways. One is a ``high-temperature expansion'', and comes from substituting $\delta_{\sigma_i,\prod_{e \ni i} \tau_e} = \frac{1}{2}(1 + \sigma_i \prod_{e \in i} \tau_e)$ and expanding the product in $i \in A$:
\begin{align}
    p_s & \propto \sum_{\eta_{i} \in \{\pm1\}} \sum_{\tau} e^{\beta \sum_{e \in \tilde{E}_A} \tau_e} \prod_{\eta_i = -1} \sigma_i \prod_{\eta_i \eta_j = -1} \tau_{\langle i,j \rangle} \\
    & \propto \sum_{\eta} \left(\prod_{\eta_i = -1} \sigma_i\right) \tanh(\beta)^{\sum_{\langle i, j\rangle \in \tilde{E}_A}(1-\eta_i \eta_j)/2} \\
    & \propto \sum_{\eta} e^{K \sum_{\langle i, j\rangle \in \tilde{E}_A} \eta_i \eta_j} \prod_{i \in A} \eta_i^{s_i} \\
    & = \sum_{\eta} e^{K \sum_{\langle i, j \rangle \in A} \eta_i \eta_j + K \sum_{i \in \partial A} \eta_i} \prod_{i \in A} \eta_i^{s_i} \label{eq:ZZ_decoherence_ising_model}
\end{align}
where $K$ is Krammers-Wannier dual to $\beta$ via $e^{-2 K} = \tanh(\beta)$. In the last line, we made explicit the fact that the sums over edges $\langle i, j \rangle \in \tilde{E}_A$ also range over the edges connecting sites inside $A$ to those outside. For these sites $j \notin A$, we should set $\eta_j =+1$, which effectively adds a boundary magnetic field term $H_{\partial A} = \sum_{i \in \partial A} \eta_i$. In this form, $p_s$ is the expectation value of $\prod_i \eta_i^{s_i}$ in the classical Ising model at inverse temperature $K$. This will be particularly useful for the R\'enyi-$k$ order parameter in Sec. \ref{sec:ZZ_decoherence}. 

Another way is by gauge fixing a particular configuration $\tau_e^*(s, b)$ for each bulk spin configuration $s$ and boundary condition $b \equiv \{b_{i}\}_{i \in \partial A^c}$, with $\partial A^c = \{ i \notin A \mid \exists j \in A, \langle i, j \rangle \in \tilde{E}_A\}$, that satisfies the same flux conditions $(-1)^{s_i} = \sigma_i = \prod_{e \ni i} \tau_e^*(s,b)$ for $i \in A$ and $(-1)^{b_i} = \prod_{e \ni i} \tau_e^*(s,b)$ for $i \in \partial A^c$. By fixing the boundary condition $b$, and assuming for simplicity that $A$ is simply connected, we have that any other configuration $\tau_e$ satisfying the same constraint $(s, b)$ is related to $\tau_e^*$ by a gauge transformation $\tau_{e} \to \tau_{e} \prod_{p \ni e} \nu_p$, where the product is over plaquettes (or faces) $p$ whose boundary edges are all in $\tilde{E}_A$, and $\nu_p \in \{\pm 1\}$~\footnote{This statement can be formalized by noting that the relative homology $H_1(A, \partial A^c) = \Z_2^{|\partial A^c|}$.}. Hence,
\begin{align}\label{eq:ZZ_decoherence_RBIM}
    p_s & \propto \sum_{b_{i \in \partial A^c}} \sum_{\nu_p \in \{\pm1\}} e^{\beta\sum_{e \in \tilde{E}_A} \tau_{e}^*(s,b) \prod_{p \ni e} \nu_p}
\end{align}
In words, $p_s$ is a realization of the random bond Ising model (RBIM) partition function with spins $\nu_p$ on plaquettes and interaction coefficients $\tau_e^{*}(s,b)$ that connect the points with $s_i = 1$ or $b_i = 1$. 
\subsubsection{Local fidelity correlator}

From Eq. \eqref{eq:ZZ_decoherence_RBIM}, we can relate the SW-SSB properties of $\rho_p$ to the RBIM at the Nishimori line. Namely, the local fidelity correlator assumes the following form:
\begin{align}
    F(\rho_A; Z_x) & = \sum_s \sqrt{p_s p_{s + \hat{e}_x}} \\
    & \propto \sum_{s,b} \mathcal{Z}_{s,b} \sqrt{\frac{\sum_b \mathcal{Z}_{s + \hat{e}_x,b}}{\sum_{b} \mathcal{Z}_{s,b}}} \\
    & = \sum_{s,b} \mathcal{Z}_{s,b} e^{-\frac{1}{2}\beta \Delta F^{(s)}_{A,x}}, \label{eq:free_energy_diff_fidelity}
\end{align}
where $\mathcal{Z}_{s,b} = \sum_{\nu} e^{\beta\sum_{e \in \tilde{E}_A} \tau_{e}^*(s,b) \prod_{p \ni e} \nu_p}$ is the RBIM partition function with bond fluxes set by the bulk spin configuration $s$ and boundary conditions $b$, and $\Delta F^{(s)}_{A,x}$ is the difference in free energy after inserting a bond flux at site $x \in A$, maintaining the other bulk fluxes $s$ fixed, but summing over the free boundary conditions $b$. This means that the bond configurations $\tau^*_{e}(s,b)$ flip along a line that connects $x$ to all points in the boundary $\partial A^c$.

The RBIM has a ferromagnetic phase for low temperatures $T < T_c$ ($p < p_c$ via Eq.~\eqref{eq:temperature_p_relation}), and a paramagnetic phase for high temperatures $T > T_c$ ($p > p_c$). In the ferromagnetic regime, inserting a bond flux at the bulk site $x$ incurs a free energy penalty roughly proportional to the distance $\ell$ from $x$ to the boundary due to the alignment of the spins. From Eq.~\eqref{eq:free_energy_diff_fidelity}, this causes the local fidelity correlator to decay exponentially in $\ell$. In the paramagnetic regime, the free energy penalty is not extensive in $\ell$, resulting in $F(\rho; Z_x) > 0$.

\subsubsection{Replica calculation and boundary magnetization}\label{sec:ZZ_replica_renyi}
By employing the expression of $p_s$ from Eq. \eqref{eq:ZZ_decoherence_ising_model}, we can relate the one-point R\'enyi-$2k$ measure $R^{(2k)} (\rho_{p,A}; Z_x)$, $k \in \Z_{> 0}$, to a $2k$-replica calculation:
\begin{align}
    & R^{(2k)} (\rho_{p,A}; Z_x) = \sum_{s} p_s^k p_{s + \hat{e}_x}^k \\
    & \propto \sum_{s} \sum_{\eta^{(\alpha)}} e^{K \sum_{\alpha=1}^{2k}H_A[\eta^{(\alpha)}]} \left(\prod_{\alpha=k+1}^{2k} \eta_x^{(\alpha)}\right) \prod_{i \in A} \prod_{\alpha=1}^{2k} (\eta_i^{(\alpha)})^{s_i} \\
    & = \sum_{\eta^{(\alpha)}} e^{K \sum_{\alpha=1}^{2k} H_A[\eta^{(\alpha)}]} \left(\prod_{\alpha=k+1}^{2k} \eta_x^{(\alpha)}\right) \prod_{i \in A} \delta\left( \prod_{\alpha=1}^{2k} \eta_i^{(\alpha)}, 1\right), \\
\end{align}
where $H_A[\eta] = \sum_{\langle i, j \rangle \in A} \eta_i \eta_j + \sum_{i \in \partial A} \eta_i$ is the classical Ising model Hamiltonian with boundary magnetic field. For $2k = 2$, this simplifies significantly, as the delta function sets $\eta^{(1)} = \eta^{(2)}$:
\begin{align}
    R^{(2)}(\rho_{p,A}; Z_x) \propto \sum_{\eta} e^{2K H_A[\eta]} \eta_x = \langle \eta_x\rangle_{2K},
\end{align}
where $\langle \cdot \rangle_{2K}$ denotes the thermal expectation value of $H_A$ at inverse temperature $2K$.

Hence, the one-point R\'enyi-2 measure has an appealing interpretation: it is the local magnetization of the Ising model on a finite region $A$ with a magnetized boundary that breaks the symmetry. This choice of boundary condition will magnetize the bulk spins whenever the system is in the ferromagnetic phase.

\subsection{Pure states satisfying cluster decomposition}

Here, we argue that the one-point fidelity equals the absolute value of the linear one-point function for infinite-volume states that are pure. These immediately exclude cat states of ordinary SSB such as the GHZ state, whose infinite volume limit is the (mixed) classical SSB ensemble. The key feature of these pure states that we rely on is their \textit{clustering}: for all local operators $O_A$ and $O_B$ respectively supported on regions $A$ and $B$, $|\langle O_A O_B \rangle - \langle O_A \rangle \langle O_B \rangle| \xrightarrow{\dist(A,B) \to \infty} 0$~\cite{bratteli_operator_2012, tasaki_lieb_2022}. Note that no assumption is made about the decay rate of the connected correlation functions. In particular, it can be superexponential.

We argue for the statement above now with slight disregard for the technicalities of infinite-volume states. In particular, let us treat the infinite-volume pure state as a ket $\ket{\psi}$ that is clustering. Then, having fixed a local operator $O_x$, we have
\begin{align}
    F(\ket{\psi}; O_x) & = \lim_{|A| \to \infty} F(\rho_A, O_x \rho_A O_x^\dagger) \\
    & = \lim_{|A| \to \infty} \max_{U_{A^c}} |\langle\psi | O_x \otimes U_{A^c} | \psi \rangle| \\
    & = \lim_{|A| \to \infty} \max_{U_{A^c}} |\langle\psi | O_x | \psi \rangle| \cdot |\langle \psi | U_{A^c} | \psi \rangle| \\
    & = |\langle O_x \rangle_\psi|,
\end{align}
where we used Uhlmann's theorem in the second line, the clustering property in the third line, and, in the last line, that $U_{A^c} = \one_{A^c}$ maximizes $|\langle \psi | U_{A^c} |\psi \rangle|$ at 1. A rigorous proof of the relation above can be obtained using properties of the general C$^*$-algebra Uhlmann fidelity~\cite{UHLMANN1976273, alberti_note_1983}. 

In the specific case of ground states of symmetric Hamiltonians, we expect the clustering property to hold in two cases: when the Hamiltonian is not in the SSB phase, and when the ground state is a physical state of an SSB Hamiltonian, which breaks the symmetry explicitly. For the latter case, the symmetry breaking pattern can be detected through a charged operator, which would also give rise to a nonzero one-point fidelity measure with the exact same magnitude as the linear one-point function.

In the illustrative sections below, we study the behavior of the LFC for (pure) ground states that do not break the symmetry either spontaneously or explicitly. Due to the preceding arguments, the limiting value of the local fidelity measure is trivially zero, so our focus will be on describing the finite-size scaling as the LFC approaches zero.

\subsection{Conformal Field Theories}
\label{sec:CFT}

An important class of critical states is conformal field theories (CFTs). We will now discuss the behavior of one-point local R\'enyi-1 correlators in CFTs, for a primary operator $\OO$ of scaling dimension $\Delta_\OO$ in $d$ spatial dimensions. 

Let us first summarize our results below. We find that for a thermal state at inverse temperature $\beta$ and $\ell\to\infty$, we have
\begin{equation}\label{eq:R1_thermal_CFT_SWSSB}
    R^{(1)}(\omega; \OO(x)) =  \frac{a_\OO}{\beta^{2\Delta_\OO}}. 
\end{equation}
In $2d$, the prefactor can be fixed solely by the normalization of the two-point function $\langle \OO \OO\rangle$, whereas in higher dimensions it is generally determined by other CFT data. 

On the other hand, for the vacuum state, we have $ R^{(1)}(\omega, \OO)  \to 0$, and for finite $\ell$ the local R\'enyi-$1$ correlator exhibits the universal scaling
\begin{equation}\label{eq:R1_groundstate_CFT_scaling}
    R^{(1)}(\rho_A; \OO(x)) \sim  \frac{c_\OO}{\ell^{2\Delta_\OO}}. 
\end{equation}
Again, in $2D$ the prefactor and subleading corrections can be determined explicitly, whereas in higher dimensions they are fixed by theory-dependent data.

\subsubsection{Two dimensional CFTs}
Let us consider two-dimensional CFTs ($D = 1+1$). Crucially, we will show that the one-point R\'enyi-$1$ correlator for a primary operator $\OO$, on a finite region $A$ can be mapped to a conventional two-point function in the complex plane via the uniformization map. Therefore, the R\'enyi-$1$ correlator is completely fixed by the normalization of the two-point function. 

Consider the spatial region $A= [u,v]$ on a thermal cylinder with inverse temperature $\beta$. We consider the one-point R\'enyi-$n$ correlator
\begin{equation}\label{eq:replicated_Renyi1}
R^{(2n)}(\rho_A; \OO(x))=\frac{\Tr_A\left( \rho_A^{n}\OO(x)\rho_A^{n}\OO^\dagger(x)\right)}{\Tr_A(\rho_A^{2n})}
\end{equation}
and analytically continue $n\to \tfrac12$ to obtain the one-point R\'enyi-1 correlator.

As in standard replica constructions, $\Tr_A(\rho_A^{2n})$ is represented by an Euclidean path integral on a $2n$-sheeted Riemann surface $\mathcal R_{2n}$ branched over $A$, with twist operators inserted at its endpoints~\cite{calabrese2004entanglement, calabrese2009entanglement}. The two operator insertions $\OO$, $\OO^\dagger$ are separated by a Euclidean time shift $n\beta$, placing them on sheets whose indices differ by $n$ in the $2n$-sheet geometry.

We first map the thermal cylinder to the complex plane
\begin{equation}\label{eq:uniformization_map_first_step}
    z = e^{\frac{2\pi}{\beta}(x+i\tau)},
\end{equation}
The endpoints of $A$ at $\tau=0$ map to $u \to z_u=e^{\frac{2\pi}{\beta}u}$, and $v \to z_v=e^{\frac{2\pi}{\beta}v}$, while the operator insertion maps to $x\to z_x=e^{\frac{2\pi}{\beta}x}$.

Under the map, Eq.~\eqref{eq:uniformization_map_first_step}, the primary vertex operators pick up a Jacobian factor, but the scaling factors associated with the twist operators cancel between numerator and denominator in Eq.~\eqref{eq:replicated_Renyi1}.

Next, we uniformize $\mathcal R_{2n}$ through
\begin{equation}\label{eq:uniformization_map_second_step}
    w^{2n}=\frac{z-z_u}{z-z_v}.
\end{equation}
Encircling a branch point sends $w\to e^{2\pi i/(2n)}w$, reproducing the replica monodromy, so that the branch points are encoded in the multivaluedness of the coordinate $w$. The $2n$ preimages of $z_x$ satisfy
\begin{equation}
w_k=w_x e^{2\pi i \frac{k}{2n}},  \qquad w_x^{2n}=\frac{z_x-z_u}{z_x-z_v}.
\end{equation}where $ k=0,\dots,2n-1$. Therefore, the operator insertions $\OO$ are located at $w_0=w_x$ and $w_n=-w_x$.

After uniformization, the surface becomes the complex $w$-plane, and the correlator reduces to a standard two-point function, $\langle \OO(w, \bar w)\OO^\dagger(w', \bar w')\rangle_{\mathbb C}
=|w-w'|^{-2\Delta_{\OO}}$. Including the Jacobian factors, we get
\begin{equation}
R^{(2n)}(\rho_A;\OO(x))
= \left[ \frac{\pi}{4n\beta} \frac{\sinh\!\frac{\pi(v-u)}{\beta}} {\sinh\frac{\pi(x-u)}{\beta} \sinh\frac{\pi(v-x)}{\beta}}\right]^{2\Delta_{\OO}}.
\end{equation}
The result depends only on the cross-ratio determined by the interval endpoints and the insertion point on the thermal cylinder. We can compare this expression with the two-point in the thermal cylinder
\begin{equation}
    \langle \OO(x,0) \OO^\dagger(x,\tau)\rangle_\beta = \left[\frac{\pi /\beta}{\sin \pi \tau/ \beta}\right]^{2\Delta_{\OO}}
\end{equation}

Analytically continuing, we obtain the one-point R\'enyi-1 for two cases of interest. For finite $\beta$, we take the covering limit and arrive at
\begin{equation}
    R^{(1)}(\omega;\OO(x))= \left( \frac{\pi}{\beta}\right)^{2\Delta_{\OO}} > 0.
\end{equation}
Therefore, two-dimensional thermal CFTs are expected to always exhibit SW-SSB, in agreement with our general discussions in Sec.~\ref{sec:thermal}.

In contrast, at zero temperature ($\beta \to \infty$), the one-point R\'enyi-1 always vanishes in the covering limit. For a symmetric covering $A = [x-\ell, x+\ell]$, we find
the universal scaling 
\begin{equation}\label{eq:R1_CFT_zero_temperature}
    R^{(1)}(\rho_A;\OO(x)) =  \frac{1}{\ell^{2\Delta_{\OO}}}. 
\end{equation}
Note that the distance to the boundary controls the decay of the R\'enyi-1 correlator. For example, for a non-isotropic covering $A =[ x- \ell_L,  x+ \ell_R]$, the correlator decays like 
\begin{equation}
    R^{(1)}(\rho_A;\OO(x)) =  \left( \frac{\ell_L^{-1} + \ell_R^{-1}}{2} \right)^{2\Delta_{\OO}}, 
\end{equation}
and the leading contribution is $1 / \ell^{2\Delta_{\OO}} $ with $\ell= \min(\ell_L , \ell_R)$ the distance to the boundary. 

\subsubsection{Higher-dimensional CFTs}
The uniformization map method does not generalize directly to higher dimensions, because $\partial A$ is no longer a set of isolated points, and the branch cut cannot be removed by a conformal transformation. Thus, for $D \geq 3$ the one-point R\'enyi correlator is instead mapped to a nontrivial codimension-two twist defect problem. However, as we argue next, the same universal conclusions, Eqs.~\eqref{eq:R1_thermal_CFT_SWSSB} and \eqref{eq:R1_groundstate_CFT_scaling} are expected to hold for both vacuum and thermal states of higher-dimensional CFTs.

First, consider a thermal state $\rho_\beta = e^{-\beta H}/Z_\beta$ of a CFT on flat space. Since the covering limit exists, we have
\begin{equation}\begin{aligned}
    R^{(1)}(\rho_\beta; \OO(x)) &= \frac{1}{Z_\beta}\Tr\left(e^{-\beta H/2} \OO(x) e^{-\beta H/2} \OO^\dagger(x)\right)\\
     &= \langle \OO(x, \beta /2) \OO^\dagger(x,0)\rangle_\beta.
\end{aligned}
\end{equation}
where $\mathcal O(x,\tau) = e^{+\tau H}\mathcal O(x)e^{-\tau H}$, and $\langle \cdots \rangle_\beta = \frac{1}{Z_\beta} \Tr( \cdots ~e^{-\beta H})$.
This correlation function is constrained by conformal invariance. Translational invariance removes the dependence on $x$, and thus $\beta$ is the only scale. Then,  the scaling of the primary $\OO$ fixes the form
\begin{equation}
    \langle \OO(x, \tau) \OO^\dagger(x,0)\rangle_\beta. = \frac{f(\tau/\beta)}{\beta^{2\Delta_{\OO}}},
\end{equation}
with $f$ a dimensionless function. Thus, provided that $f(1/2) \neq 0$, we obtain
\begin{equation}\label{eq:R1_CFT_finiteT_higher_d}
    R^{(1)}(\rho_\beta; \OO(x)) = \frac{a_{\mathcal O}}{\beta^{2\Delta_{\OO}}} > 0,
\end{equation}
where $a_{\mathcal O}$ is theory-dependent data. In particular, $R^{(1)}>0$ for any finite inverse temperature $\beta<\infty$, so thermal CFTs in any spacetime dimension are expected to exhibit SW-SSB, as expected from Sec.~\ref{sec:thermal}.

We next consider the vacuum state. Based on the previous discussion, the one-point R\'enyi-1 correlator is expected to vanish in the covering limit. We are instead interested in determining the universal scaling of $R_A^{(1)}$ with $\ell = {\rm dist}(x, \partial A)$. 

To do so, let $A$ be a ball of radius $\ell$ centered at $x$, which by translational invariance we place at the origin. By the Casini-Huerta-Myers formula~\cite{casini_derivation_2011}, the reduced density matrix on the ball is
\begin{equation}
\begin{aligned}
\rho_A &= \frac{e^{-K_A}}{\Tr_A e^{-K_A}},\\ 
K_A &= 2\pi \int_{|x|<\ell}\hspace{-1.2 em} d^{d}x \left(\frac{\ell^2-|x|^2}{2\ell} T_{00}(x) \right).
\end{aligned} 
\end{equation}
In other words, $\rho_A$ looks like a thermal state with respect to energy density $T_{00}$, with an effective position-dependent inverse temperature $\beta_{\rm eff}(x) = \pi\frac{\ell^2-|x|^2}{\ell}$. Crucially, near the operator insertion, this effective inverse temperature is finite,  $\beta_{\rm eff}\sim \ell$, and the thermal result, Eq.~\eqref{eq:R1_CFT_finiteT_higher_d}, can anticipate the scaling. More precisely, we can perform a Weyl transformation that maps the ball to a unit-radius ball, yielding
\begin{equation}\label{eq:R1_CFT_scaling}
    R^{(1)}(\rho_A; \mathcal O(x)) = \frac{C_{\mathcal O}}{ \ell^{2\Delta_{\mathcal O}}}, 
\end{equation}
where $C_\OO$ is an $\ell$-independent constant determined by the CFT data. Therefore, for a general covering region, the expected scaling of the R\'enyi-$1$ one-point correlator is precisely as the decay of the two-point function in flat space.

\subsection{Metals and semimetals}
\label{sec:metals}
In this section, we consider ground states of translationally invariant free fermion systems with a U(1) particle-number symmetry. Among these states, the most interesting gapless examples are Fermi metals and semimetals. As we will show next, neither of these systems exhibits SW-SSB, but the local one-point R\'enyi-1 does exhibit universal critical behavior fixed by the low-energy structure.

For a Fermi metal with a smooth codimension-one Fermi surface, in any spatial dimension, we find that 
\begin{equation}\label{eq:Fermi_metal_scaling_intro}
    R^{(1)}(\rho_A;c_x) \sim \frac{a_{\rm FS}}{\ell}, 
\end{equation}
with a non-universal coefficient $a_{\rm FS}$ which depends on the geometry of the Fermi surface and the shape of the region $A$.  

We can understand this result as follows. In $1d$, the low-energy physics is determined by two Fermi points and is thus described by a free Dirac CFT, yielding the $\frac{1}{\ell}$ scaling according to Eq.~\eqref{eq:R1_CFT_scaling}. In higher dimensions, however, the low-energy modes live on a codimension-one Fermi surface, and each point defines a chiral one-dimensional mode, propagating in the direction of the Fermi velocity. Crucially, the density of these modes is finite and, upon integrating over the entire surface, yields the same $1d$ scaling. 

By contrast, in a Dirac semimetal with isolated linear nodes, the low-energy modes are concentrated near a finite number of Dirac points. The low-energy excitations are thus governed by a Dirac CFT. Hence
\begin{equation}\label{eq:semimetal_scaling_intro}
    R^{(1)}(\rho_A; c_x) \sim \frac{a_{\rm SM}}{\ell^{2\Delta_\psi}},
\end{equation}
where $\Delta_\psi=\frac{d}{2}$ is the scaling dimension of a free Dirac fermion in $D = d+1$ dimensions.

More broadly, we will argue that the scaling of the Rényi-1 correlator is controlled by the formula
\begin{equation}\label{eq:free_fermion_general_formula}
      R^{(1)}(\rho_A; c_x) \sim \int_{\bar A} \! d^dy~|C_{xy}|^2 .
\end{equation}
In other words, it is controlled by processes in which the particle inserted at $x$ propagates to an arbitrary point in the complement of $A$ and returns. This reproduces both Eqs.~\eqref{eq:Fermi_metal_scaling_intro} and Eq.~\eqref{eq:semimetal_scaling_intro}, and allows us to argue for the behavior of other states, including diffusive metals.

The rest of this section is organized as follows. First, we derive the free-fermion formula that expresses $R_A^{(1)}$ in terms of the restricted single-particle correlation matrix. We then analyze the one-dimensional Fermi gas by diagonalizing the sine kernel, reproducing the CFT result. Next, we generalize the argument to higher-dimensional Fermi metals and derive the Fermi-surface formula for arbitrary smooth Fermi surfaces. Finally, we present the general argument, leading to Eq.~\eqref{eq:free_fermion_general_formula} and apply it to weakly disordered diffusive metals.

\subsubsection{R\'enyi-1 correlators in free fermion systems}
A useful property of free fermions is that Wick's theorem holds for correlation functions restricted to any region $A$. Consequently, the reduced density matrix $\rho_A$ of a Gaussian state on region $A$ is also a Gaussian state, and specifying $\rho_A$ is equivalent to specifying the two-point correlation matrix 
\begin{equation}
    C_A(x,y) = \langle c_x^\dagger c_y\rangle, \qquad  x,y\in A.
\end{equation}
We can diagonalize $C_A(x,y)$ in terms of the effective single-particle eigenstates as
\begin{equation}\label{eq:free_fermions_diagonalize}
    C_A(x,y)=\sum_\alpha\lambda_\alpha \phi_{\alpha}^*(x)\phi_\alpha(y),
\end{equation}
where $\lambda_\alpha\in[0,1]$ is the effective average occupation number of the eigenstate $\alpha$, with wavefunction $\phi_\alpha(x)$. For each $\alpha$, we denote the un-occupied and occupied states as $|0_\alpha\rangle$ and $|1_\alpha\rangle$, respectively. The many-body reduced density matrix factorizes as 
\begin{equation}
    \rho_A =  \bigotimes_\alpha \left[(1-\lambda_\alpha)\ket{0_\alpha}\bra{0_\alpha}
    +\lambda_\alpha\ket{1_\alpha}\bra{1_\alpha}\right].
\end{equation}
Then, the one-point R\'enyi-1 correlator takes the form
\begin{equation}\label{eq:R1_free_fermions_efficient}
\begin{aligned}
    R^{(1)}(\rho_A; c_x) =& \left[ \sqrt{C_A(1-C_A)} \right]_{(x,x)}\\
    =& \sum_\alpha |\phi_\alpha(x)|^2\sqrt{\lambda_\alpha(1-\lambda_\alpha)}.
\end{aligned}
\end{equation}
The same expression also holds for $R_A^{(1)}(c_x^\dagger)$. Therefore, we can determine R\'enyi-1 correlators by solely working in a single-particle basis, leading to great simplifications for both analytical and numerical approaches.

\subsubsection{Fermi gas in \texorpdfstring{$1d$}{1d}}

Consider a one-dimensional Fermi gas with Fermi momentum $k_F$. The ground state is 
\begin{equation}
    |{\rm GS}_{k_F}\rangle = \prod_{|k|\leq k_F} \hspace{-0.2em }c_k^\dagger ~ |{\rm vac}\rangle,
\end{equation}
where $ |{\rm vac}\rangle$ denotes the empty state. The low-energy physics is described by two-Fermi points, and thus by a free Dirac CFT. Therefore, after accounting for the normalization of the microscopic two-point function,
\begin{equation}\label{eq:two_point_function_1d}
   \langle {\rm GS}_{k_F}|c_x^\dagger c_0|{\rm GS}_{k_F}\rangle = \frac{\sin (k_Fx)}{\pi x},
\end{equation}
the one-point R\'enyi-1 correlator reduces to the CFT result in Eq.~\eqref{eq:R1_CFT_zero_temperature} with scaling dimension $\Delta = \frac{1}{2}$. We now show this explicitly; the same structure will later be used to analyze higher-dimensional Fermi gases.

We make two simplifications. First, we take the region to be the semi-infinite line $ A= [ 0,\infty)$ and insert the operator at $x = \ell$. This is equivalent to considering the finite interval $A= [- \ell,  R]$, inserting the operator at $x =0$, and then taking $R \to \infty $, with $\ell$ fixed. Second, we linearize the dispersion near the Fermi points $k = \pm k_F$, which captures the dominant contribution when $k_F \ell \gg 1$. Explicitly, we write 
\begin{equation}
    c_x \simeq e^{ i k_Fx}\psi_R(x) +e^{- i  k_F x} \psi_L(x), 
\end{equation}
where $\psi_{R}$ and $\psi_{L}$ annihilate right and left-moving low-energy chiral fermions near momenta $k_F$ and $-k_F$, respectively.

Consider first the right-moving chiral fermion $\psi_R$, with filled momenta $p<0$. The correlation kernel is \cite{casini2009reduced}
\begin{equation}\label{eq:chiral-projector}
    C_R(x,x') = \frac{1}{2\pi}\int_{p < 0} \hspace{-0.5em}e^{ i  p(x-x')}= \frac{1}{2}\delta(x-x') - \frac{ i }{2\pi}\frac{1}{x-x'}, 
\end{equation}
where the term $\frac{1}{x-x'}$ is taken in the principal-value regularization. The eigenfunctions are Mellin waves, with eigenvalues \cite{casini2009reduced} 
\begin{equation}
    \phi_\alpha(x)= \frac{1}{\sqrt{2\pi x}}x^{i \alpha}, \qquad \qquad \lambda_\alpha = \frac{1}{1+e^{2\pi\alpha}}.
\end{equation}
Then, by Eq.~\eqref{eq:R1_free_fermions_efficient} we have
\begin{equation}\label{eq:single_chiral_1d}
\begin{aligned}
    R^{(1)}(\rho_{[0,\infty)}; \psi_R(\ell)) &= \int_\R \hspace{-0.2em} d\alpha ~|\phi_\alpha(\ell)|^2\sqrt{\lambda_\alpha(1-\lambda_\alpha)}\\
    &= \int_\R\hspace{-0.2em} d\alpha ~\frac{1}{2\pi \ell}\frac{1}{2 \cosh (\pi \alpha)}\\
    &= \frac{1}{4\pi \ell}
\end{aligned}
\end{equation}
Note that this result does not depend on $k_F$. Since the left-moving branch contributes the same amount, for the $1d$ Fermi gas on the half-line, we get the leading contribution ($k_F \ell \gg 1$)
\begin{equation}\label{eq:1d_half_line_two_chiral_fermions}
    R^{(1)}(\rho_{[0,\infty)}; c_\ell) =  \frac{1}{2\pi\ell}. 
\end{equation}
For a finite interval, with the operator at distance $\ell_L$ and $\ell_R$ from the left and right boundaries, respectively, we get
\begin{equation}
    R^{(1)}(\rho_{[-\ell_L,\ell_R]};c_0) =  \frac{1}{2\pi} \left( \frac{1}{\ell_L} + \frac{1}{\ell_R}\right).
\end{equation}
In the case where the operator is inserted at the midpoint, we get
\begin{equation}\label{eq:R1_Femi_metal_1d}
    R^{(1)}(\rho_{[-\ell,\ell]}; c_0)=  \frac{1}{\pi\ell}.
\end{equation}
This agrees with the CFT prediction once the normalization in Eq.~\eqref{eq:two_point_function_1d} is taken into account. In Fig.~\ref{fig:R1_fermi_sea_1d}, we compare this analytic result with a direct numerical computation of $R^{(1)}(\rho_A;c_\ell)$, finding excellent agreement in the regime $\frac{1}{k_F} \ll \ell \ll  L$. 

\begin{figure}[t]
    \centering
    \includegraphics[width=1\linewidth]{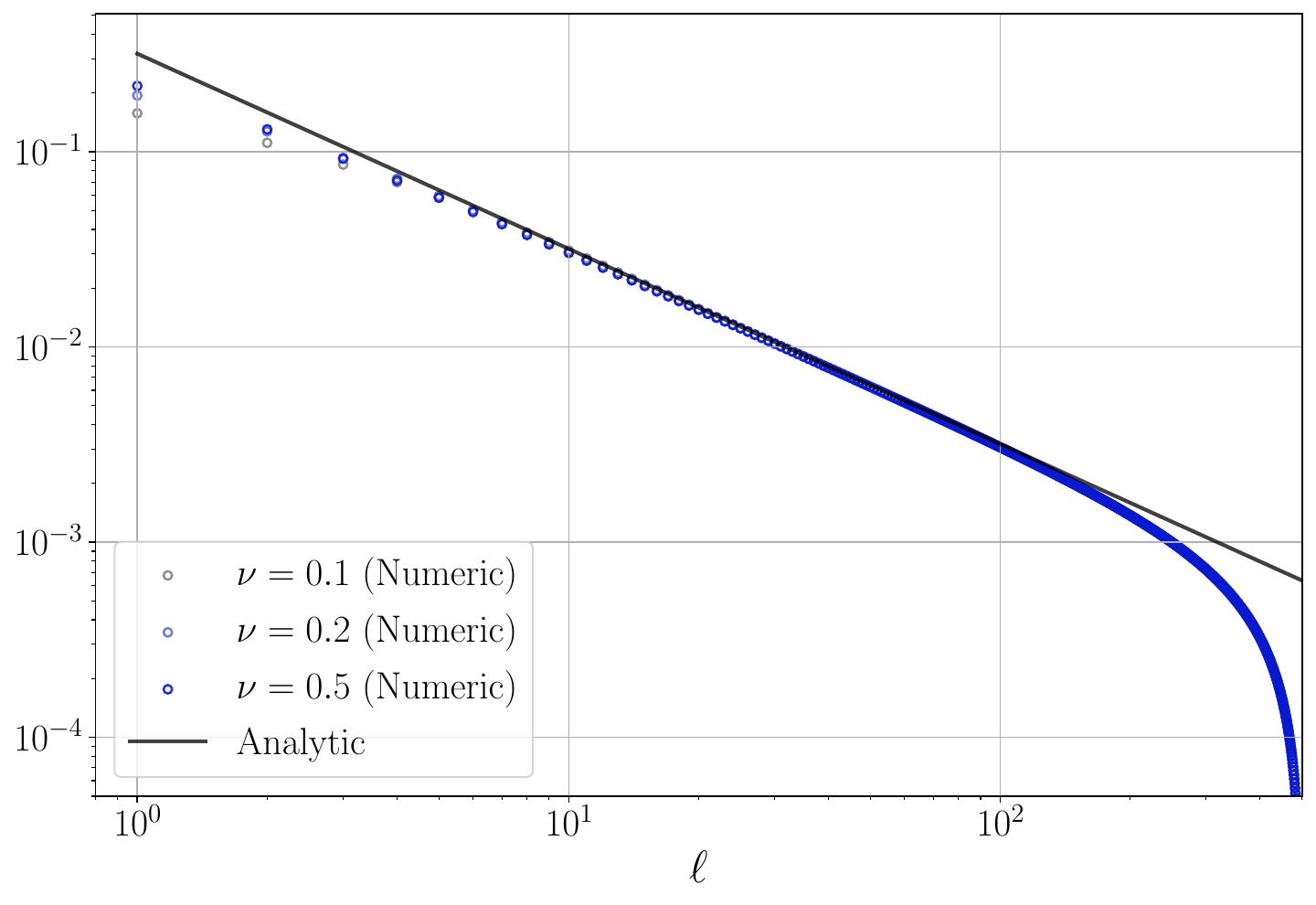}
    \caption{One-point R\'enyi-$1$ correlator, $R^{(1)}(\rho_A;c_x)=\Tr_A(\sqrt{\rho_A}\,c_x^\dagger\sqrt{\rho_A}\,c_x)$, computed for the ground state of a $1d$ free Fermi gas at various fillings $\nu$, for the interval $A= [x-\ell, x+\ell]$. For region sizes $1 \ll |A| \ll L$, the numerical data points show excellent agreement with the analytical calculation (solid line), given by Eq.~\eqref{eq:R1_Femi_metal_1d}.}  
    \label{fig:R1_fermi_sea_1d}
\end{figure}

\subsubsection{Fermi Metals in \texorpdfstring{$d\geq 2$}{}}
We now turn to Fermi metals in $d \geq 2$. In this case, the ground state is specified not only by the volume enclosed by the Fermi surface, but also by its shape. Alternatively, one may specify the single-particle dispersion $\epsilon(k)$. 

For illustration, we begin with the two-dimensional Fermi gas with a circular Fermi surface, and then generalize to arbitrary smooth Fermi surfaces in any dimension.

For a circular Fermi surface in $2d$, the two-point function is
\begin{equation}
\begin{aligned}
    \langle {\rm GS}_\nu|c_x^\dagger c_0|{\rm GS}_\nu\rangle &= \frac{k_F^2}{2\pi } \frac{J_1(k_F|x|)}{k_F |x|}\\
    &\sim \sqrt{ \frac{k_F}{2\pi^3} }\frac{\cos\left(k_F |x| - \frac{3\pi}{4}\right)}{|x|^{3/2}} 
\end{aligned}
\end{equation}
Based on this scaling, one might naively expect $R^{(1)}(\rho_A;c_x) \sim \ell^{-3/2}$. However, as we show next, this is incorrect. The reason is that the R\'enyi-$1$ correlator is determined by a continuum of effectively one-dimensional independent chiral modes. 

To see this explicitly, consider the half-plane
\begin{equation}
    A=\{ (x,y) : x \geq 0 \},
\end{equation}
and insert the operator at $(x,y)=(\ell,0)$. The half-plane is translation invariant in the $y$ direction, so the problem decomposes into independent momentum sectors parallel to $\partial A$. Fourier transforming along this direction,
\begin{equation}
    \varphi_q(x)=\int_\R\frac{d  y}{\sqrt{2\pi}}e^{- i qy}\varphi(x,y), 
\end{equation}
decomposes the correlation matrix as
\begin{equation}
    C_A =  \int_{-k_F}^{+k_F}\frac{dq}{2\pi}~C_q, 
\end{equation}
where the momentum-dependent one-dimensional kernel is
\begin{equation}
    C_q(x,x') = \frac{\sin \left(K_q(x-x')\right)}{\pi(x-x')}, \qquad 
    K_q=\sqrt{k_F^2-q^2}. 
\end{equation}

Therefore,
\begin{equation}
    R_A^{(1)}(c_\ell) = \int_{-k_F}^{k_F}\frac{dq}{2\pi}  \left[\sqrt{C_q(1-C_q)}\right]_{(\ell,\ell)},
\end{equation}
For modes satisfying $K_q \ell \gg 1$, we can use the one-dimensional result, Eq.~\eqref{eq:1d_half_line_two_chiral_fermions}. Thus, we get 
\begin{equation}\label{eq:Fermi_sea_halfplane_2d_result}
    R^{(1)}(\rho_A; c_\ell) \simeq \frac{k_F}{2\pi^2}\frac{1}{\ell}, \qquad A = \{x_\perp\geq0\} \subset\R^2
\end{equation}

We can now extend the half-plane calculation to an arbitrary smooth Fermi surface.  Consider the half-space
\begin{equation}
    A=\{ x_\perp\geq 0 \},
\end{equation}
and denote by $\hat x_\perp$ the unit normal to $\partial A$.  Decomposing momentum as $ k = k_\perp \hat x_\perp + q$, we note that $q$ is conserved by translation invariance along the boundary of $\partial A$. Thus, for each fixed $q$, the problem reduces to a one-dimensional fermion problem in the $x_\perp$ direction. Let $N(q)$ be the number of intersections of the line $K_q=\{k_\perp \hat x_\perp+q:\ k_\perp\in \mathbb R\}$ with the Fermi surface. Then, each chiral mode contributes according to Eq.~\eqref{eq:single_chiral_1d}, yielding
\begin{equation}
    R^{(1)}(\rho_A;c_\ell) \simeq \frac{1}{4\pi\ell} \int \frac{d^{d-1}q}{(2\pi)^{d-1}} N(q).
\end{equation}
We can now express this as an integral over the Fermi surface, by noting that $d^{d-1}q = \left|\hat v_F(k)\cdot \hat x_\perp\right| dS_k$, with $dS_k$ the area element on the Fermi surface and $\hat v_F(k)=v_F(k)/|v_F(k)|$ its normal. Therefore, we have 
\begin{equation}\label{eq:general_FS_formula_halfplane}
    R^{(1)}(\rho_A;c_\ell) \simeq \frac{1}{4\pi\ell} \int_{\rm FS}\hspace{-0.1em}\frac{dS_k}{(2\pi)^{d-1}}  \left|\hat v_F(k)\cdot \hat  x_\perp \right| .
\end{equation}
This is the general expression for half-space regions. In particular, for a circular Fermi surface in $2d$, we have $ \int dS_k \left|\hat v_F\cdot \hat  x_\perp \right| = k_F \int d\theta~ |\cos \theta| = 4k_F $, which reproduces Eq.~\eqref{eq:Fermi_sea_halfplane_2d_result}.

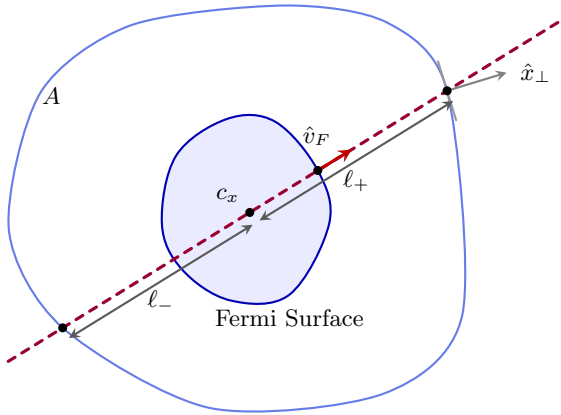
\begin{figure}[b]
\centering
\begin{tikzpicture}[scale=0.75, line cap=round, line join=round, >=stealth]

\tikzset{
    region/.style={draw=blue!85!green!60!, thick},
    fs/.style={draw=blue!70!black, fill=blue!8, thick},
    traj/.style={draw=purple!80!black, very thick, dashed},
    vel/.style={draw=red!75!black, very thick, ->},
    bnormal/.style={draw=black!50, thick, ->},
    tangent/.style={draw=black!40, thick},
    dist/.style={draw=black!65, thick, <->},
    dot/.style={circle, fill=black, inner sep=1.2pt}
}

\coordinate (cx) at (0,0);
\coordinate (k)  at (1.20,0.74);
\coordinate (bp) at ($(cx)!2.90!(k)$);
\coordinate (bm) at ($(cx)!-2.75!(k)$);
\coordinate (m)  at ($(cx)!-3.55!(k)$);
\coordinate (p)  at ($(cx)!4.55!(k)$);

\coordinate (cxl) at ($(cx)!0.018!(bm)$);
\coordinate (cxr) at ($(cx)!0.018!(bp)$);
\coordinate (off) at (0.11,-0.18);

\draw[region]

plot[smooth cycle, tension=0.5] coordinates {
    (-4.25,-0.30) (-3.70,2.10) (-2.10,3.30) (0.60,3.65)
    (2.42,3.35) (bp) (3.75,-1.80) (2.35,-3.20)
    (-0.70,-3.45) (bm)
};

\draw[fs]
plot[smooth cycle, tension=0.6] coordinates {
    (-1.55,-0.20) (-1.30,1.00) (-0.40,1.65) (0.45,1.60)
    (k) (1.4,-0.2) (0.70,-1.45) (-0.35,-1.55) (-1.20,-1.00)
};

\draw[traj] (m) -- (p);
\draw[vel] (k) -- ($(k)+(0.60,0.37)$);
\node[above left=1pt] at ($(k)+(0.42,0.26)$) {$\hat v_F$};

\draw[tangent] ($(bp)+(-0.16,0.52)$) -- ($(bp)+(0.16,-0.52)$);
\draw[bnormal] (bp) -- ($(bp)+(1.05,0.32)$);
\node[right=2pt] at ($(bp)+(1.05,0.32)$) {$\hat x_\perp$};
\draw[dist] ($(bm)+(off)$) -- ($(cxl)+(off)$) node[midway, below] {$\ell_-$};

\draw[dist] ($(cxr)+(off)$) -- ($(bp)+(off)$) node[midway, below] {$\ell_+$};

\node[dot] at (cx) {};
\node[dot] at (k) {};
\node[dot] at (bp) {};
\node[dot] at (bm) {};

\node at (-3.5,2.05) {$A$};
\node at (0.7,-1.85) {$\mathrm{Fermi~Surface}$};
\node[above left=1pt] at (cx) {$c_x$};

\end{tikzpicture}
\caption{Geometric representation of the R\'enyi-$1$ one-point formula for a Fermi metal, Eq.~\eqref{eq:general_smooth_region_FS}. For each point $k$ on the Fermi surface, the direction of the Fermi velocity $\hat v_F$ yields an effective chiral one-dimensional problem (dashed line) in real space, with the operator insertion $c_x$ at distances $\ell_\pm(k)$ from the boundary $\partial A$.}
\label{fig:Fermi_metal_general_region}
\end{figure}
Lastly, we extend the result to arbitrary smooth regions $A$. The idea is to apply the half-space result, locally along each chiral trajectory, as depicted in Fig.~\ref{fig:Fermi_metal_general_region}. More concretely, for each point $ k$ on the Fermi surface, the unit Fermi velocity $\hat v_F(k)$ defines a unique line passing through $x$ that intersects the boundary $\partial A$ at finitely many points. For simplicity, let us assume $A$ is convex and that $x$ lies in its interior so that there are only two intersections, and let $\ell_+(k)$ and $\ell_-(k)$ be the distances from $x$ to $\partial A$ along the directions $+\hat v_F$ and $-\hat v_F$, respectively. This defines a chiral $1d$ problem with endpoints at distances $\ell_\pm$. Thus, we get
\begin{equation}\label{eq:general_smooth_region_FS}
    R^{(1)}(\rho_A; c_x) \simeq \frac{1}{4\pi} \int_{\rm FS}
    \frac{dS_k}{(2\pi)^{d-1}} \left( \frac{1}{\ell_+(k)}+\frac{1}{\ell_-(k)}
    \right).
\end{equation}
This is the general smooth-region formula, which yields the leading universal behavior $ R_A^{(1)}(c_x) \sim 1/\ell$. 

For example, if $A$ is a half-space, one of the two distances is infinite, while the other is $\ell(k)=\ell/|\hat v_F(k)\cdot \hat x_\perp|$, which gives Eq.~\eqref{eq:general_FS_formula_halfplane}. 

Another particularly relevant example is a ball region of radius $\ell$ with the operator insertion at its center. In that case, regardless of the direction $\hat v_F$, we have $\ell_\pm = \ell$, and therefore
\begin{equation}\label{eq:disk_region_FS}
    R^{(1)}(\rho_A; c_0) \simeq \frac{{\rm Area(FS)}}{(2\pi)^d \ell} , \quad A = \{ |x| \leq \ell\} \subset \R^d. 
\end{equation}
The coefficient is simply given by the area of the Fermi surface, and is therefore non-universal data. In Fig.~\ref{fig:R1_fermi_disk_2d}, this result is compared with its numerical counterpart for the $2d$ square lattice at various fillings. There is a remarkable agreement, once $\ell \gg k_f^{-1}$. For small filling $\nu \lesssim 0.2 $, the Fermi surface can be effectively approximated by a circle, but for higher fillings the specific shape of its surface has to be taken into account via its area, according to Eq.~\eqref{eq:disk_region_FS}.

\begin{figure}[ht]
    \centering
    \includegraphics[width=1\linewidth]{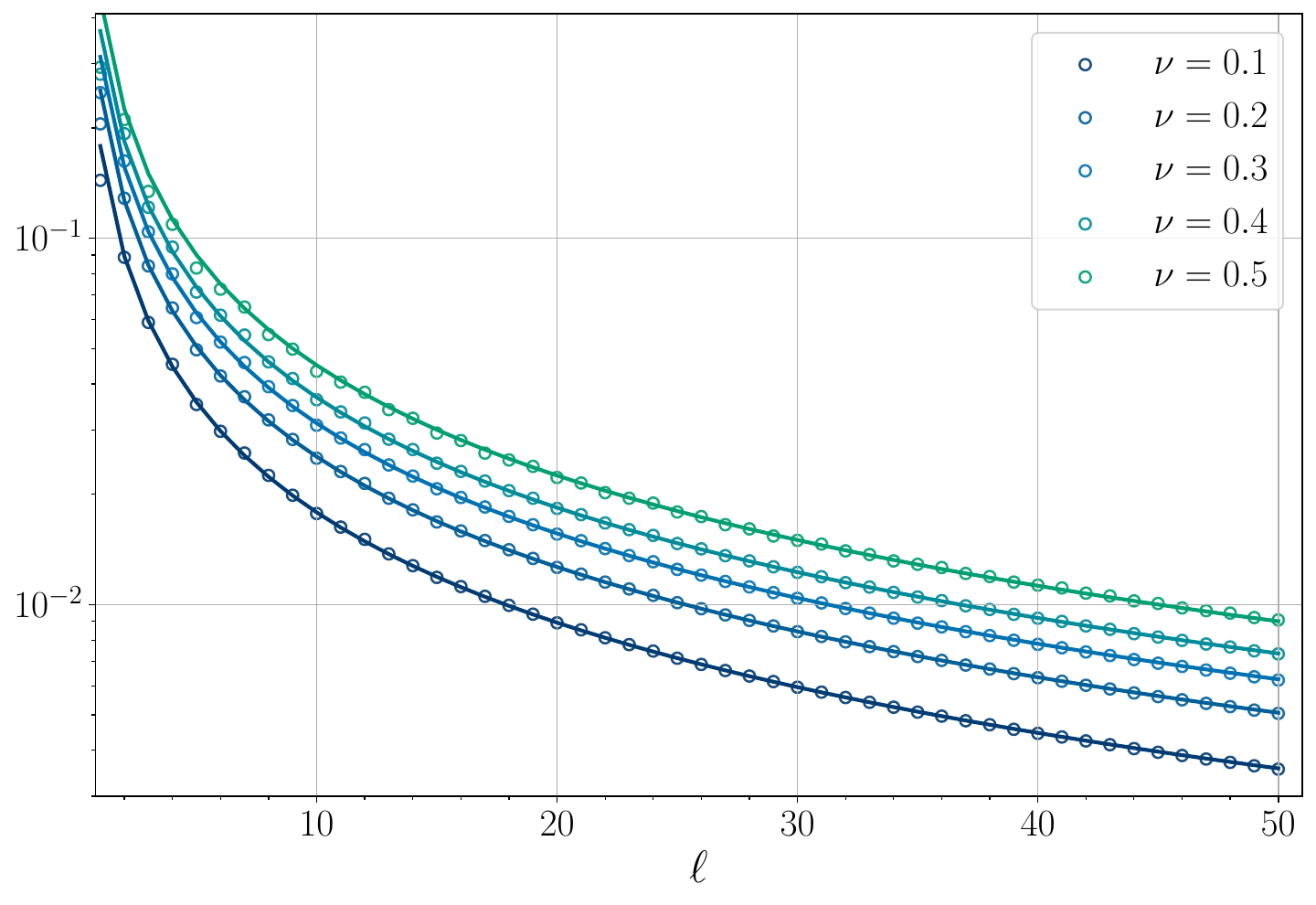}
    \caption{One-point R\'enyi-$1$ correlator, $R^{(1)}(\rho_A;c_x)=\Tr_A(\sqrt{\rho_A}\,c_x^\dagger\sqrt{\rho_A}\,c_x)$, computed for the ground state of $2d$ free Fermi gas on square lattice at various fillings $\nu$. The region $A$ is a disk of radius $\ell$ centered at the insertion point $x$. The numerical data points agree with the analytical calculation, Eq.~\eqref{eq:disk_region_FS} (with no adjustable parameters), shown as solid lines.}  
    \label{fig:R1_fermi_disk_2d}
\end{figure}

Lastly, we briefly comment on Dirac semimetals. Consider the $\pi$-flux state on the square lattice. Away from half-filling, the Fermi level intersects the bands in small Fermi pockets and, and we expect the universal behavior $R_A^{(1)}(\rho_A;c_x)\sim 1/ \ell$. At half-filling, however, the Fermi pockets shrink to isolated Dirac nodes, and the low-energy theory is instead described by massless Dirac fermions in $2+1$ dimensions. Since the fermion field has scaling dimension $\Delta_\psi=1$, Eq.~\eqref{eq:R1_CFT_scaling} gives $R_A^{(1)}(\rho_A;c_x)\sim \frac{1}{\ell^2} $. This crossover from Fermi-surface scaling away from half-filling to Dirac-CFT scaling at half-filling is confirmed numerically in Fig.~\ref{fig:R1_semimetal_2d}.

\begin{figure}[b]
    \centering
    \includegraphics[width=1\linewidth]{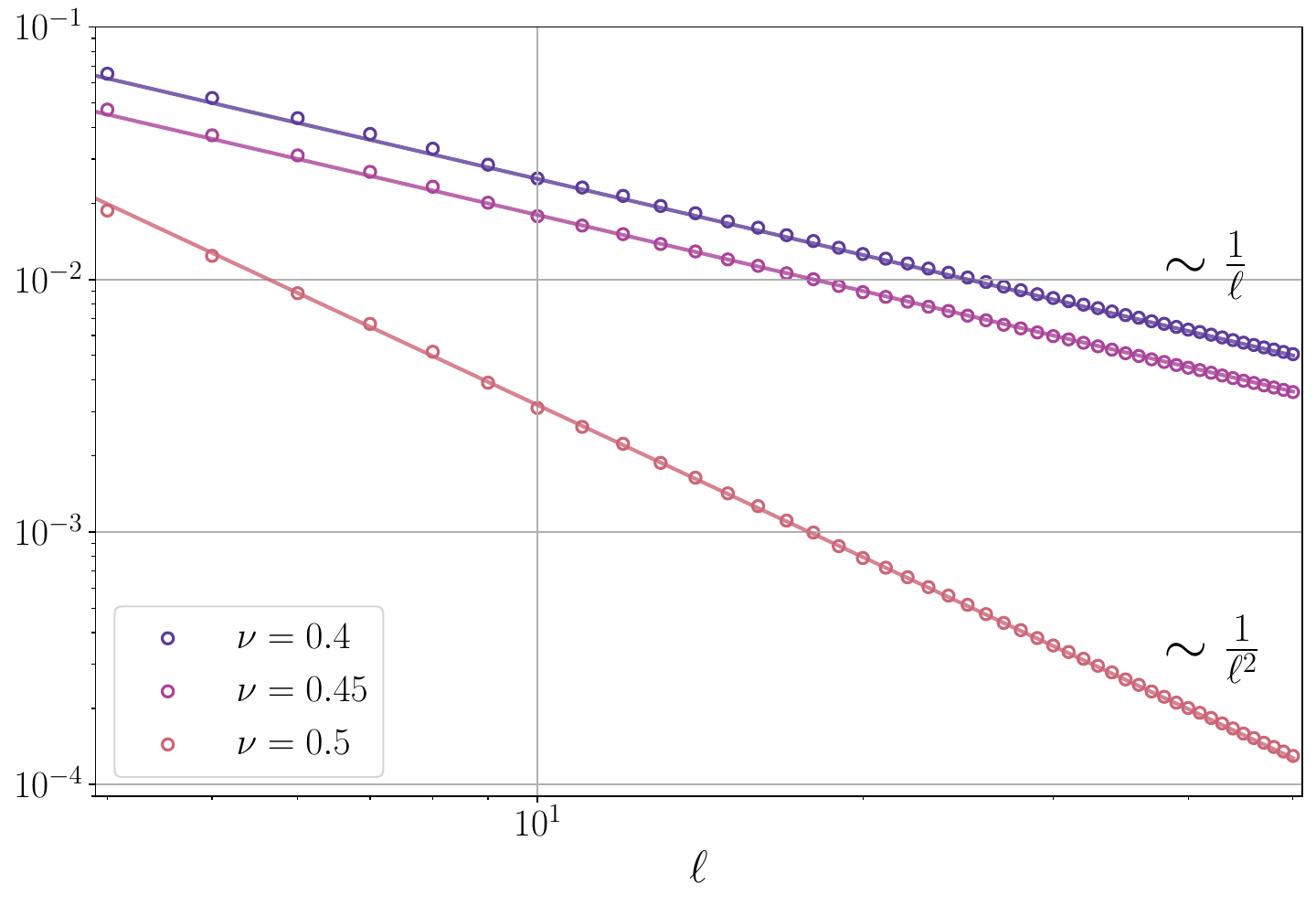}
    \caption{One-point R\'enyi-$1$ correlator, $R^{(1)} (\rho_A;c_x) = \Tr_A(\sqrt{\rho_A}\,c_x^\dagger\sqrt{\rho_A}\,c_x)$, computed for the ground state of the $\pi$-flux state on the $2d$ square lattice at several fillings $\nu$. The region $A$ a disk of radius $\ell$ centered at the insertion point $x$.}  
    \label{fig:R1_semimetal_2d}
\end{figure}

\subsubsection{General argument and diffusive metals} 
\label{sec:general_argument_diffusive_metal}
We now present a more general way of understanding the scaling of the Rényi-1 correlation in free-fermion systems. This argument encompasses metals and semimetals discussed above, and also allows us to extend the discussion to weakly disordered diffusive metals. 

We start by considering a replicated version of Eq.~\eqref{eq:R1_free_fermions_efficient},
\begin{equation}\label{eq:replicated_version_renyi_1}
    \mathcal R_{2q}(c_x) =  \left[C_A^q(1-C_A)^q\right]_{(x,x)}, 
\end{equation}
which for $q=1/2$ reduces to the R\'enyi-$1$ correlator. Note that for other values of $q$ it does not correspond to the R\'enyi-$2q$ correlator. Instead, for integer $q$, this quantity admits a clean interpretation, which will allow us to determine its asymptotic scaling.

Let us write
\begin{equation}
    C_A= P_A P_F P_A ,
\end{equation}
where $P_A$ denotes the single-particle projector onto region $A$, and $P_F$ denotes the single-particle projector onto the occupied states of the ground state (e.g. states inside the Fermi surface). Since $P_F^2=P_F$, for $x\in A$ and $q =1$, we have
\begin{equation}\label{eq:free_fermions_R_2_leading_decay}
\begin{aligned}
    \mathcal R_2(c_x)  &= \left[C_A(1-C_A)\right]_{(x,x)}  \\
              &= \left[P_A P_F P_{\bar A} P_F P_A\right]_{(x,x)} \\   
              &= \int_{\bar A} d^dy\, |C(x,y)|^2 .
\end{aligned}
\end{equation}
In other words, this corresponds to a sum of processes where the particle propagates from $x$ to an arbitrary point $y$ outside $A$, and then returns to $x$.
This expression generalizes to integer $q$, 
\begin{equation}
    \mathcal R_{2q}(c_x) =\left[  \left(P_F P_{\bar A}P_F P_A\right)^q \right]_{(x,x)}.
\end{equation}
This corresponds to processes where the particle starts at $x$, goes back and forth between $A$ and $\bar A$ a total of $q$ times, and then returns to $x$. Crucially, when $x$ is far from the boundary, the leading decay is controlled by the initial escape and the final return, and hence we expect that the universal scaling of $\mathcal R_{2q}$ (and in particular of the R\'enyi-$1$ correlator) is captured by $R_2$, Eq.~\eqref{eq:free_fermions_R_2_leading_decay}.

To make this intuition more precise, we write Eq.~\eqref{eq:replicated_version_renyi_1} in the basis that diagonalizes $C_A$, Eq.~\eqref{eq:free_fermions_diagonalize}, and introduce the local
density of states
\begin{equation}
    \rho_x(\epsilon)= \sum_{\alpha}\delta(\epsilon-\epsilon_\alpha) |\phi_\alpha(x)|^2 ,
\end{equation}
where we parametrized the spectrum of $C_A$ in terms of the effective single-particle spectrum, $\lambda_\alpha=\frac{1}{1+e^{\epsilon_\alpha}}$. Then, we find
\begin{equation}
    \mathcal R_{2q}(c_x) = \int_{-\infty}^{\infty} d\epsilon\, \frac{\rho_x(\epsilon)}  {(2\cosh(\epsilon/2))^{2q}} .
\end{equation}

The kernel ${(2\cosh(\epsilon/2))^{-2q}}$ strongly localizes the integral at $\epsilon =0$. Then, under the assumption that the local density of states $\rho_x(\epsilon)$ is smooth on the scale $\epsilon=O(1)$, we may approximate $\rho_x(\epsilon) \approx \rho_x(0)$, yielding the leading contribution
\begin{equation}
    \mathcal R_{2q}(c_x) \simeq \frac{\Gamma(q)^2}{\Gamma(2q)}\,\rho_x(0).
\end{equation}
Higher corrections can be systematically computed in a Sommerfeld-type expansion, by Taylor-expanding $\rho_x(\epsilon)$ around $\epsilon=0$, with only even derivatives of $\rho_x(\epsilon)$ contributing. Then, the leading behavior of $\rho_x(0)$ can then be obtained from Eq.~\eqref{eq:free_fermions_R_2_leading_decay}, and we get
\begin{equation}
    R^{(1)}(\rho_A;c_x) \simeq \pi \int_{\bar A} d^dy\, |C(x,y)|^2 .
\end{equation}
This expression corresponds to the leading scaling contribution to the R\'enyi-$1$ correlator for free-fermion systems.

We can now revisit previous examples. For a clean metal in $d$ spatial dimensions, we have, up to oscillatory factors, $|C(x,y)|^2 \sim \frac{1}{|x-y|^{d+1}}$, which yields the universal scaling $R^{(1)}(\rho_A;c_x)\sim \frac{1}{\ell}$. In contrast, for a Dirac semimetal, $|C(x,y)|^2\sim \frac{1}{|x-y|^{2d}}$, and we find $R^{(1)}(\rho_A;c_x)\sim \frac{1}{\ell^d}$, in accordance with the CFT result.

Lastly, we discuss a diffusive metal. In this case, the disorder-averaged two-point function $\overline{C(x,y)}$ decays exponentially beyond the elastic mean free path, but the second-moment has a power law decay \cite{Altland_Simons_2010}
\begin{equation}
    \overline{|C(x,y)|^2} \sim \frac{1}{|x-y|^{d+2}}
\end{equation}
in the diffusive regime. Then,
\begin{equation}
    \overline{R^{(1)}(\rho_A;c_x)} \sim \frac{1}{\ell^2}.
\end{equation}
Therefore, weak disorder changes the clean Fermi-surface scaling from $\ell^{-1}$ to $\ell^{-2}$ in the diffusive regime.

To numerically verify this, we consider the Anderson model on the square lattice,
\begin{equation}\label{eq:Anderson_model_2d}
    H = -\sum_{\langle ij\rangle} (c_i^\dagger c_j +  c_j^\dagger c_i )
    + \sum_i \epsilon_i c^\dagger_i c_i ,
\end{equation}
where the on-site potentials $\epsilon_i$ are independent random variables drawn uniformly from $[-W/2,W/2]$. In $2d$ dimensions, all single-particle eigenstates are localized in the thermodynamic limit, but at weak disorder, the localization length can be parametrically larger than the mean-free path. Consequently, on intermediate length scales $\ell_e \ll \ell \ll \xi_{\rm loc}$, local observables exhibit diffusive metallic behavior~\cite{Mirlin_Anderson_transitions}.  In Fig.~\ref{fig:Anderson_model}, we verify that the R\'enyi-$1$ correlator exhibits approximate diffusive metallic behavior, $\overline{R^{(1)}(\rho_A;c_x)} \sim 1/\ell^2$.

\begin{figure}
    \centering
    \includegraphics[width=0.99\linewidth]{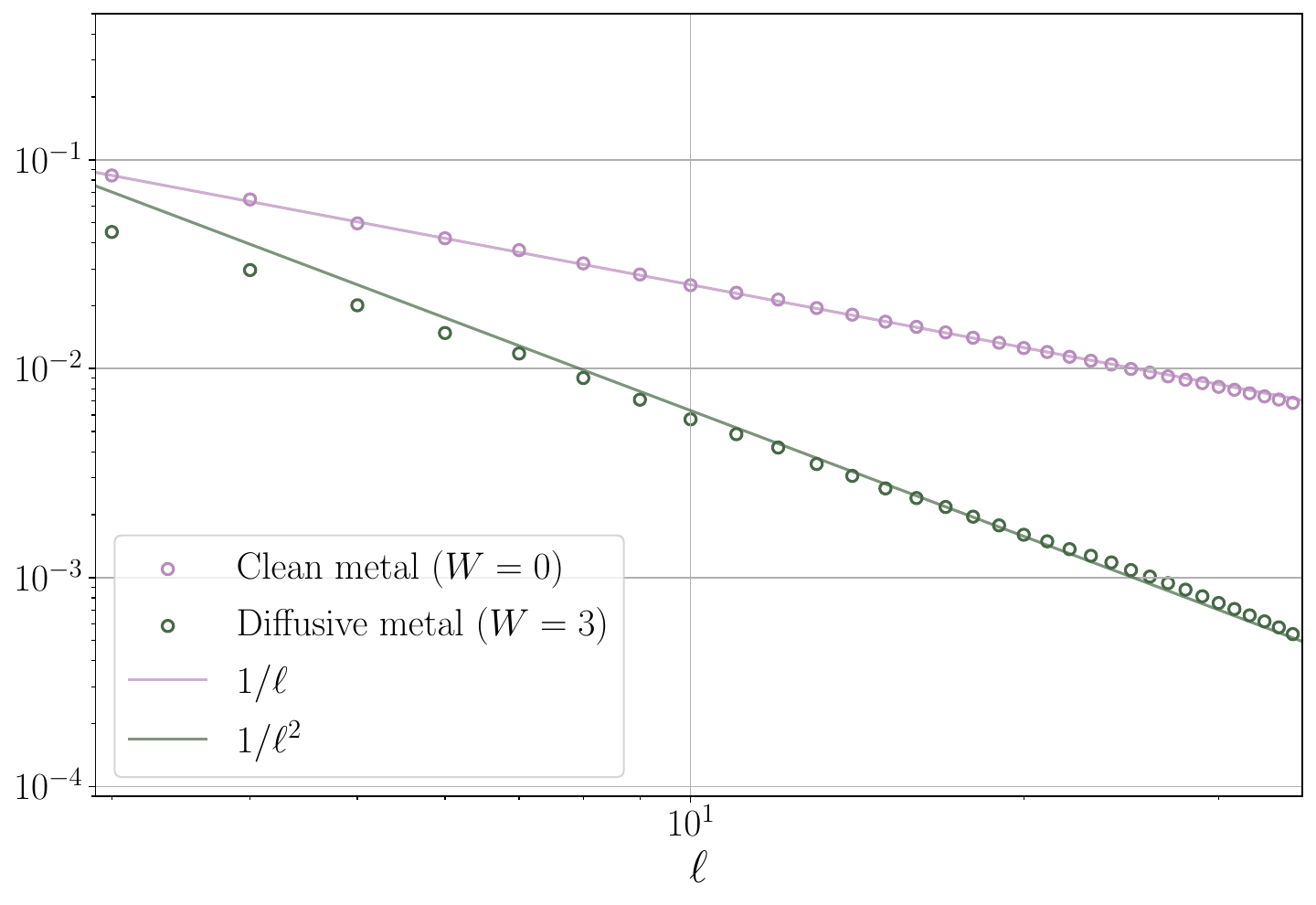}
    \caption{One-point R\'enyi-$1$ correlator, $R^{(1)}(\rho_A;c_x) = \Tr_A(\sqrt{\rho_A}\,c_x^\dagger\sqrt{\rho_A}\,c_x)$, computed for the ground state of the $2d$ Anderson model on the square lattice, Eq.~\eqref{eq:Anderson_model_2d}, at filling $\nu=0.2$. The region $A$ is a disk of radius $\ell$ centered at the insertion point $x$. For disorder bandwidth $W=3$, the R\'enyi-$1$ correlator exhibits the diffusive-metal scaling $\overline{R^{(1)}(\rho_A;c_x)}\sim \ell^{-2}$, in contrast to the clean limit ($W = 0$), which follows the Fermi-metal scaling $R^{(1)}(\rho_A;c_x)\sim \ell^{-1}$.}
    \label{fig:Anderson_model}
\end{figure}

\subsection{Random Gaussian states}
\label{sec:Gaussian}

In this section, we consider a different kind of free fermion system that provides a realization in close spirit of the mechanism discussed in Sec.~\ref{sec:localbutnotglobal}. Consider a random quadratic fermion Hamiltonian, given by
\begin{equation}\label{eq:Hamiltonian_fermions_GOE}
    H=\sum_{ij=1}^{N} J_{ij} c_i^\dagger c_j ,
\end{equation}
where $J=J^T$ is sampled from the Gaussian Orthogonal Ensemble (GOE), with zero mean and variance $\overline{J_{ij}^2}=J^2/N$ for $i\neq j$, up to the conventional factor of two for diagonal entries. This Hamiltonian is free, and its eigenstates are labeled by the occupation of the single-particle states. For simplicity, we will consider only the ground state in the presence of a uniform chemical potential. 

For finite filling $\nu = \frac{N_{\rm occ}}{N}$, Ref.~\cite{Magan2016RandomFreeFermions} showed that the correlation function of eigenstates coincides with their corresponding thermal expectation values, with temperature fixed by $\nu =\overline{ n_\beta}$, with $n_\beta$ the Fermi distribution. In this sense, the reduced density matrix is thermal and thus provides an explicit realization of the ETH mechanism described in Sec.~\ref{sec:localbutnotglobal}. Note that the Gaussian state does not satisfy the standard (stronger) sense of ETH as it has extensive amount of conserved quantities. But for our purpose, we only need $\rho_A$ to be thermal for $\ell_A/L\to0$, which is satisfied by the random Gaussian state.

More precisely, for an eigenstate of finite filling $\nu$, the one-point fidelity is
\begin{equation}
    F(\rho_A;c_x) = \sqrt{\nu(1-\nu)},
\end{equation}
since disorder averaging gives $\overline{C_{ij}} = \nu \delta_{ij}$ (assuming self-averaging). In contrast, the global two-point fidelity reduces to an ordinary two-point function,
\begin{equation}
    F(|\Psi\rangle\langle\Psi|;c_x^\dagger c_y) = \left|\langle\Psi| c_x^\dagger c_y |\Psi\rangle \right|, 
\end{equation}
which for $x\neq y$, vanishes after disorder averaging, since $\overline{|C_{xy}|^2}= \frac{\nu}{N}$. 

Therefore, GOE free-fermion eigenstates exhibit local but not global SW-SSB: for $1\ll |A|\ll N$ the reduced state is locally thermal and has finite one-point LFC, but the state is pure and does not have global fidelity order. This behavior is confirmed numerically in Fig.~\ref{fig:R1_fGOE}.

\begin{figure}[t]
    \centering
    \includegraphics[width=0.99\linewidth]{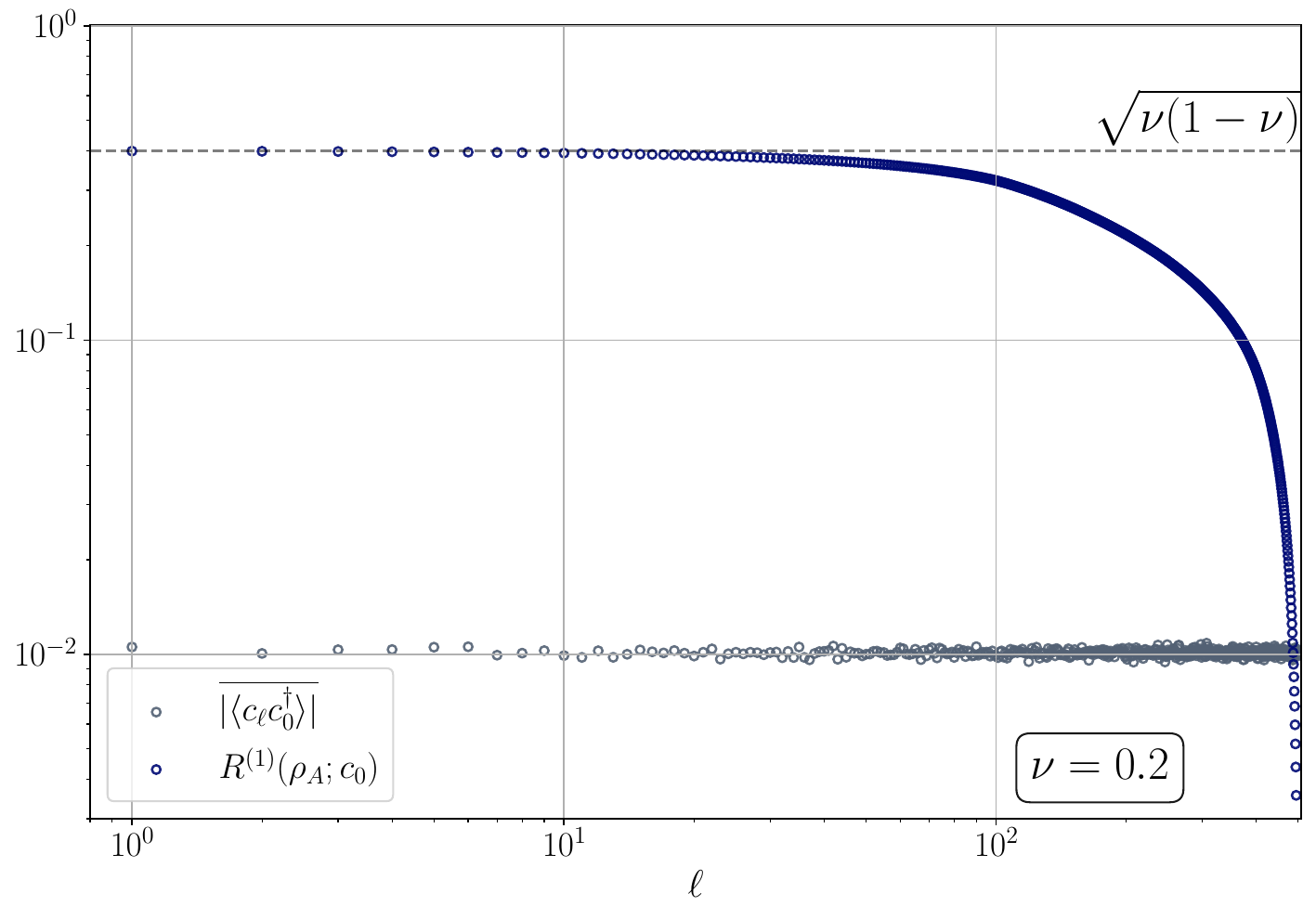}
    \caption{One-point R\'enyi-$1$ correlator, $R^{(1)}(\rho_A;c_0) = \Tr_A(\sqrt{\rho_A}\,c_0^\dagger\sqrt{\rho_A}\,c_0)$, computed numerically for the ground state of the random quadratic fermion Hamiltonian, Eq.~\eqref{eq:Hamiltonian_fermions_GOE}, at filling $\nu = 0.2$, using the interval $A = [-\ell, \ell]$. For $\ell \ll L = 1000$, the R\'enyi-$1$ correlator plateaus at the value $\sqrt{\nu(1-\nu)}$ (dashed line), in contrast to the vanishing disorder-averaged two-point function, $\overline{|\langle c_\ell c_0^\dagger\rangle|}$.}   
    \label{fig:R1_fGOE}
\end{figure}

\section{Non-Abelian Generalization}\label{sec:non_abelian_definition}
In this section, we extend the local formulation of SW-SSB to non-Abelian symmetry groups. We consider a finite (or compact) group $G$ with unitary representation $U_g$, and let $\{O_x^{(\alpha)}\}_{\alpha=1}^n$ be a local operator multiplet transforming in a non-trivial $n$-dimensional irreducible representation (irrep) $D$; that is,
\begin{equation}\label{eq:Non_abelian_irrep_D}
    U_g O_x^{(\alpha)} U_g^\dagger = \sum_{\beta=1}^n D_{\alpha\beta}(g) O_x^{(\beta)}.
\end{equation}

A natural generalization of the local fidelity correlator is $F(\rho_A; O_x^{(\alpha)})$. However, unlike in the Abelian case, a single component of an operator multiplet is basis-dependent, so we must ensure that the characterization of the mixed state is independent of this choice of basis. One way to achieve this is by maximizing over all normalized linear combinations within the irrep multiplet.

Therefore, we define the local fidelity correlator
\begin{equation}\label{def:nonabelian_one_point}
    \max_{|v| = 1} F(\rho_A;v\cdot O_x)
    \defeq \max_{|v| = 1}
     F\left(\rho_A,\sum_\alpha v_\alpha O_x^{(\alpha)}\right).
\end{equation}
where the optimization is with respect to complex unit-norm vectors $v$, $|v| = \sum_\alpha |v_\alpha|^2 = 1$. This correlator is manifestly basis-independent, since any change of basis $O_\alpha \to \sum_\beta U_{\alpha\beta} O_\beta$ can be absorbed into a corresponding rotation of the vector $v$ over which we optimize.

Then, we define local SW-SSB as follows.
\begin{definition}[Non-Abelian SW-SSB]\label{def:local_SW-SSB_non_abelian}
A state $\omega$ has \textbf{strong-to-weak symmetry breaking} (in the local sense) if, for some non-trivial irreducible representation $O_x^{(\alpha)}$, the following two conditions are met:
\begin{enumerate}
    \item The local one-point fidelity measure of SW-SSB is non-vanishing: 
    \begin{equation}
     F(\omega; O_x ) \eqdef \lim_{|A|\to \infty} \max_{|v| = 1} F(\rho_A;v \cdot O_x)>0.   
    \end{equation}
    \item The state does not exhibit ordinary SSB: $ \lim_{|x-y| \to \infty} \langle O_x^{(\alpha)} O_y^{(\beta)}\rangle_\omega = 0$ for all $\alpha, \beta$.
\end{enumerate}
\end{definition}

The following result provides an equivalent expression for the order parameter solely in terms of purifications and, in particular, does not involve any explicit optimization over $v$.
\begin{lemma}\label{lemma:max_v_equivalent_definition}
  Let $O_x^{(\alpha)}$ transform in an $n$-dimensional irreducible representation of some group $G$. Then, for any purification $|\Psi \rangle$ of the mixed state $\rho_A$, we have
    \begin{equation}
    \max_{|v| = 1}~F(\rho_A; v \cdot O) = \max_{U\in \mathcal{U}(\mathcal{H}_A')} \left(\sum_{\alpha = 1}^n \left| \langle \Psi | U \otimes O_x^{(\alpha)}|\Psi \rangle \right|^2\right)^{1/2}.
    \end{equation}
\end{lemma}
\begin{proof}
    By Uhlmann's theorem, we have
     \begin{equation}
     \begin{aligned}
         \max_{|v| = 1}~F(\rho; v \cdot O) &= \max_{|v| = 1}\max_{U\in \mathcal{U}(\mathcal{H}_A')}  \langle \Psi | U \otimes \sum_\alpha v_\alpha O_x^{(\alpha)}|\Psi \rangle \\
         &= \max_{U\in \mathcal{U}(\mathcal{H}_A')} \max_{|v| = 1} \sum_\alpha v_\alpha  \langle \Psi | U \otimes O_x^{(\alpha)}|\Psi \rangle \\
         &= \max_{U\in \mathcal{U}(\mathcal{H}_A')} \left( \sum_\alpha    \left| \langle \Psi | U \otimes O_x^{(\alpha)}|\Psi \rangle \right|^2 \right)^{1/2}
     \end{aligned}
     \end{equation}
     where we took the maximum over $v$ as follows: $\max_v \sum_\alpha v_\alpha u_{\alpha} = \sum_\alpha \frac{u_\alpha^*}{\sqrt{\sum_\alpha |u_\alpha|^2}} u_\alpha = \sqrt{\sum_\alpha |u_\alpha|^2}$, with $u_\alpha = \langle \Psi | U \otimes O_x^{(\alpha)}|\Psi \rangle $. 
\end{proof}

Next, we present several equivalent characterizations of SW-SSB. We show that although the individual one-point correlator $F(\rho; O_x^{(\alpha)})$ depends on the choice of basis, the statement that at least one such correlator is nonzero does not. Therefore, if the goal is simply to diagnose SW-SSB, the existence of a single nonvanishing one-point correlator already provides a sufficient condition. Then, we show that the definition of SW-SSB based on Eq.~\eqref{def:nonabelian_one_point} is equivalent to the diagnostic based on the kinematically natural completely positive map $\tau_O \defeq \frac{1}{n} \sum_{\alpha=1}^n O_x^{(\alpha)} [\cdot] (O_x^{(\alpha)})^{\dagger} $.

\begin{theorem}[Equivalent characterizations]\label{thm:nonabelian_equiv}
Let $O_x^{(\alpha)}$ transform in an $n$-dimensional irreducible representation of a group $G$. Then, the following are equivalent:
\begin{enumerate}
    \item The state has a non-vanishing local SW-SSB measure for the multiplet,  
    \begin{equation}
            \max_{|v| = 1}F(\rho;v \cdot O_x)>0,
    \end{equation}
    \item In any basis $O_x^{(\alpha)}$, there exists at least one component $\alpha$ such that
    \begin{equation}
        \lim_{|A|\to\infty} F(\rho_A;O_x^{(\alpha)})>0.
    \end{equation}
    \item The fidelity
    \begin{equation}
        \lim_{|A| \to \infty} F(\rho_A, \tau_O[\rho_A]) > 0, 
    \end{equation}
    where $\tau_O = \frac{1}{n} \sum_{\alpha=1}^n O_x^{(\alpha)} [\cdot] (O_x^{(\alpha)})^{\dagger} $.
\end{enumerate}
\end{theorem}

\begin{proof}
    First, we show the equivalence between 1 and 2. 
    Assume that $\lim_{|A| \to \infty} ~F(\rho; v \cdot O_x)  = c> 0 $ for some unit-norm $v$. Then, by Lemma~\ref{lemma:max_v_equivalent_definition}, for all $A$, there exists $U\in \mathcal{U}(\mathcal{H}_A')$ such that $\sum_{\alpha = 1}^n \left| \langle \Psi_A | U \otimes O_x^{(\alpha)}|\Psi_A \rangle \right|^2 \geq  c^2$, for any basis $O_x^{(\alpha)}$. This implies that  $n \max_\alpha F(\rho_A; O_x^{(\alpha)})^2 \geq   \sum_{\alpha = 1}^n F(\rho_A; O_x^{(\alpha)})^2  \geq  c^2$. We now take the covering limit, which commutes with the optimization over $\alpha$ because $n$ is finite. Therefore, for any basis, there is at least one element such that $ \lim_{|A| \to \infty }F(\rho_A;  O_x^{(\alpha)})  \geq \frac{c}{\sqrt{n}} > 0$.
    The converse direction is immediate by noting that  $\max_v F(\rho_A; v \cdot O_x) \geq F(\rho_A;  O_x^{(\alpha)}) $ for all bases and all $\alpha$.  

    Now, we show the equivalence between 2 and 3. To do that, we will show that there exists $\alpha$ such that 
    \begin{equation}
    \frac{1}{\sqrt{n}} F(\rho; O_x^{(\alpha)})\leq F(\rho,\tau_O[\rho]) \leq  \sqrt{n} ~F(\rho; O_x^{(\alpha)})    
    \end{equation}
    Write 
    \begin{equation}
         F(\rho_A, \tau_O[\rho_A] ) = \Tr\sqrt{\frac{1}{n}\sum_\beta X_\beta X_\beta^\dagger }
    \end{equation}
    where $X_\beta = \sqrt{\rho_A} O_x^{(\beta)} \sqrt{\rho_A}$ is a positive operator. Thus, we can bound
    \begin{equation}
      \frac{1}{\sqrt{n}} \Tr\sqrt{X_\alpha X_\alpha^\dagger } \leq \Tr\sqrt{\frac{1}{n}\sum_\beta X_\beta X_\beta^\dagger }
    \end{equation}
    for all $\alpha$. Additionally, for positive operators, we have $\sqrt{P+Q} \leq \sqrt{P} + \sqrt{Q}$, and we can bound
     \begin{equation}
     \begin{aligned}
    \Tr\sqrt{\frac{1}{n}\sum_\beta X_\beta X_\beta^\dagger } &\leq \frac{1}{\sqrt n}\sum_\beta  \Tr\sqrt{ X_\beta X_\beta^\dagger }\\
    &\leq \sqrt n \max_\alpha  \Tr\sqrt{ X_\alpha X_\alpha^\dagger }
    \end{aligned}
    \end{equation}
    Therefore, we have 
    \begin{equation}
    \frac{1}{\sqrt{n}} \max_\alpha F(\rho_A; O_x^{(\alpha)})\leq F(\rho_A,\tau_O[\rho_A]) \leq  \sqrt{n} \max_\alpha F(\rho_A; O_x^{(\alpha)}) .
    \end{equation}
    Since there are finitely many $\alpha$, we can commute the covering limit with the optimization over $\alpha$, yielding the desired result.
\end{proof}
As an immediate corollary, we show that the non-Abelian version of local SW-SSB enjoys a stability theorem.
\begin{corollary}[Stability theorem]
If a mixed state $\rho$ has SW-SSB in the one-point fidelity sense and $\E_{\rm SFD}$ is a strongly symmetric finite-depth local quantum channel, then $\E_{\rm SFD}[\rho]$ has SW-SSB in the one-point fidelity sense. 
\end{corollary}
\begin{proof}
Let $O_x^{(\alpha)}$ be a charged operator in some basis such that $F(\rho; O_x^{(\alpha)})  > 0$. Then, the proof of Theorem \ref{thm:stability} proceeds exactly as before. In particular, we can consider $\tilde O_x^{(\alpha)} =  U O_x^{(\alpha)}U^{-1} $, which transforms in the same irrep as $O_x^{(\alpha)}$, since $[U_g, U] = 0$. Then, we find $\tilde O_{x'}^{(\alpha)}$ such that $F(\Tr_{\bar A}\E_{\rm SFD}[\rho]; \tilde O_{x'}^{(\alpha)}) = O(1)$. 
\end{proof}

\subsection{Two-point fidelity}
In this section, we compare the one-point and two-point formulations in the non-Abelian setting. 

Let us first note that a covariant formulation is given by the singlet
\begin{equation}
    O_{xy} \defeq \frac{1}{n}\sum_{\alpha=1}^n O_x^{(\alpha)} {O_y^{(\alpha)}}^\dagger .
\end{equation}
Not only is $O_{xy}$ basis-independent, but if $\rho$ is strongly symmetric, we have
\begin{equation}
    F(\rho; O_x^{(\alpha)} O_y^{(\beta)\dagger}) =  \delta_{\alpha\beta} F (\rho; O_{xy}),  
\end{equation}
as can be seen from the orthogonality relations for the irrep $D$, defined in Eq.~\eqref{eq:Non_abelian_irrep_D}. Therefore, the two-point fidelity measure of (global) SW-SSB is simply given by $F(\rho; O_{xy})$. 

First, we show that global SW-SSB implies local SW-SSB.
\begin{theorem}\label{thm:non_abelian_first_inequality_of_measures}
For any irreducible representation $\{O_x^{(\alpha)}\}_\alpha$, we have
\begin{equation}
\max_\alpha \norm{O_y^{(\alpha)}}_\infty ~ \max_{|v|=1} F(\rho_A; v \cdot O) \ge F(\rho;O_{xy}).
\end{equation}
\end{theorem}
\begin{proof}
Note that
\begin{equation}
\begin{aligned}
     F(\rho;O_{xy}) &= \norm{\sqrt{\rho} \frac{1}{n}\sum_{\alpha} O^{(\alpha)}_x {O_y^{(\alpha)}}^\dagger \sqrt{\rho}}_1 \\
     &\leq \frac{1}{n}\sum_{\alpha} \norm{\sqrt{\rho}  O^{(\alpha)}_x {O_y^{(\alpha)}}^\dagger \sqrt{\rho}}_1   \\
     &= \frac{1}{n}\sum_{\alpha} F(\rho; O^{(\alpha)}_x {O_y^{(\alpha)}}^\dagger )
\end{aligned}
\end{equation}
Then, we apply Eq.~\eqref{eq:first_inequality}, and we find
\begin{equation}
\begin{aligned}
     F(\rho;O_{xy}) &\leq \max_\alpha \norm{{O_y^{(\alpha)}}^\dagger}_\infty ~\frac{1}{n}\sum_{\alpha}  F(\rho; O^{(\alpha)}_x )\\
\end{aligned}
\end{equation}
Finally, we bound $F(\rho; O^{(\alpha)}_x) \leq \max_v  F(\rho; \sum_\alpha v_\alpha O^{(\alpha)}_x)$ to arrive at the desired result.
\end{proof}

As discussed in Sec.~\ref{sec:localbutnotglobal}, the converse does not hold in general: local does not imply global SW-SSB. However, as we show next, the local characterization in terms of local one- and two-point fidelities is equivalent; that is the non-Abelian analogue of Theorem \ref{theorem:average_one-point_two-point} holds. 
\begin{theorem}\label{thm:non_abelian_second_inequality_of_measures}
Let $A$ be a finite region, and let $\{O_x^{(\alpha)}\}_{x \in \Omega }$ be a collection of charged multiplets transforming in the same $n$-dimensional irreducible representation of some group $G$. Then, for any unit-norm vector $v$, we have
\begin{equation}
    \sum_{xy \in \Omega} F(\rho; O_{xy}) \geq \frac{1}{n}\left( \sum_{x\in \Omega}  F(\rho; v\cdot O_{x})\right)^2
\end{equation}
\end{theorem}
\begin{proof}   
    Let $\ket{\Psi} \in \Hilb_A \otimes \Hilb'_A$ be any purification of $\rho_A$, with $\dim(\Hilb_A') \geq \dim(\Hilb_A)$. By Uhlmann's theorem applied to $F(\rho_A, v \cdot O_x)$, we have
    \begin{equation}
        F(\rho_A, v\cdot O_x) = \sum_\alpha v_\alpha \braket{\Psi | O_x^{(\alpha)} \otimes U_x | \Psi} = \sum_\alpha v_\alpha \langle \Psi |  \Phi_x^{(\alpha)} \rangle.
    \end{equation}
    Then, by Uhlmann's theorem applied to $F(\rho_A; O_{xy})$, we have
    \begin{equation}
    \begin{aligned}
       F(\rho_A; O_{xy}) & = \max_{U \in \mathcal{U}(\Hilb'_A)} \left|\frac{1}{n} \sum_\alpha \langle\Psi | O_x^{(\alpha)} O_y^{(\alpha)^\dagger} \otimes U | \Psi \rangle \right| \\
        & \geq 
        \frac{1}{n} \sum_\alpha \langle\Psi | O_x^{(\alpha)} O_y^{(\alpha)^\dagger} \otimes U_x U_y^\dagger | \Psi \rangle  \\
        & =   \frac{1}{n} \sum_\alpha \langle \Phi_y^{(\alpha)} | \Phi_x^{(\alpha)} \rangle,
    \end{aligned}
    \end{equation}
    Averaging over $x, y \in \Omega$ in the final equation above, we have
    \begin{align}
        \sum_{x,y \in \Omega} F(\rho_A; O_{xy}) & \geq \frac{1}{n} \sum_{\alpha} \langle \Phi^{(\alpha)} |  \Phi^{(\alpha)} \rangle,
    \end{align}
    where we have defined the vector $|\Phi^{(\alpha)} \rangle  = \sum_{x \in \Omega} | \Phi_x^{(\alpha)} \rangle $. Lastly, using the Cauchy-Schwarz inequality,
    \begin{equation}
    \begin{aligned}
        \sum_\alpha \langle \Phi^{(\alpha)} | \Phi^{(\alpha)} \rangle & \geq  \sum_\alpha \left|\langle \Psi | \Phi^{(\alpha)} \rangle \right|^2\\
         & \geq  \left| \sum_\alpha v_\alpha \langle \Psi | \Phi^{(\alpha)} \rangle \right|^2\\
        & \geq  \left|\sum_x   F(\rho_A, v \cdot O_x)  \right|^2,
    \end{aligned}
    \end{equation}
    which establishes the theorem.
\end{proof}

\section{Discussion}
\label{sec:discussion}

In this work, we introduced the notion of local strong-to-weak spontaneous symmetry breaking (SW-SSB) Eq.~\eqref{eq:FOP}, based on the local fidelity (or R\'enyi) correlator, Eq.~\eqref{eq:LFC}. The main advantage, compared to the previous framework based on global fidelity or R\'enyi correlators, is that only reduced density matrices on local regions are required. Conceptually, this local feature allowed the notion of SW-SSB to be defined directly in the thermodynamic limit, putting the notion on a more solid theoretical ground. More practically, the local notion of SW-SSB resolved previous challenge on scalability: in the absence of any prior knowledge of the state, previous global measures require exponential amount of measurement resource. In contrast, for a $d$ dimensional system with linear size $L$, with ${\rm Poly}(L)$ amount of resource, our local measure can detect SW-SSB up to length scale $O[(\log L)^{1/d}]$.  

In practice, the state in a physical system may have certain physical structures that allow one to measure the SW-SSB order more efficiently. This is the case in the recent experimental detection of SW-SSB order in a dephased cold Fermi gas~\cite{SWSSBExp}: the initial state was close to a free-fermion Gaussian state, which allowed an efficient variational estimation of the R\'enyi correlators. Even in such systems, we expect our local fidelity correlator to allow for more efficient detection of the SW-SSB order, especially as the system size scales up. 

We end by discussing some future directions:

\begin{itemize}
    \item \textit{Connection to thermalization}: in Sec.~\ref{sec:localbutnotglobal} we discussed eigenstate-thermalized states as an example with local, but not global, SW-SSB order. Earlier results also suggested strong connections between SW-SSB and thermalization, including SW-SSB in Gibbs thermal states~\cite{lessa2025strong} and emergent hydrodynamics from $U(1)$ SW-SSB in steady states~\cite{huang_hydrodynamics_2025}. This emerging picture suggests a deeper relation between SW-SSB and thermalization. For example, it may be possible to understand thermalization itself as (local) SW-SSB associated with the continuous time-translation symmetry -- the discrete time-translation version of this idea has already appeared in Ref.~\cite{GarrattChalker2021}. A natural order parameter in this spirit can be constructed as follows: for a given state $\rho$ (or the infinite-system version $\omega$), take a generic local operator $O_x$ and subtract off the expectation value $\hat{O}_x:=O_x-\langle O_x\rangle$, and normalize it so that $\norm{\hat{O}_i}_\infty=1$. If the state has strong time-translation symmetry, i.e. a fixed energy, and the spectrum is non-degenerate, then $\hat{O}_x\rho \hat{O}_x^{\dagger}$ lives in orthogonal energy sectors. This motivates the following local fidelity correlator detecting time-translation SW-SSB:
    \begin{equation}
        \lim_{|A|\to\infty}F(\rho_A, O_x-\langle O_x\rangle).
    \end{equation}
    Interestingly, as long as we choose sufficiently generic local operators, the order parameter is not sensitive to the exact form of Hamiltonian. This is hinting towards a definition of ``local thermalization''. It will be interesting to ask whether this form of local thermalization is related to other notions of thermalization in the literature.

    \item \textit{Connection to charge sharpening and scrambling:} It has been previously pointed that SW-SSB is related to the phenomenon of charge fuzziness found in monitored dynamics under symmetry constraints~\cite{agrawal_entanglement_2022, barratt_transitions_2022, singh_mixedstate_2026, vijay_holographically_2025}, and more recently to charge scrambling~\cite{lee_charge_2026}. The common ground between these concepts is the nonlocal spread of local charge degrees of freedom. Since our work defines SW-SSB via local density matrices, it appears to be a natural framework to discuss these relations. For instance, the order-disorder inequalities of Sec.~\ref{sec:order-disorder} imply a quantitative tradeoff between the LFC and how close $\rho_A$ can be to a strongly symmetric state, which is the extreme limit of a state with localized charges within $A$. These, however, are only bounds coming from measures of strong asymmetry, which is not the only information carried by the (local) SW-SSB order. For these reasons, we leave the investigation of more precise connections to future studies.

    \item \textit{Local fidelity as defect problem}: as we have demonstrated in Sec.~\ref{sec:examples}, in critical states, the calculation of the local fidelity or R\'enyi correlators becomes an interesting defect problem on its own. For CFT ground states, we showed (Sec.~\ref{sec:CFT}) that the R\'enyi correlators are given by the ordinary two-point function, up to a universal overall constant undetermined in dimension $d>1$. For more general types of critical states, the Fermi metal examples (Sec.~\ref{sec:metals}) clearly showed that the R\'enyi correlator scales differently from ordinary two-point function. It is an interesting problem to better understand local fidelity/R\'enyi correlators as defects at other types of criticality, such as Liftshitz-type of criticality (with dynamical exponent $z\neq1$), or other interacting critical states without translation symmetry. Another related question concerns the purification picture of the local fidelity correlator (Theorem~\ref{thm:pure_state_non_local_charge}): can we find some examples for which we know the optimized, possibly non-local, operator $O_{A^c}$ that gives the LFC?  Can we understand the difference in scaling between the fidelity and the two-point functions for the critical examples, which also imply a difference between MI and CMI scaling?
  
    \item \textit{Simplified numerical scaling}: In defining SW-SSB solely through the local density matrix, we allow the global state to be infinitely large. In cases where the local density matrix $\rho_A$ of an infinitely large system is accessible --- for example, in decoherence problems $\rho_A = \Tr_{A^c} \E[\rho_0]$ where the reduced noise channel $\Tr_{A^c} \circ \E$ has a simple description --- the local fidelity correlator $F(\rho_A; O_x)$ has the advantage of having only one relevant length scale, namely the distance $\ell$ between $x$ and the boundary $|\partial A|$ (assuming an isotropic region $A$). This makes the scaling ansatz of critical (mixed) states only depend on $\ell$ through some power law, whereas normally there would be an additional dependence on an universal function $f(\ell/L)$, where $L$ is the total system size. We expect this simplification to also make the estimation of critical data (e.g. transition points and critical exponents) more accurate given the same computational resources, but leave its practical verification through examples for the future.

 \end{itemize}

\textit{Note added}: for complementary parallel works on the local formulation of strong-to-weak symmetry breaking, see the works by Liu, Yi and Else in Ref.~\cite{LiuYiElse2026}, and by Zhang in Ref.~\cite{Zhang2026}. 

For complementary perspective on order-disorder relation in the context of ground states with quenched disorder, developed in parallel to ours by Yi and Wang, see Ref.~\cite{YiWang2026}.

\begin{acknowledgments}
    We acknowledge illuminating discussions with Dominic Else, Tarun Grover, Ruizhi Liu, Zack Weinstein, Cenke Xu, Jinmin Yi, Yizhi You and Carolyn Zhang. This research was supported in part by grant NSF PHY-2309135 to the Kavli Institute for Theoretical Physics (KITP). We acknowledge support from the Natural Sciences and Engineering Research Council of Canada (NSERC) under Discovery Grant No. RGPIN-2020-04688 (F.D. and L.A.L.), No. RGPIN-2018-04380 (L.A.L.) and No. RGPIN-2023-03467 (F.D.).  This work was also supported by an Ontario Early Researcher Award. Research at Perimeter Institute is supported in part by the Government of Canada through the Department of Innovation, Science and Industry Canada and by the Province of Ontario through the Ministry of Colleges and Universities.
\end{acknowledgments}

\appendix

\section{Unitary order parameters suffice}
\label{app:unitary_OP}

Here, we prove that it suffices to assume that the order parameter in the LFC is unitary whenever one such case exists, for the purpose of diagnosing local SW-SSB order. For that, we will need the following lemma:

\begin{lemma}\label{lemma:russo_dye_finite}
    For any subnormalized matrix $M$, $\norm{M}_\infty \leq 1$, there exist unitaries $U_1$ and $U_2$ such that
    \begin{equation}
        M = \frac{U_1 + U_2}{2}
    \end{equation}
\end{lemma}
\begin{proof}
    Let $M = U \Sigma V^\dagger$ be the singular value decomposition of $M$, where $U, V$ are unitary and $\Sigma = \text{diag}(\sigma_1, \dots, \sigma_n)$. Since $\|M\|_\infty \leq 1$, we have $0 \leq \sigma_j \leq 1$ for all $j$.

    Define diagonal matrices $D_1$ and $D_2$ with entries:
    \begin{equation}
        (D_{1,2})_{jj} = \sigma_j \pm i\sqrt{1 - \sigma_j^2} = e^{\pm i \arccos(\sigma_j)}
    \end{equation}
    Note that $|(D_{1,2})_{jj}| = 1$, so $D_1$ and $D_2$ are unitary. Their average satisfies:
    \begin{equation}
        \frac{D_1 + D_2}{2} = \Sigma
    \end{equation}
    Define $U_1 = U D_1 V^\dagger$ and $U_2 = U D_2 V^\dagger$. Being products of unitaries, $U_1$ and $U_2$ are unitary. Finally,
    \begin{equation}
        \frac{U_1 + U_2}{2} = U \left( \frac{D_1 + D_2}{2} \right) V^\dagger = U \Sigma V^\dagger = M \qedhere
    \end{equation}
\end{proof}

\begin{theorem}
    For any subnormalized charged operator $O$, $\norm{O}_\infty \leq 1$, there exists a unitary $V$ of equal support satisfying
    \begin{equation}
        F(\rho; V) \geq F(\rho; O).
    \end{equation}
    Moreover, if there exists a charged unitary with the same charge and smaller than or equal support than $O$, then $V$ can also be assumed to have the same charge and support as $O$.
\end{theorem}
\begin{proof}
    By Uhlmann's theorem, given a purification $\ket{\xi_\rho} \in \Hilb \otimes \Hilb'$, $\Hilb' \simeq \Hilb$, of $\rho \in \mathcal{L}(\Hilb)$, there exists an unitary $W \in \mathcal{U}(\Hilb')$ acting on the ancillary space $\Hilb'$ such that
    \begin{equation}
        F(\rho, O \rho O^\dagger) = |\langle \xi_\rho | O \otimes W | \xi_\rho \rangle|.
    \end{equation}
    By Lemma \ref{lemma:russo_dye_finite}, we can decompose $O$ as the average of unitaries $U_1$ and $U_2$ supported on the same region as $O$, so that
    \begin{align}
        & F(\rho, O \rho O^\dagger) \\
        & = \left| \frac{1}{2}\braket{\xi_\rho | U_1 \otimes W | \xi_\rho} + \frac{1}{2}\braket{\xi_\rho | U_2 \otimes W | \xi_\rho} \right| \\
        & \leq \max (|\langle \xi_\rho | U_1 \otimes W | \xi_\rho \rangle|, |\langle\xi_\rho | U_2 \otimes W | \xi_\rho \rangle|) \\
        & \leq \max_{W'\in \mathcal{U}(\Hilb')} \max (|\langle \xi_\rho | U_1 \otimes W' | \xi_\rho \rangle|, |\langle\xi_\rho | U_2 \otimes W' | \xi_\rho \rangle|) \\
        & = F(\rho, V \rho V^\dagger),
    \end{align}
    where in the last equation we chose $V$ to be the unitary that maximizes the right-hand side of the second line.

    For the second statement, we can assume the existence of a unitary $V'$ with the same charge and support as $O$. Then, the operator $\tilde{O} = V'^\dagger  O$ commutes with the symmetry action. By Schur's lemma, $\tilde{O}$ has a block-diagonal decomposition $\tilde{O} = \sum_{\mu} \tilde{O}_\mu$, where the sum is over all irreps $\mu$, and $\tilde{O}_\mu$ is supported on the isotypic space of $\mu$. We then apply Lemma \ref{lemma:russo_dye_finite} to each of the blocks, $\tilde{O}_\mu = (\tilde{U}_{1,\mu} + \tilde{U}_{2,\mu})/2$, thus arriving at $O = (U_1 + U_2)/2$, with $U_i = V'\sum_\mu \tilde{U}_{i,\mu}$. By construction, both $U_1$ and $U_2$ are charged operators with charge $\lambda$ and same support as $O$.
\end{proof}

This existence condition above depends solely on the irrep structure of the symmetry action. If all of the irrep sectors have the same dimension, then there will always be unitary charged operators with any charge. This is the case of the usual $\Z_n$ representation via $\prod_i X_i$. Namely, $Z_i^k$ is a single-site unitary that has all possible charges by varying $k \in \{0, \ldots, n-1\}$.

We note that this proof can be readily extended to the two-point fidelity, in the form of $F(\rho, V_x V_y^\dagger \rho V_y V_x^\dagger) \geq F(\rho, O_x O_y^\dagger \rho O_y O_x^\dagger)$, for $V_x$ and $V_y$ localized around sites $x$ and $y$, respectively. One just needs to apply Lemma \ref{lemma:russo_dye_finite} for each $O_x$ and $O_y$ separately.

\section{Connected fidelity is not positive semidefinite}\label{app:counter_Example_PSD_Fc}

The positive-semidefinite property in Sec.~\ref{sec:connected_correlator} is \textit{not} true for the fidelity measure. Indeed, a counterexample is $\rho = \frac{1}{2} \one$ and $O = \{\ketbra{0}{0}, \ketbra{+}{+}, \ketbra{1}{1}\}$. In this case, the connected fidelity matrix is
\begin{equation}
    F_c = \frac{1}{4}
    \begin{pmatrix}
        1 & \sqrt{2} -1 & -1 \\
        \sqrt{2} -1 & 1 & \sqrt{2} -1 \\
        -1 & \sqrt{2} -1 & 1
    \end{pmatrix},
\end{equation}
which is not a positive-semidefinite matrix. For example, for $v = (1, -\sqrt{2}, 1)^T$, we have $v^T F_c v = -(\sqrt{2}-1)^2/2 < 0$. For comparison, the R\'enyi-1 and R\'enyi-2 connected matrices for the same example are
\begin{equation}
    R^{(1)}_c = \frac{1}{4}
    \begin{pmatrix}
        1 & 0 & -1 \\
        0 & 1 & 0 \\
        -1 & 0 & 1
    \end{pmatrix}, \qquad
    R^{(2)}_c = \frac{1}{8}
    \begin{pmatrix}
        3 & 1 & -1 \\
        1 & 3 & 1 \\
        -1 & 1 & 3
    \end{pmatrix},
\end{equation}
which are both positive-semidefinite. 

Even if we ask for the operators $O_x$ to pairwise commute, there are essentially classical counterexamples: By taking a pure state $\rho = \ketbra{\psi}{\psi}$ and diagonal operators $[O_x]_{ij}= (\Omega_x)_i \delta_{ij}$ for which $\braket{\psi | O_x |\psi} = 0$, then $[F_c]_{x,y} = |\text{Cov}(\Omega_x, \Omega_y)|$ with respect to the probability distribution defined by $|\psi_i|^ 2$. Even though the covariance matrix of the random variables $\Omega_x$ is positive semidefinite, taking their entry-wise absolute value may not be. Indeed, it is easy to check that this happens for $\Omega = \{(1,1-1,-1),(\sqrt{2},0,0,-\sqrt{2}),(1,-1,1,-1),(0, -\sqrt{2}, \sqrt{2}, 0)\}$.

\section{Local SW-SSB of a finite-temperature paramagnet}
\label{app:paramagnet}

As an illustrative example, we consider the strongly symmetric density matrix corresponding to a thermal paramagnet, with $\Z_2$ symmetry generated by $U = \prod_{i} X_i$, and the paramagnetic fixed-point Hamiltonian $H=-\sum_i X_i$:
\begin{equation}
\rho_+ = \frac{1}{Z_\beta}\frac{1+U}{2} e^{-\beta H}.
\end{equation}
The partition function is $Z_\beta = \sum_{b_i =\pm1}' e^{-\sum_j b_j}$, where the sum is restricted to configurations satisfying $\prod_i b_i = 1$. It was shown in Ref.~\cite{lessa2025strong} that this state exhibits SW-SSB, with two-point fidelity given in the thermodynamic limit by
\begin{equation}\label{eq:two_point_fidelity_paramagnet}
\lim_{|i-j|\to \infty}F(\rho; Z_iZ_j) = \frac{1}{\cosh^2\beta}.
\end{equation}
As a result, the one-point fidelity is finite at any finite temperature, $\beta < \infty$. It is instructive to see how this arises. 

Although $\rho_+$ is strongly symmetric, the reduced density matrix on $A$ is not, allowing for a non-vanishing one-point fidelity. In fact,
\begin{equation}\label{eq:paramagnet_reduced_density_matrix}
\begin{aligned}
\rho_A &= \frac{1+\epsilon}{2}\rho_+ + \frac{1-\epsilon}{2}\rho_-,
\qquad
\rho_{\pm} = \frac{1}{Z_{\pm}(A)} P_{\pm} e^{-\beta H_A} P_{\pm},
\end{aligned}
\end{equation}
where $H_A = \sum_{i \in A} B_i$ and $U_A = \prod_{i \in A} X_i$ are the restrictions of the onsite Hamiltonian and symmetry to $A$, $P_\pm = \frac{1 \pm U_A}{2}$, and we introduced the parameter $\epsilon = \frac{ \tanh^{|\bar A|}\beta + \tanh^{|A|}\beta}{\tanh^N\beta + 1 }$.

The one-point fidelity is then
\begin{equation}\label{eq:one_point_fidelity_paramagnet}
\begin{aligned}
F(\rho_A; Z_i) &= \frac{1 - \epsilon^2}{2} F(\rho_+, Z_i \rho_- Z_i)\\
&= \frac{\sqrt{1 - t^{2|\bar A|}}}{1 + t^{N}} \frac{1}{\cosh\beta}\\
\end{aligned}
\end{equation}
This shows that the system exhibits SW-SSB for any finite temperature, $\beta \neq \infty$. Note that the thermodynamic and zero-temperature limits do not commute. At finite volume, there is a smooth crossover at temperature $\beta^{-1} \sim 1/N$ from the phase with unbroken strong symmetry to the SW-SSB phase. Once the thermodynamic limit $N \to \infty$ is taken, however, the zero-temperature limit ($\beta \to \infty$) and the large-region limit ($|A| \to \infty$) do commute.

\bibliography{References}
\end{document}